%% file: main.tex
\documentclass[manuscript,anonymous=false,natbib=false,screen,authorversion]{acmart}

\usepackage{multirow}
\usepackage{float}
\usepackage{multicol}
\usepackage{soul}

\usepackage{cancel} 
\usepackage{nicefrac}
\usepackage{makecell}

\graphicspath{{img}}

\makeatletter
\newcommand\footnoteref[1]{\protected@xdef\@thefnmark{\ref{#1}}\@footnotemark}

\AtBeginDocument{%
  }

\setcopyright{acmcopyright}
\copyrightyear{2023}
\acmYear{2023}
\acmDOI{XXXXXXX.XXXXXXX}

\RequirePackage[
  datamodel=acmdatamodel,
  style=acmnumeric,
  sortcites=true
  ]{biblatex}

\begin{document}

\title{Evaluation Measures of Individual Item Fairness for Recommender Systems: A Critical Study}


\author{Theresia Veronika Rampisela}
\email{thra@di.ku.dk}
\orcid{0000-0003-1233-7690}
\affiliation{%
  \institution{University of Copenhagen}
  \streetaddress{Universitetsparken 1, 2100}
  \city{Copenhagen}
  \country{Denmark}}

\author{Maria Maistro}
\email{mm@di.ku.dk}
\orcid{0000-0002-7001-4817}
\affiliation{%
  \institution{University of Copenhagen}
  \city{Copenhagen}
  \country{Denmark}
}

\author{Tuukka Ruotsalo}
\email{tr@di.ku.dk}
\orcid{0000-0002-2203-4928}
\affiliation{%
  \institution{University of Copenhagen}
  \city{Copenhagen}
  \country{Denmark}
}
\affiliation{%
  \institution{LUT University}
  \country{Finland}
}

\author{Christina Lioma}
\email{c.lioma@di.ku.dk}
\orcid{0000-0003-2600-2701}
\affiliation{%
  \institution{University of Copenhagen}
  \city{Copenhagen}
  \country{Denmark}
}

\renewcommand{\shortauthors}{Rampisela et al.}

\newcommand{\rs}{RSs}

\newcommand{\floorkm}{\left\lfloor\frac{km}{n}\right\rfloor}

\newcommand{\our}{$_{\text{our}}$}
\newcommand{\ori}{$_{\text{ori}}$}
\newcommand{\up}{$\uparrow$}
\newcommand{\down}{$\downarrow$}

\newcommand{\ups}{\texorpdfstring{\up}{}} 
\newcommand{\dws}{\texorpdfstring{\down}{}}

\newcommand{\explainsig}{Asterisk ($^*$) denotes a statistically significant correlation ($\alpha=0.05$), after applying the Benjamini-Hochberg procedure.}

\newcommand{\nonreal}{\textbf{non-re\-a\-li\-sa\-bi\-li\-ty}}
\newcommand{\nofix}{{\large $\circ$}}
\newcommand{\bbullet}{{\large $\bullet$}} 

\begin{abstract}
Fairness is an emerging and challenging topic in recommender systems. In recent years, various ways of evaluating and therefore improving fairness have emerged. In this study, we examine existing evaluation measures of fairness in recommender systems. 
Specifically, we focus solely on exposure-based fairness measures of individual items that aim to quantify the disparity in how individual items are recommended to users, separate from item relevance to users. We gather all such measures and we critically analyse their theoretical properties. We identify a series of limitations in each of them, which collectively may render the affected measures hard or impossible to interpret, to compute, or to use for comparing recommendations. 
We resolve these limitations by redefining or correcting the affected measures, or we argue why certain limitations cannot be resolved. We further perform a comprehensive empirical analysis of both the original and our corrected versions of these fairness measures, using real-world and synthetic datasets. Our analysis provides novel insights into the relationship between measures based on different fairness concepts, and different levels of measure sensitivity and strictness. We conclude with practical suggestions of which fairness measures should be used and when. Our code is publicly available. 
To our knowledge, this is the first critical comparison of individual item fairness measures in recommender systems.
\end{abstract}



\maketitle

\section{Introduction}\label{s:intro}

The concept of fairness in Recommender Systems (RSs) is commonly understood as treating users or items that are alike, in a similar way. It can be studied either for individual items or users (individual fairness), or for groups of items or users (group fairness). Individual fairness typically refers to similar individuals being treated similarly \cite{Biega2018EquityRankings}, where sometimes the similarity of individuals is measured through a certain metric, e.g. distance of the individual representation \cite{LiYunqi2023FairnessApplications}. 
In this work, we focus solely on fairness for individual items, and specifically on evaluation measures designed to quantify individual item fairness in RSs. While there exist comprehensive surveys on group fairness measures \cite{Raj2022MeasuringResults,Zehlike2022FairnessSystems}, to the best of our knowledge, no critical analysis of evaluation measures of individual item fairness for \rs~ has been presented. 

Individual item fairness is an important type of fairness that occurs in many RS scenarios. For example, it is helpful for new item discovery and for ensuring that recommended items come from different providers or creators. Individual item unfairness may occur due to popularity bias, which causes some items to be recommended more often than others \cite{Zhu2021Popularity-OpportunityFiltering}. Some items may not even be recommended at all. For instance, interesting content from emerging content creators could be recommended less frequently than less or equally interesting content from popular creators. 

In a traditional recommendation scenario, a model produces a top-$k$ list of recommendations across all users. This output is typically evaluated by measuring the recommendation relevance of the top-$k$ recommendations to each user. On top of that, one can measure the individual item fairness of the top-$k$ recommendations. 
On a high level, we distinguish between two definitions of individual item fairness. In the \textbf{first definition}, individual item fairness is understood as \textit{all items}\footnote{All items in the dataset or all items in the recommendations given to all users. Both definitions are used in the literature and also in our work in $\S$\ref{ss:onlyfair}.}
\textit{having equal exposure}. 
Exposure (also known as attention \cite{Biega2018EquityRankings} or coverage \cite{Wang2022ProvidingSystems}) refers to an item appearing in the top $k$ recommendations for a user. The definition of fairness above does not consider how relevant the recommended items are to the users; it only considers how uniform their exposure is. 
Given that relevance is a key aim in \rs, fairness has also been given a \textbf{second definition} as \textit{all items having an equal opportunity for exposure, where the opportunity is based on item relevance to users or similar other criteria \cite{Wu2022JointRecommendation,Diaz2020EvaluatingExposure}.}

Several evaluation measures have been used to empirically evaluate fairness for individual item fairness in \rs~\cite{Wang2022}. 
All of the measures evaluate fairness based on the recommendation list; the input is the recommendation list, and in most cases, the measure compares the exposure of items in the recommendation list and the total number of items in the dataset. Some of the measures follow the first definition of fairness, which is purely exposure-based, while other measures evaluate fairness jointly with relevance. 
Here, we focus solely on fairness measures of the first type, that is measures that evaluate only fairness without considering relevance. Even if these fairness-only measures have been used to evaluate fairness in prior work, they have not been extensively analysed for their limitations and possible consequences arising thereof. They have also not been analysed in relation to one another, as most research in individual item fairness only uses a single fairness measure. 
As a result, it is currently unclear how to interpret these measures and select which measures should be used under various circumstances. For instance, we find cases where a measure is theoretically bound between $[0,1]$ but empirically cannot reach either of the endpoints, which makes the interpretation of the measure extremely difficult. 
It is also unknown if there are theoretical limitations and cases where the measures would fail. For instance, we find that 
a measure will give the highest scores no matter the recommendation, given a certain number of users, items, or cut-off. 
We also find that a measure cannot be computed as the formulation is not well-defined for a commonly-encountered case.

The goals of this work are to address the above research gaps, by examining existing measures of individual item fairness in \rs, presenting analytical limitations and solving them (or justifying why they cannot be solved), and examining how the scores of the measures change in relation to each other in different scenarios. As such, 
we contribute the following: 
\begin{itemize}
    \item We review individual item fairness measures in \rs~($\S$\ref{s:priorwork}).
    \item We identify and analyse 5 theoretical limitations of those measures, where three of the limitations are novel 
    and two limitations have already been identified but without a formal/theoretical explanation, which we provide ($\S$\ref{s:limitations}). 
    \item We propose theoretical corrections of the measures to resolve their limitations or explain why some limitations are unresolvable ($\S$\ref{s:extensions}). Note that none of the measures is without limitations (Tab.~\ref{tab:limitation-summary}).
    \item  We present an empirical analysis of the existing measures to study the correlation among measures, investigate the relations between fairness and relevance, and compare the corrected measures against the original measures on real-world and synthetic datasets ($\S$\ref{s:exp}). 
    \item We provide insights to guide the correct use of different types of individual item fairness measures ($\S$\ref{s:discussion}).
\end{itemize}

\section{Individual item fairness measures}
\label{s:priorwork}

We present our notation ($\S$\ref{ss:notation}) and the  
eight exposure-based evaluation measures of individual item fairness that we study in this work ($\S$\ref{ss:onlyfair}). 
To the best of our knowledge, these eight measures are all the measures of individual item fairness that exist in RSs.\footnote{Based on publications up to August 2022.} Several of these measures are taken from \cite{Wang2022}. 

\subsection{Notation and Examination Functions}
\label{ss:notation}

\begin{table}[tb]
\caption{Summary of the notation.}
\label{tab:notation}
\begin{tabular}{ll}\toprule
Notation &Explanation \\\midrule
$U=\{u_1, u_2, \dots, u_m\}$ &The set of users \\
$I = \{i_1, i_2, \dots, i_n\}$ &The set of items \\
$|U| = m$ &The unique number of users in dataset \\
$|I| = n$ &The unique number of items in dataset\\
$k$ & The cut-off threshold\\
$km$ & The total number of recommendation slots\\
$r_{u,i} \in \{0,1\}$ & The relevance of item $i$ to user $u$ \\
$z(u,i) \in \{1, 2, \dots, n\} $ & The rank position of item $i$ for user $u$\\
$z(u,i,w) \in \{1, 2, \dots, n\} $ & The rank position of item $i$ for user $u$ in round $w$\\
$R_{u}^{k}$ & The top $k$ recommendations for user $u$ \\
$R_{u,w}^{k}$ & The top $k$ recommendations for user $u$ in round $w$ \\
$R$ & The set of top $k$ unique items recommended to all users\\
$1_\mathcal{A}(x)=1$ if $x \in \mathcal{A}$, else 0 & Indicator function\\
\bottomrule
\end{tabular}
\end{table}

\begin{table*}
\caption{Examination functions used in fairness measures.}
\label{tab:exp-weigh}
\begin{tabular}{llll}
\toprule
         & Equation & Measure & Reference                        \\ 
         \midrule
uniform  & $e_{\text{unif}}(u,i)\equiv 1, \forall (u,i)$                         & Jain, QF, Ent, Gini, FSat, VoCD & \cite{jain1984quantitative,Zhu2020FARM:APPs,Shannon1948ACommunication, Gini1912VariabilitaMutabilita, Wang2022ProvidingSystems,Patro2020FairRec:Platforms} \\

DCG      & $e_{\text{DCG}}(u, i, w) = 1/\log_2 (z(u,i,w)+1)$         & Gini-w & \cite{Do2021Two-sidedDominance}   \\
RBP      & $e_{\text{RBP}}(u,i,w) = \gamma^{z(u,i,w)-1}$         & II-D, AI-D   & \cite{Wu2022JointRecommendation}                 \\
\bottomrule
\end{tabular}
\end{table*}

Given a set of users $U=\{u_1, u_2, \dots, u_m\}$, with $|U| = m$,  and 
a set of items $I = \{i_1, i_2, \dots, i_n\}$, with $|I| = n$, for each user $u \in U$ we rank all $n$ items to produce the full recommendation list. 
The list of the top $k$ recommended items to user $u$ is $R_u^{k}$, 
and the set of all top $k$ recommended items to all users is $R = \bigcup_{u \in U}{R_u^{k}}$.
If $u$ finds the recommended item $i$ relevant, we write $r_{u,i}=1$, otherwise $r_{u,i}=0$. The rank position of item $i$ in the recommendation list for user $u$ is $z(u,i)$. 
As there exist fairness measures that consider multiple rounds of recommendation, we also use the following notations for such measures: the rank position of item $i$ for user $u$ in round $w$ is $z(u,i,w)$ and $R_{u, w}^{k}$ is the list of the top $k$ recommended items to user $u$ in round $w$.
Tab.~\ref{tab:notation} summarizes our notation.

Fairness for individual items is closely linked to exposure, which is identified as the appearance of an item in the top $k$ recommendations for a user. Exposure can be quantified in several ways using different examination functions $e(\cdot)$ in various binary or graded ways. 
Examination functions are functions modelling the probability of a user seeing an item that is exposed to the user. All examination functions in this paper assume that this probability only depends on $z(u,i)$, the rank position of an item $i$ for user $u$. Tab.~\ref{tab:exp-weigh} presents the examination functions used by the individual item fairness measures that are included in this paper. These examination functions apply either no discount, logarithmic discount, or exponential-like discounting. 

The simplest examination function is uniform, $e_{\text{unif}}$, and assumes a constant weight of $1$ for each rank position \cite{jain1984quantitative,Zhu2020FARM:APPs,Shannon1948ACommunication, Gini1912VariabilitaMutabilita, Wang2022ProvidingSystems,Patro2020FairRec:Platforms}.
The two other examination functions are $e_{\text{DCG}}$ \cite{Do2021Two-sidedDominance} 
and $e_{\text{RBP}}$ \cite{Wu2022JointRecommendation}, which 
use the discount function based on Discounted Cumulative Gain (DCG)~\cite{Jarvelin2002CumulatedTechniques} and Rank-Biased Precision (RBP)~\cite{Moffat2008Rank-biasedEffectiveness} respectively.
In $e_{\text{RBP}}$, the parameter $\gamma$ is the user's patience, i.e., the probability of the user examining the next ranked item. For example, $\gamma=0.8$ in \cite{Wu2022JointRecommendation} and $\gamma=0.5$ in \cite{Diaz2020EvaluatingExposure}.

\subsection{Measures of individual item fairness}
\label{ss:onlyfair}

We present the eight measures that so far have been used to quantify fairness for individual items (without considering item relevance), as well as the context in which they are used in the original work. We use the subscript $\cdot_{\text{ori}}$ to denote the original formulation of a measure as opposed to our corrected version that we present later in $\S$\ref{s:extensions}. Note that these measures are typically used with a fixed cut-off $k$, and therefore the number of recommendation slots $km$ is also fixed. 

\subsubsection[Jain's Index (Jain)]{Jain's Index (Jain) \cite{jain1984quantitative}}\label{sss:jain-ori} Jain, which was originally defined for fairness in computer networks, 
has been used in \rs~\cite{Zhu2020FARM:APPs} to measure how consistent item exposure is in relation to the number of times an item is recommended. The original work uses this measure to evaluate ``fairness of the recommendation opportunity''. Jain is the ratio between the square of the number of recommendation slots and the sum of squares of the number of times each item is recommended, where the ratio is divided by the number of items in the dataset. It is calculated as follows: 
\begin{equation}
\label{eq:jain-ori}
\text{Jain}_{\text{ori}} = 
\frac{
    \left[
        \sum\limits_{i\in I} \sum\limits_{u\in U}1_{R_{u}^{k}}(i)
    \right]^2
    }
{n \sum\limits_{i\in I} \left[\sum\limits_{u\in U}
1_{R_{u}^{k}}(i)\right]^2}
= \frac{(km)^2}{n \sum\limits_{i\in I} \left[\sum\limits_{u\in U}
1_{R_{u}^{k}}(i)\right]^2}
\end{equation}
where 
$1_{R_{u}^{k}}(i)=1$ if item $i$ is in the top $k$ recommendations for user $u$, and 0 otherwise. 
$\sum\nolimits_{u\in U}
1_{R_{u}^{k}}(i)$ counts how many times item $i$ is recommended in the top $k$ across all users. The range of Eq.~\eqref{eq:jain-ori} is $[0,1]$.\footnote{This range is based on the original paper, \cite{jain1984quantitative}} The range of Jain matters as \cite{Zhu2020FARM:APPs} analysed the absolute values of Jain, in addition to observing the difference of the scores. 
The higher the Jain score, the fairer the recommendation with respect to individual items (i.e., items are exposed consistently with respect to other items in the dataset). 
E.g., if $60\%$ of the items in the dataset are exposed equally to all users, for instance, $R_{u_1}^{3}=[i_1,i_2,i_3],\ R_{u_2}^{3}=[i_4,i_5,i_6],\ n=10$, then $\text{Jain}=0.6$. However, this interpretation does not hold and becomes less intuitive when items are not exposed equally, which is often the case. For instance, $R_{u_1}^{3}=[i_1,i_2,i_3],\ R_{u_2}^{3}=[i_1,i_2,i_4],\ R_{u_3}^{3}=[i_1,i_5,i_6],\ n=10$. In this case, $60\%$ of the items are exposed but $\text{Jain}=0.476$. In real-life, it is unlikely that items are exposed equally, which means that Jain's interpretation suffers from this limitation, more often than not. 

\subsubsection{Qualification Fairness (QF) \cite{Zhu2020FARM:APPs}} \label{sss:qf-ori}
QF is a modification of Jain that measures how many items are in the set of top $k$ recommended items $R$, divided by $n$, the total number of items in the dataset. 
The authors of the original measure explained in \cite{Zhu2020FARM:APPs} that the measure only considers whether an item in the dataset is recommended, as opposed to how many times it is recommended. 

\begin{equation}
\label{eq:qf-ori}
    \text{QF$_{\text{ori}}$} =  
    \frac{
        \left[
            \sum\limits_{i\in I} 1_{R}(i)
        \right]^2
        }
    {n \sum\limits_{i\in I}
    \left[1_{R}(i)\right]^2} 
    = \frac{
        \left[
            \sum\limits_{i\in I} 1_{R}(i)
        \right]^2
        }
    {n \sum\limits_{i\in I}
    1_{R}(i)}                                   
    = \frac{\sum\limits_{i\in I} 1_{R}(i)}{n} 
    = \frac{|R|}{n}                                 
\end{equation}

\noindent The QF range is $[0,1]$. The higher the score, the fairer the recommendation. 
Along with the relative comparison of the QF scores, the absolute values of QF scores are also taken into account in \cite{Zhu2020FARM:APPs,Mansoury2020FairMatch:Systems}, where 
a score of $1$ means that all items in the dataset are in the top $k$ at least once. 
Formally, QF (Eq.~\ref{eq:qf-ori}) is equivalent to Coverage \cite{HerlockerEvaluatingSystems}, a measure of diversity, which has been used to evaluate fairness too \cite{Mansoury2020FairMatch:Systems}.

\subsubsection{Entropy (Ent) \cite{Shannon1948ACommunication}} Ent measures how uniform the exposure of the recommended items is  \cite{Patro2020FairRec:Platforms,Mansoury2021ASystems,Mansoury2020FairMatch:Systems}. 
In \cite{Patro2020FairRec:Platforms}, Lorenz curves were used to detect massive differences in individual item exposures and therefore, Entropy-like measure was proposed to quantify the inequality of item exposure: 

\begin{equation}
   \text{Ent$_{\text{ori}}$} = 
        - \sum\limits_{i \in I}{p(i) \log{p(i)}} \qquad \text{and}\qquad  p(i) = \frac{\sum\limits_{u\in U}1_{R_{u}^{k}}(i)}{km}
        \label{eq:ent-ori}
\end{equation}
where $p(i)$ is the recommendation frequency of $i$, i.e., how often item $i$ is recommended in the top $k$ to any user in the dataset, divided by the available recommendation slots $km$. In~\cite{Patro2020FairRec:Platforms}, $\log$ is the log base-$n$. 
It is unclear what log base is used in \cite{Mansoury2021ASystems,Mansoury2020FairMatch:Systems}. When the log base is $b$, Ent ranges between $[0, \log_b{n}]$ while for log base-$n$, the range is $[0, 1]$. 

In the past, this measure has been used by \cite{Patro2020FairRec:Platforms, Mansoury2020FairMatch:Systems, Mansoury2021ASystems} to compare recommender models using absolute values, where a higher Ent is interpreted as having a more uniform distribution of the recommended items, and thus fairer.

\subsubsection{Gini Index (Gini) \cite{Gini1912VariabilitaMutabilita}} \label{sss:gini-ori} 
Gini is a measure of variability i.e., the mean difference from all observed quantities \cite{Ceriani2012TheGini}. It is most commonly used to measure inequality in the distribution of economic income, where the intuition is that a Gini score of 1 means that one entity receives all the income. 
Similarly in \rs, it is used to measure how much the distribution of item exposure deviates from an equal/uniform distribution \cite{Mansoury2020FairMatch:Systems, Do2022OptimizingRankings, Do2021Two-sidedDominance}. Formally, Gini is defined as follows: 

\begin{equation}
   \text{Gini$_\text{ori}$} = \frac{\sum\limits_{j=1}^n{(2j-n-1) Ex_j}}{n\sum\limits_{j=1}^n Ex_j} 
   \qquad \text{and}\qquad  
Ex_j = \sum\limits_{u\in U} \sum\limits_{w=1}^W 1_{R_{u,w}^{k}}(x_j)\cdot e_{(\cdot)}(u,x_j,w)
        \label{eq:gini-ori}
\end{equation}
where $x_j$ is the item with the $j$-th least amount of $Ex_j$, the total exposure received by that item across $W$ rounds of recommendations.\footnote{This is the total exposure received by each item, including items that are not recommended to any users. The scores are sorted as $Ex_1, Ex_2, \dots Ex_n$. Ties between $Ex_j$'s do not affect the final score.} $W$ is the number of rounds, and 
$1_{R_{u,w}^{k}}(i)=1$ when item $i$ is in user $u$'s top $k$ recommendation list in round $w$. 
The examination function $e_{(\cdot)}(u,x_j,w)$ used in \cite{Mansoury2020FairMatch:Systems} is uniform, $e_{\text{unif}}(u,x_j,w)\equiv 1$ and in \cite{Do2022OptimizingRankings, Do2021Two-sidedDominance} the examination function $e_{\text{DCG}}$  
(see Tab.~\ref{tab:exp-weigh}) is used.
We refer to the latter case as Gini-w. 
Gini and Gini-w's range is $[0,1]$, where 0 means that there is an equal distribution of item exposure for all items in the dataset (fairest case). 
The range is important for the interpretability of the measure. Having an interpretable range is more important for how Gini and Gini-w quantify fairness when looking at the absolute values of the measure, like in \cite{Mansoury2020FairMatch:Systems, Do2021Two-sidedDominance}, than when looking at the difference in model rankings, like in \cite{Do2022OptimizingRankings}.

\subsubsection{Fraction of Satisfied Items (FSat) \cite{Patro2020FairRec:Platforms}}\label{sss:fsat}FSat is defined in the context of maximin-shared fairness, where fairness means that each item is 
recommended at least $\frac{km}{n}$ times, as there are only $km$ slots that should ideally be distributed equally between $n$ items. However, this distribution is impossible if 
the number of slots is not divisible by the number of items in the dataset, $n \nmid km$. The requirement is relaxed to recommending each item at least $\left\lfloor \frac{km}{n}\right\rfloor$ times (the maximin share) for it to be a fair recommendation. An item is satisfied iff its exposure is more than or equal to the maximin share. $\text{FSat}$ measures the number of satisfied items divided by the total number of items: 
\begin{equation}
\label{eq:fsat-ori}
    \text{FSat}_\text{ori}=
    \frac{1}{n} 
         \sum\limits_{i \in I} 
        \delta
            \left(
            \sum\limits_{u\in U} 1_{R_u^{k}}(i) \geq \left\lfloor \frac{km}{n}\right\rfloor
            \right)
\end{equation}

\noindent where $\delta(\cdot)=1$ when the expression $\cdot$ is True and 0 otherwise. 
FSat has a range of $[0,1]$, and the higher, the fairer. The range of values matters, for example \cite{Patro2020FairRec:Platforms} has used both absolute values and difference in values to interpret FSat. 

\subsubsection{Violation of Coverage Disparity (VoCD) \cite{Wang2022ProvidingSystems}}\label{sss:vocd-ori} 
VoCD is a fairness constraint. In \cite{Wang2022ProvidingSystems}, VoCD 
is used to optimise the recommendations for fairness during the training process, but not used for evaluating the final recommendation output. We include VoCD in this work to provide insights into what transpires when VoCD is used as an evaluation measure in its current formulation. VoCD is also the only measure operating on the Lipschitz condition~\cite{Dwork2012FairnessAwareness}, which requires similar individuals to be treated similarly. 
The idea behind VoCD is that any two $\alpha$-similar recommended items should receive similar coverage. 
Two distinct items $i,i' \in R$ are $\alpha$-similar if $d(i,i')=1-sim(i,i')\leq \alpha$, 
where $d(i,i')$ is the cosine distance and $sim(i,i')$ is the cosine similarity between the embeddings\footnote{The original paper specifically defines that they use embeddings, but in principle, any representation could work.} of item $i$ and $i'$ respectively, and $\alpha$ is a parameter. 
Similar coverage means that the Coverage Disparity (CD) of those items, which is proportional to their exposure difference, must not exceed a threshold $\beta$. 
VoCD thus measures the average violation of the maximum allowed coverage disparity $\beta$ in all pairs of $\alpha$-similar recommended items:
\begin{equation}
\label{eq:vocd-ori}
    \text{VoCD$_\text{ori}$}
    = \frac{1}{|A|} \sum\limits_{\forall (i, i') \in A}{\max{(CD(i,i')-\beta,0})}
    \qquad \text{and} \qquad
    CD(i,i') = 
    \left| 
    \frac{\sum\limits_{u\in U}1_{R_{u}^{k}}(i) - \sum\limits_{u\in U}1_{R_{u}^{k}}(i')
    }{
    \max\left(
        \sum\limits_{u\in U}1_{R_{u}^{k}}(i), \sum\limits_{u\in U}1_{R_{u}^{k}}(i')
        \right)} 
    \right|
\end{equation}

\noindent where $A$ is the set of all pairs of $\alpha$-similar items and $CD(i,i')$ is the coverage disparity between item $i$ and $i'$. $\text{VoCD}=0$ means that there is no violation of coverage disparity between any pairs (i.e., fair). 
In \cite{Wang2022ProvidingSystems}, the absolute value of VoCD affects the parameter that is used to control fairness during the training process. 
VoCD is customisable  w.r.t.~fairness and similarity: a lower\footnote{This is \textit{lower} instead of \textit{higher} as the original authors \cite{Wang2022ProvidingSystems} defined the term $\alpha$-similar items based on the cosine distance between the two items being not more than $\alpha$} 
$\alpha$ means a stricter similarity requirement and a lower $\beta$ means a stricter fairness requirement. 

\subsubsection{Individual-user-to-individual-item disparity (II-D) \cite{Wu2022JointRecommendation}} \label{sss:iid-ori} 
II-D was first defined by \cite{Diaz2020EvaluatingExposure} to quantify the mean squared difference between system exposure and random exposure in individual queries and individual items, 
where there is a distribution of rankings (stochastic rankings). II-D is a resulting component of decomposing another measure in the original work, which quantifies item fairness proportional to item relevance to users. It was redefined by \cite{Wu2022JointRecommendation} for $W$ rounds of recommendations in RSs as: 

\begin{equation}
    \label{eq:iid-ori}
     \text{II-D}_\text{ori} = \frac{1}{m} \frac{1}{n} 
     \sum\limits_{u \in U} \sum\limits_{i \in I} \left(E_{u,i}^{ } - E_{u,i}^{\sim}\right)^2
\end{equation}

\begin{equation}
    \label{eq:eui}
E_{u,i} = \frac{1}{W} \sum\limits_{w=1}^{W}1_{R_{u,w}^{k}}(i) \cdot e_{\text{RBP}}(u,i,w)
\qquad \text{and} \qquad
E_{u,i}^{\sim} = \frac{1-\gamma^{k}}{n(1-\gamma)} 
\end{equation}

\noindent where $E_{u,i}$ is the expected exposure of $i$ to $u$ as per a stochastic ranking policy,  
$E_{u,i}^{\sim}$ is the expected exposure of $i$ to $u$ based on a uniformly random distribution over all permutations of items, and $\gamma$ is an arbitrarily-set parameter for user patience. 
The examination function based on RBP (see Tab.~\ref{tab:exp-weigh}) is used in $E_{u,i}$ and the equation of $E_{u,i}^{\sim}$ is derived based on the same examination function \cite{Wu2022JointRecommendation}. The range of II-D is not well-known,\footnote{The authors of the original measure did not state the range of II-D.} but a lower value means a fairer recommendation. In \cite{Wu2022JointRecommendation}, min-max normalisation is performed on II-D post-computation, such that the range is $[0,1]$. 
This range of values matters; for example \cite{Wu2022JointRecommendation} has used the absolute values of II-D to analyse fairness and relevance trade-off, in addition to looking at the difference in model rankings based on II-D scores. 

\subsubsection{All-users-to-individual-item disparity (AI-D) \cite{Wu2022JointRecommendation}} \label{sss:aid-ori} 
AI-D computes the mean of the squared difference between system exposure and random exposure in each item. AI-D is similar to II-D in the sense that it is originally used for multiple rounds of recommendations and also a component resulting from the decomposition of another measure proposed by \cite{Wu2022JointRecommendation} that considers fairness w.r.t.~relevance. 
However, unlike II-D, AI-D is sensitive to whether an item is recommended to multiple users due to the aggregation of the difference in exposure being done per item.

\begin{equation}
    \label{eq:aid-ori}
     \text{AI-D}_\text{ori} = 
        \frac{1}{n} \sum \limits_{i \in I}
            \left(
                \frac{1}{m}\sum \limits_{u \in U} E_{u,i}^{ }
                - \frac{1}{m}\sum \limits_{u \in U} E_{u,i}^{\sim}
            \right)^2
\end{equation}
where $E_{u,i}^{ }$, $E_{u,i}^{\sim}$ are as per Eq.~\eqref{eq:eui}. The range of AI-D is not well-known, but a lower value means fairer recommendation. 
Like on II-D, a post-computation min-max normalisation is also performed by \cite{Wu2022JointRecommendation} on AI-D, resulting in a $[0,1]$-range. 
The range of values matters; for instance \cite{Wu2022JointRecommendation} has used the absolute values of AI-D to analyse fairness and relevance trade-off, on top of looking at the difference in model rankings based on AI-D scores. 

\section{Measure limitations}
\label{s:limitations}
We identify $5$ theoretical limitations in the measures presented in $\S$\ref{s:priorwork} (summarised in Tab. \ref{tab:limitation-summary}). 
We use the term `limitation' in the sense that regardless of the reason, a measure fails to quantify or fulfill properties that are important for evaluating fairness. 
Some of these limitations rarely occur, e.g. related to edge cases ($\S$\ref{ss:always_fair}), yet some are more likely to occur in practical scenarios ($\S$\ref{ss:nonreal} \& $\S$\ref{ss:undefinedness}). In the headings, we put the name of the affected measures in brackets. 

In practice, even if the limitation transpires by design, the design of the measure still restricts its usage under the conditions that we explain below. The identified limitations are independent of the recommender algorithm, as long as the recommender is a top-$k$ recommender, which is the most common recommendation scenario in practice. 
We accompany each measure name by  \up~or \down, denoting that the higher (\up) or the lower (\down) the score of the measure, the fairer the recommendation. 

\subsection{Limitation 1: Non-realisability}\label{ss:nonreal}  

This is a novel limitation identified by us and affects \emph{all measures}. We define non-realisability as the limitation whereby the max/min score of the evaluation measure cannot be reached at the top-$k$. 
As argued in \cite{Moffat2013SevenMetrics}, a desirable property of effectiveness measures is their realisability. While the realisability property in \cite{Moffat2013SevenMetrics} is related to the number of relevant items, non-realisability for fairness measures is related to the number of recommendation slots ($km$) and the number of items that are in the dataset ($n$); we explain this relationship below as part of the causes of this limitation. 
In practice, the non-realisability limitation makes fairness scores hard to interpret because 
when the worst or best possible fairness score varies based on the dataset ($m$, $n$) and experimental choice of threshold $k$, it is unknown whether the fairness score obtained for a model is closer to the max or min score. 
For instance, if a higher-is-fairer measure ranges in $[0, 1]$ and a model achieves a score of $0.2$, one might think that the model is not very fair. However, if the maximum achievable score for that case (e.g., if all top $k$ items are fair) is $0.22$, then the model might actually be fair, but this cannot be known 
from the score of the evaluation measure.
We identify four different causes of non-realisability.

\subsubsection{Cause 1 (\up Jain, \up QF, \up Ent, \up FSat, \down Gini, \down Gini-w)} 
\label{sss:cause-no-recommend}
Non-realisability can occur if the most unfair score is only given to an unrealistic recommendation scenario, which we explain next as it differs per measure. 
Specifically, the score can never be 0 for \up Jain, unless the number of slots is 0, i.e., $k=0$ or $m=0$, which does not make sense because $k=0$ means that no recommendation at all is outputted, and $m=0$ means that there are zero users to recommend to. The score of \up Ent can only be 0 when there are no items in the dataset ($n=0$), which also does not make sense because it means that there is nothing to recommend. For \up QF/FSat, the score can only be 0 if there are \textit{no} recommended items to \textit{any} users ($|R|=0$), which does not make sense either. 
On the other hand, the score can only be 1 for \down Gini/Gini-w when a single item is recommended at all $k$ slots for each user, which is a highly unlikely artificial outcome; to our knowledge, no reasonably performing recommender model can produce such an output. 
All of the above conditions are unrealistic. 
A consequence of the above is that fairness is overestimated by these measures. Instead of the unrealistic situations above, it is the realistically unfairest recommendation for \up Jain/QF/Ent/FSat that should be mapped to 0 and for \down Gini/Gini-w to 1. 

\subsubsection{Cause 2 (\up Jain, \up QF, \up Ent, \down Gini, \down Gini-w, \up FSat, \down II-D, \down AI-D)}\label{sss:cause2} The second cause of non-realisability that we identify is that when the number of recommendation slots $km$ is less than the number of items $n$, then the score cannot be the fairest, as some items cannot be exposed due to the limited availability of recommendation slots.  
This means that the max score (most fair) cannot be reached by the measure, even if all recommended items are fair. 
In the datasets in Tab.~ \ref{tab:dataset}, $km > n$ for $k \in \{10,20\}$, but for some very large datasets like LFM-1b \cite{Schedl2016TheRecommendation} or Four\-squ\-are NYC \cite{Yang2013ASystem}, $km < n$. 

For II-D and AI-D, we show that the score cannot be the fairest, during single or multiple rounds of recommendations. For example, given $k=1, m=2, n=3$, and considering all possible orders of recommendations for a single round of recommendation, the lowest scores for II-D and AI-D are \nicefrac{2}{9} and \nicefrac{1}{18} respectively. 
In the case of multiple recommendation rounds, the total number of slots is $kmW$, where $W$ is the number of rounds. As an example, given $k=1, m=2, W=2, n=5$, the lowest possible II-D and AI-D are 0.06 and 0.01 respectively. So, in these cases, the lowest score (that should be the fairest) is not zero, even if the recommendations are made as fair as possible (close to random exposure).

This leads to an underestimation of fairness: the fairest value of the measure cannot be reached for some datasets (regardless of recommendation quality), so the measure is not evenly robust across datasets and can underestimate fairness. 

\subsubsection{Cause 3 (\up Jain, \up Ent, \down Gini, \down Gini-w \down II-D, \down AI-D)}\label{sss:cause3} The third cause of non-realisability that we identify is that for Jain, Ent, Gini, and Gini-w, even when $km > n$, the measure still cannot reach the theoretical fairest value if the number of items is not an exact multiple of the recommendation slots, $n \nmid km$. E.g., if $km=4$, $n=3$, three slots can be filled with one unique item each, but no matter which item fills the last slot, one item will be recommended one more time than the rest. 
Likewise, the same applies for II-D and AI-D when the number of slots across all rounds ($kmW$) is not divisible by $n$. For example, when $k=m=W=2, n=3$, the minimum II-D is $0.02$ and the minimum AI-D is $0.005$. 
This limitation consequently leads to the same issue of robustness and underestimation of fairness, as described in \textit{Cause 2}.

\subsubsection{Cause 4 (\down Gini-w, \down VoCD, \down II-D, \down AI-D)}
\label{sss:cause_nonreal}
The fourth case of non-realisability that we identify is that measures cannot reach the theoretical (un)fairest value, as the exact formulation of the max/min achievable score is 
unknown,\footnote{Here, `unknown' is only related to the exact max/min formulation, as the maximum and minimum can always be computed by enumerating all possibilities, albeit being a costly process.}  
making the score hard to interpret. This happens because the most/least fair recommendation cannot be analytically determined due to parameters in the measure or item exposure being weighted by a non-uniform examination function.  
This causes the measure to have a range different from its theoretical range, i.e., $[0,1]$, which we explain next for each measure. 

\down Gini-w may not reach $0$ or $1$ due to non-uniform exposure. E.g., when $k=n=3, m=2$, the minimum and maximum \down Gini-w for all possible recommendation lists are $0.0373$ and $0.156$ respectively. As $n|km$, this is separate from \textbf{non-realisability}, \textit{Causes 1--3}.

For \down VoCD, as $CD(i,i') \in [0,1)$,  \down VoCD $\in [0,1-\beta)$. However, the score of \down VoCD depends on item similarity, making the most unfair score unreachable. Even though it is impossible to formulate an exact achievable maximum value for \down VoCD, 
we formally prove in App.~\ref{app:vocd} that $\forall \alpha \in [0,2], \beta \in [0,1)$, the maximum \down VoCD, $\text{VoCD}_{\max} \leq \frac{m-1}{m} - \beta$. This is obtained when there is only one pair of similar items which is recommended $1$ and $m$ times each. E.g., 
$R_{u_1}^{2} = [i_1, i_2],\ R_{u_2}^{2} = R_{u_3}^{2} = [i_1, i_3]$ where $\text{VoCD}_{\max}$ is $\nicefrac{2}{3}$, which happens when only $i_1$ and $i_2$ are similar. So, the most unfair score is non-realisable, as Eq.~\eqref{eq:vocd-ori} depends on item-pair similarity. 

\down II-D and \down AI-D also may not reach 0 or 1, even in the context of multiple rounds of recommendations and having enough slots for all items. For example, when $k=m=W=2, n=3$, considering all possible ways of recommending items, the minimum values of II-D and AI-D are 0.02 and 0.005 respectively, while the maximum values are 0.187 for both. We posit this to be due to the exponential-like exposure.

A summary of situations producing the theoretical most (un)fair scores in existing measures is given in Tab.~\ref{tab:situations}.

\begin{table}[tb]
    \centering
    \caption{Situations that produce the theoretical most (un)fair score in existing individual item fairness measures.}
    \resizebox{1.0\textwidth}{!}{
    \begin{tabular}{lll}
        \toprule
         Measure&  Most Unfair & Most Fair\\
         \midrule
         Jain & no recommendation slots ($km=0$)& all items recommended the same amount\\
         QF& no items recommended ($|R|=0$)& all items exposed, no matter how many times\\
         Ent & no items in the dataset ($n=0$)& all items recommended the same amount\\
         Gini, Gini-w & a single item is recommended at all rank positions for all users  & all items recommended the same amount (or same total exposure weight)\\
         FSat& no items recommended ($|R|=0$)& all items recommended $\geq \left\lfloor\frac{km}{n} \right\rfloor$ times\\
         VoCD& not possible to deduce from the formula and description& \parbox[t]{98mm}{all pairs of similar recommended items recommended similar amount with a normalised difference} \\
         II-D &not possible to deduce from the formula and description& \parbox[t]{98mm}{exposure distribution of recommended items matches exposure given by random distribution}\\
         AI-D &not possible to deduce from the formula and description& \parbox[t]{98mm}{exposure distribution of recommended items matches exposure given by random distribution, with the most possible number of unique items in top $k$}\\ 
         \bottomrule
    \end{tabular}
    }
    \label{tab:situations}
\end{table}

\subsection[Quantity-insensitivity (QF)]{Limitation 2: Quantity-insensitivity (QF\ori)}\label{ss:quantity-insensitivity} 
This limitation is part of the design choice by the authors of the original QF measure \cite{Zhu2020FARM:APPs} based on a specific concept of fairness, which we explain next. 
Quantity-insensitivity means that the measure ignores how often an item is recommended across all users in a recommendation round. 
In economics, \cite{Allison1978MeasuresInequality} states that `sensitivity to transfer' is a basic criterion of an inequality measure. Similarly, we think that when exposure increases (or decreases) for an item, the fairness measure should be sensitive to the change. 
Meanwhile, QF\ori~ makes no distinction between items that are recommended once or more than once.
To illustrate, consider these scenarios: 1) $R_{u_1}^{2}=[i_1, i_2], R_{u_2}^{2}=[i_2, i_3],  R_{u_3}^{2}=[i_1, i_3]$ and 2) $R_{u_1}^{2}=R_{u_2}^{2}=[i_1, i_2],  R_{u_3}^{2}=[i_1, i_3]$. Assuming $n=5$,  QF\ori~ would be 0.6 for both cases, even though in the second scenario $i_1$ is recommended more times than $i_2$, which is recommended more than $i_3$.

As a result of this design choice, the score does not reflect the repeated recommendations of the same item to many users, which may indicate unfairness (e.g., popularity bias). This is a design limitation that one should be aware of when using QF.

\subsection[Undefinedness (Ent)]{Limitation 3: Undefinedness (Ent\ori)} 
\label{ss:undefinedness} 
We define the limitation of undefinedness as the measure giving an undefined value. In practice, this limitation renders the measure incomputable when encountering an (edge) case.\footnote{This is reminiscent of the completeness property in \cite{Moffat2013SevenMetrics}.} 
This is not negligible, as an assessment question for the design of (fairness) measures is related to how the measure responds to edge cases \cite{Raj2022MeasuringResults}. 
For Ent\ori, the case is related to the possibility of encountering the undefined value of $\log{0}$ during computation.
Undefinedness happens when there is at least one item from the dataset that does not appear in the recommendations at all,\footnote{
$\log{p(i)}$ in Eq.~\ref{eq:ent-ori} is undefined if item $i$ is not in the recommendation list for any users ($p(i)=0$).} 
which happens often because not all items in the datasets are guaranteed to be at the top $k$. 
For example, given $R_{u_1}^{2}=R_{u_2}^{2}=[i_1, i_2]$ and $I=\{i_1, i_2, i_3\}$, the value of $p(i_3)=0$, and this entails $\log{p(i_3)}=\log{0}$. 
Such a situation is common, as it will later be seen in Tab.~\ref{tab:base-main}, and therefore we do not consider this as an edge case. When the measure is incomputable for several models, its interpretation is less meaningful. 

We exclude the case where no item is recommended to any users ($|R|=0$), as this is a trivial case and it does not make much sense to evaluate fairness when there is no item being recommended. Regardless of the triviality, we identify some measures that are incomputable under this edge case: Jain\ori, Ent\ori, Gini\ori, Gini-w\ori, and VoCD\ori. Note that we do not consider these four measures to have the undefinedness limitation just for this reason. 

\subsection{Limitation 4: Always-fair (\ups FSat)}
\label{ss:always_fair}
We define the limitation of always-fair as the measure giving the fairest score regardless of the content of the recommendation list, under a specific condition depending on the particular measure. 
This happens for \up FSat when $km<n$, as empirically discovered by \cite{Patro2020FairRec:Platforms}. The maximin share, in this case, is 0 and all items are deemed satisfied as per the definition in $\S$\ref{sss:fsat}, regardless of the actual distribution of recommended items. 
While this is partly due to the design choice of FSat which is based on the maximin share, this means that \up FSat will always be $1$, rendering the measure unsuitable for use cases where $km<n$. 
Even though this limitation has been empirically identified before, there was no formal definition of it, and that is what we do here.

\subsection{Limitation 5: Item-representation-dependence (\dws VoCD)} \label{ss:item-rep-dep}
We formally identify this limitation in this work, even though it is part of an intentional choice of the measure, as opposed to an accidental or unforeseen byproduct of the design. Item-representation-dependence means that the score of the measure varies according to how item representations are built 
(e.g., embedding, graphs). The max value of \down VoCD depends on which item pairs are similar based on how items are represented. 
Even though the dependence on item representation is part of the design choice for VoCD, the limitation of this design should be taken into consideration when one chooses a fairness measure, e.g. ensuring a same way of representing items for a fair comparison.
 
Depending on how item representations are built, there may be different pairs of similar items in the set $A$, yielding different \down VoCD scores. E.g., given $R_{u_1}^{2} = [i_1, i_2],\ R_{u_2}^{2} = [i_1, i_3]$, if $A=\{(i_1, i_2)\}$ as per one item representation, \down VoCD $= 0.5 - \beta$. However, if according to another item representation, $A=\{(i_2, i_3)\}$, then \down VoCD $=0$. 
While the former score represents a somewhat unfair RS, the latter denotes that the RS is fair. 
Note that one can use the same recommendation algorithm and the same dataset, but with different ways of representing the items, different VoCD scores may be obtained. Therefore, the limitation still holds even when the comparison is only performed within an algorithm and a dataset.

\begin{table}[tb] 
\caption{
Measures of individual item fairness and their theoretical limitations. We identify four different causes for non-realisability, denoted by \textit{C} in this table. 
}
\label{tab:limitation-summary}
\resizebox{\textwidth}{!}{
\begin{tabular}{l|c||c|c|c|c|c|c|c|c|c}
\toprule
\midrule
\parbox[t]{0.4\textwidth}
{Legend\\
\bbullet: we fully resolve the limitation\\
\nofix: the limitation is unresolvable ($\S$\ref{ss:nofix})\\
\checkmark: another measure resolves the limitation} 
&\rotatebox[origin=r]{90}{Source}
&\rotatebox[origin=r]{90}{Jain \cite{jain1984quantitative}}
&\rotatebox[origin=r]{90}{QF \cite{Zhu2020FARM:APPs}}
&\rotatebox[origin=r]{90}{Ent \cite{Shannon1948ACommunication}}
&\rotatebox[origin=r]{90}{Gini \cite{Gini1912VariabilitaMutabilita}}
&\rotatebox[origin=r]{90}{Gini-w \cite{Do2021Two-sidedDominance}}
&\rotatebox[origin=r]{90}{FSat \cite{Patro2020FairRec:Platforms}}
&\rotatebox[origin=r]{90}{VoCD \cite{Wang2022ProvidingSystems}}
&\rotatebox[origin=r]{90}{II-D \cite{Wu2022JointRecommendation}}
&\rotatebox[origin=r]{90}{AI-D \cite{Wu2022JointRecommendation}}
\\
\midrule
\midrule
non-realisability: cannot reach max/min score (cause number denoted by \textit{C}) &&&&&&&&&\\
\textit{C1.} Most unfair score is only given to an impossible scenario    &us&\bbullet&\bbullet&\bbullet&\bbullet&\bbullet&\bbullet&&&\\
\textit{C2.} Fewer recommendation slots compared to number of items&us&\bbullet&\bbullet&\bbullet&\bbullet&\bbullet&\bbullet&&\nofix&\nofix\\
\textit{C3.} Number of recommendation slots is indivisible by number of items&us&\bbullet&&\bbullet&\bbullet&\nofix&&&\nofix&\nofix\\
\textit{C4.} Non-realisability due to unknown formulation of max/min score
&us&&&&&\nofix&&\nofix&\nofix&\nofix\\
\midrule 
quantity-insensitivity: ignores frequency of item recommendation   &\cite{Zhu2020FARM:APPs}&&$\checkmark$&&&&&&&\\
\midrule
undefinedness: cannot be computed (undefined value)   &us&&&\bbullet&&&&&&\\
\midrule 
always-fair: gives fairest score regardless of recommendation contents   &\cite{Patro2020FairRec:Platforms}&&&&&&\nofix&&&\\
\midrule 
item-representation-dependence: depends on how items are represented   &us&&&&&&&\nofix&&\\
\midrule 
\bottomrule
\end{tabular}
}
\end{table}

\section{Resolving limitations}
\label{s:extensions}

We explain how we resolve each limitation or why it is unresolvable. For the remainder of this paper, we refer to the original version of an evaluation measure $M$ as $M_{\text{ori}}$, and to our modified version of an evaluation measure $M$ as $M_{\text{our}}$. When $\cdot$\ori~or $\cdot$\our~is not specified, we refer to both the original and modified version simultaneously.

\subsection{Resolving non-realisability (Limitation 1) and undefinedness (Limitation 3)}
\label{ss:non_realisability}

\begin{table}[tb]
\caption{
Most (un)fair score @$k$. Note that the Most Fair @$k$ scores (except for Gini-w) remain the same as the theoretical fairest value (i.e., 0 or 1) when the number of recommendation slots is divisible by the number of items, $n \mid km$.}
\label{tab:bounds}
\begin{tabular}{lll}
\toprule
Measure &     Most Unfair @$k$ & Most Fair @$k$ \\
\midrule
   \up Jain\ori &          $\frac{k}{n}$ &        $\frac{(km)^2}{
 n\left(n\floorkm^2 
 + (km\bmod{n})\left(2 \floorkm+1\right) 
 \right)} $ \\
   \up QF\ori &            $\frac{k}{n}$ &        $\min\left(\frac{km}{n},1\right)$ \\
   \up Ent\ori &           $\log{k}$ &            Eq.~\eqref{eq:ent-max}\\ 
   \down Gini\ori &          $1-\frac{k}{n}$ &      $ \frac{(n-km \bmod n) (km \bmod n)}{kmn} $\\ 
  \down Gini-w\ori &          Eq.~\eqref{eq:gini-w-max} &      Eq.~\eqref{eq:gini-w-min} when $km\leq n$\\
   \up FSat\ori &          $\frac{k}{n}$   &      1\\
   \down VoCD\ori &          $\frac{m-1}{m} - \beta$ &      0\\
\bottomrule
\end{tabular}
\end{table}

We resolve causes 1, 2, and 3 of non-realisability via post-calculation correction of under/overestimated fairness scores based on the theoretical min and max values for a scenario with limited $km$ recommendation slots. 
Specifically, we rescale the range of the measures to the actual theoretically achievable most fair and unfair values. The rescaled measures either retain the $[0,1]$-range, or are now ranged in $[0,1]$. To rescale, we compute the measure's upper and lower bounds when possible (see Tab.~\ref{tab:bounds} and App.~\ref{app:boundsproof} for the derivations) by considering the most fair and unfair recommendation case for each measure, which we explain next.

For the measure that suffers from non-realisability due to \textit{Causes 1--2} (and quantity-insensitivity) i.e., QF, we posit that the most unfair recommendation $@k$ is when the same $k$ unique items are recommended to each of the $m$ users, resulting in the min exposure for items in $I$. 
Therefore, each of these $k$ items is recommended $m$ times. 
This is equivalent to the recommendation generated by Pop \cite{Rashid2002GettingYou} that gives the same $k$ most popular items to all users. 
We posit that the most fair case $@k$ is when $\min(km,n)$ unique items are recommended to $m$ users, i.e., the max number of items allowed by the recommendation slots is exposed.  

For measures that suffer from non-realisability, due to \textit{Causes 1--3} but not quantity-insensitivity (Jain, Ent, Gini, Gini-w, FSat), we consider the most fair recommendation to be when ($n-km \bmod n$) items are exposed $\floorkm$ times and the rest ($km \bmod n$) items are exposed $\floorkm +1$ times. The most unfair case for Jain, Ent, Gini, Gini-w, and FSat is the same as QF.

We then perform min-max normalisation (hereafter referred to as normalisation) using the most unfair/fair bounds as the min/max possible value of the measure. The general process is: let $x_{\max}$ be the max possible value of a fairness score $x$, and $x_{\min}$ be the min possible value of $x$. The normalised score of $x$, denoted by $x'$, is calculated using $x'= \frac{x-x_{\min}}{x_{\max}-x_{\min}}$. Note that this normalisation does not work when $x_{\max}=x_{\min}$ due to division by zero. This happens when the most unfair recommendation is equal to the fairest recommendation, which occurs when $k=n$. We therefore exclude this case from our corrections.

After normalisation, the measures quantify fairness of items by considering the following components: the recommendation list, the number of items in the dataset, and the most fair and most unfair recommendation. As a result, the most fair recommendation scenario and most unfair recommendation scenario are now mapped to the endpoints, instead of the unrealistic scenarios in Tab.~\ref{tab:situations}. 
Next, we explain in detail the normalisation of each measure.

For \up Jain\ori, we apply the following normalisation on Eq.~\eqref{eq:jain-ori}: as the higher the score, the fairer, we use Eq.~\eqref{eq:jain-max} as $x_{\max}$ and $x_{\min} = \frac{k}{n}$ because these are the most fair and unfair @$k$ values.  Eq.~\eqref{eq:jain-max} simplifies to $y=\frac{km}{n}$ when $km<n$, and simplifies to 1 when $n \mid km$.

 \noindent\begin{minipage}{0.55\linewidth}
    \begin{equation}
    \label{eq:jain-max}
         \text{Jain}_{\max} = \frac{(km)^2}{
         n\left(n\floorkm^2 
         + (km\bmod{n})\left(2 \floorkm+1\right) 
         \right)}   
    \end{equation}
\end{minipage}%
\begin{minipage}{0.45\linewidth}
   \begin{equation}
    \label{eq:jain-norm}
        \text{Jain}_{\text{our}}  =  \frac{\text{Jain}_{\text{ori}}-\frac{k}{n}}{\text{Jain}_{\max}-\frac{k}{n}}
    \end{equation}
\end{minipage}

We normalise \up QF\ori~ similarly to the above. As the higher the QF score, the fairer, we use $x_{max} = \min(y,1)$ and $x_{min} = \frac{k}{n}$. 
\begin{equation}
\label{eq:qf-norm}
\text{QF}_{\text{our}}  = 
\begin{cases}
    \frac{\text{QF}_{\text{ori}}-\frac{k}{n}}{1-\frac{k}{n}} 
    = \frac{|R|-k}{n-k} & \text{if } km \geq n \\
    \frac{\text{QF}_{\text{ori}}-\frac{k}{n}}{\frac{km}{n}-\frac{k}{n}} 
    = \frac{|R|-k}{k(m-1)} & \text{otherwise}
\end{cases}
\end{equation}

We acknowledge that by performing this normalisation on QF\ori, there is now an additional way of measuring fairness. The new score can now be interpreted as ``given that each item should be recommended at least once, how fair the recommendation is w.r.t. the most unfair and the fairest recommendation''. The most (un)fair recommendation and the normalisation depend on the number of recommendation slots, and we argue that it is better to use this number instead of using the number of items in the datasets; in practice, the number of items shown to the users is almost always limited. However, if QF is meant to be used for detecting the percentage of items that are exposed, then the QF\ori~ may be used at the expense of needing additional information on what is the best possible QF score, to know the limit of how fair the recommendation can be.
 
For \up Ent\ori, as the higher the score, the fairer, we use 
 Eq.~\eqref{eq:ent-max} as  $x_{\max}$ when $km \geq n$ or $x_{\max}=\log{km}$ otherwise, and 
 $x_{\min} = \log{k}$ for normalisation.  Eq.~\eqref{eq:ent-max} simplifies to $\log{n}$ when $n \mid km$.
\begin{align}
\label{eq:ent-max}
\text{Ent}_{\max} = 
-
(n - km \bmod n)\left(\frac{\floorkm}{km}\log{\frac{\floorkm}{km}}\right)
- (km \bmod n)\left(\frac{\floorkm+1}{km}\log{\frac{\floorkm+1}{km}}\right) 
\end{align}

However, \up Ent\ori~ still suffers from \textbf{undefinedness}, which we resolve by restricting the sum over $i$ in Eq.~\eqref{eq:ent-ori} to only recommended items:
\begin{equation}
\label{eq:ent-fair}
    \text{Ent$_{\text{def}}$} = 
    - \sum\limits_{i \in R}{p(i) \log{p(i)}} 
\end{equation}
where $p(i)$ is calculated via Eq.~\eqref{eq:ent-ori} and $p(i)>0$ because $i \in R$. Performing normalization on Eq.~\eqref{eq:ent-fair}, we obtain:
\begin{equation}
\label{eq:ent-norm}
\text{Ent}_{\text{our}}  = 
\begin{cases}
    \frac{\text{Ent}_{\text{def}}-\log{k}}{\text{Ent}_{\max}-\log{k}}& \text{if } km \geq n \\
    \frac{\text{Ent}_{\text{def}}-\log{k}}{\log{km}-\log{k}}
    = \frac{\text{Ent}_{\text{def}}-\log{k}}{\log{m}}  & \text{otherwise}
\end{cases}
\end{equation}

For \down Gini\ori, as the lower the score, the fairer, we use Eq.~\eqref{eq:gini-min} as $x_{\min}$ and $x_{\max} = 1-\frac{k}{n}$. Eq.~\eqref{eq:gini-min} simplifies to $1-y$ when $km>n$, and simplifies to 0 when $n \mid km$. 
\begin{multicols}{2}\noindent
    \begin{equation}
    \label{eq:gini-min}
     \text{Gini}_{\min} = 
     \frac{
     (n-km \bmod n) (km \bmod n)
     }
     {kmn} 
    \end{equation}
    \begin{equation}
    \label{eq:gini-norm}
    \text{Gini}_{\text{our}}  = 
    \frac{\text{Gini}_{\text{ori}}- \text{Gini}_{\min} }{1-\frac{k}{n} -  \text{Gini}_{\min} } 
\end{equation}
\end{multicols} For \down Gini-w\ori, we only consider cases where $km \leq n$ as finding the most fair recommendation list across all users for the other cases is analytically not possible. The problem does not have a closed form solution and computing the solution requires solving a constrained optimization that considers all possible permutations of recommendations across users. 
We use Eq.~\eqref{eq:gini-w-min} as $x_{\min}$ when $km \leq n$, $x_{\min}=0$ otherwise, and Eq.~\eqref{eq:gini-w-max} as $x_{\max}$ to normalise \down Gini-w\ori.
\begin{multicols}{2}\noindent
    \begin{equation}
\label{eq:gini-w-min}
 \text{Gini-w}_{\min} =
  \frac{
    \sum\limits_{\ell=1}^k 
    \sum\limits_{j=n-\ell m+1}^{n- \ell m + m} (2j-n-1) \log_{\ell + 1}{2}}
    {mn\sum\limits_{\ell=1}^k{\log_{\ell+1}{2}}}
\end{equation}
\begin{equation}
\label{eq:gini-w-max}
 \text{Gini-w}_{\max} = 
 \frac{
 \sum\limits_{\ell=1}^k{(n-2\ell+1) \log_{\ell+1}{2}}}
    {n\sum\limits_{\ell=1}^k{\log_{\ell+1}{2}}}
\end{equation}
\end{multicols}\noindent
\begin{equation}
\label{eq:gini-w-norm}
\text{Gini-w}_{\text{our}}  = 
\begin{cases}
    \frac{
    \text{Gini-w}_{\text{ori}}- \text{Gini-w}_{\min}}
    {
    \text{Gini-w}_{\max}  -  \text{Gini-w}_{\min} } & \text{if } km \leq n \\ 
    \frac{
    \text{Gini-w}_{\text{ori}}}
    {
    \text{Gini-w}_{\max}}  & \text{otherwise}
\end{cases}
\end{equation}

For \up FSat\ori, as the higher the score, the fairer, we normalise Eq.~\eqref{eq:fsat-ori} using $x_{\max}=1$ and $x_{\min}=\frac{k}{n}$.\footnote{VoCD, II-D, and AI-D can be normalised in a similar way after computationally approximating their empirical minimum and maximum values.}  

\begin{equation}
\label{eq:fsat-norm}
\text{FSat}_{\text{our}}  = 
    \frac{\text{FSat}_{\text{ori}}-\frac{k}{n}}{1-\frac{k}{n}} 
\end{equation}

\subsection{Resolving quantity-insensitivity (Limitation 2)}\label{ss:resolve-quantity} 
While the quantity-insensitivity limitation is due to the design choice of the QF measure ($\S$\ref{ss:quantity-insensitivity}), we reason that a solution to this limitation is needed in case one would like to use a measure that is equation-wise similar to QF\ori, but needs the measure to be sensitive to the change of exposure received by the item. Hence, the \textbf{quantity-insensitivity} limitation for \up QF\ori~ may simply be resolved by calculating \up Jain\ori~ instead, as \up QF\ori~ o\-riginates from \up Jain\ori~ (see Eq.~\ref{eq:jain-ori} and \ref{eq:qf-ori}). \up Jain\ori~ gives different weights between items that have been recommended once and more than once. Referring to the toy example given in $\S$\ref{ss:quantity-insensitivity} where \up QF\ori~ would be 0.6 for both cases, but \up Jain\ori~ $\approx 0.514$ for the first scenario and \up Jain\ori~ $=0.6$ for the second scenario. \up Jain\ori~ returns a higher fairness score in the second scenario as the distribution of the frequency count of items is more balanced.

\subsection{Unresolvable limitations}
\label{ss:nofix} 
We resolve three out of five limitations (see Tab. \ref{tab:limitation-summary}). The remaining limitations are unresolvable for the following reasons. 

\noindent\textbf{Non-realisability} 
due to \textit{Cause 4} cannot be resolved because no closed-form solutions have been discovered for the recommendations that produce the best and worst fairness scores for each measure. Computing the solution requires solving a constrained optimization problem that cannot be practically solved for large datasets such as \rs~data. 
While the measures could theoretically be corrected in a similar manner to the ones in $\S$\ref{ss:non_realisability}, it is only possible to do so after computing the most/least fair recommendation list, which is impractical.

\noindent\textbf{Item-representation-dependence} cannot be resolved because item representation is required by the measure to determine whether items are similar. It is avoidable if all items are considered similar to each other, but this is unrealistic. 
Moreover, to get comparable scores, all \rs~should use the same representation of items, which cannot be guaranteed. This point is not handled in the original definition of VoCD \cite{Wang2022ProvidingSystems}, where the only item representation considered is item embeddings. While it is possible to use the same external representation for item representation (e.g. similarity matrix based on tags, keyword, or other features) for different \rs, as commonly done with diversity metrics \cite{Ziegler2005ImprovingDiversification}, the fairness score may still change depending on the representation used to determine the item similarity.

\noindent\textbf{Always-fair} cannot be resolved for FSat\ori~ because attempting to resolve this would require tampering with the definition of `satisfied' based on the concept of maximin share. This definition of `satisfied' is an integral part of the measure. Replacing the maximin share criterion with another requirement would turn FSat\ori~ into a different measure. 
Therefore, as this limitation is related to the concept of measure, we are unable to recommend a solution to the limitation without changing the design of the measure.

\section{Empirical analysis}
\label{s:exp}
We experimentally analyse the relevance and fairness of several recommenders and compare the original measures ($\S$\ref{s:priorwork}) to our corrected versions ($\S$\ref{s:extensions}) for the six datasets shown in Tab.~\ref{tab:dataset}. We present the results for Lastfm and Ml-1m in this section, and for the other datasets in App.~\ref{app:extend-result}.

\subsection[Experimental setup]{Experimental setup}
\label{ss:setup}

\noindent \textbf{Dataset preprocessing}. 
We use six freely available datasets (see Tab.~\ref{tab:dataset}) from~\cite{Zhao2021RecBole:Algorithms}:\footnote{\url{https://github.com/RUCAIBox/RecSysDatasets}} 
Lastfm;\footnote{\url{http://www.lastfm.com} \label{fn:lfm}} Ml-1m \cite{Harper2015TheContext}; Book-x \cite{Ziegler2005ImprovingDiversification}; Amazon-lb, Amazon-dm, and Amazon-is \cite{Ni2019JustifyingAspects}. We remove users/items with $<5$ interactions and use $80\%/10\%/10\%$ to train/validate/test, with a user-based random split for Lastfm and Book-x (timestamps are not available), and a user-based temporal split for all other datasets, i.e., the last $10\%$ of each user's interactions are in the test set. We convert ratings $\ge3$ on Ml-1m \& Amazon-* and ratings $\ge6$ on Book-x to 1. We discard the rest of the ratings. 
We choose these thresholds as the ratings are from 1--5 in Ml-1m and Amazon-*, and 0--10 in Book-x. We do not convert for Lastfm as it uses implicit feedback, so all interactions have a value of 1. For duplicate values, we keep the last interaction.  

\noindent \textbf{Recommenders}. For recommendation we use: Pop \cite{Rashid2002GettingYou} (recommends $k$ most popular items), item-based K-Nearest Neighbours (ItemKNN) \cite{Deshpande2004Item-basedAlgorithms}, 
Sparse Linear Method (SLIM) \cite{Ning2011SLIM:Systems}, 
Bayesian Personalized Ranking (BPR) \cite{RendleBPR:Feedback}, 
Neural Graph Collaborative Filtering (NGCF) \cite{Wang2019NeuralFiltering}, 
Neural Matrix Factorization (NeuMF) \cite{He2017NeuralFiltering}, and 
Variational Autoencoder with multinomial likelihood (MultiVAE) \cite{Liang2018VariationalFiltering}. 
We use training batch sizes of $4096$, 
Adam~\cite{Kingma2014Adam:Optimization} as optimizer, and 
the RecBole library ~\cite{Zhao2021RecBole:Algorithms}.
We train BPR, NGCF, NeuMF, and MultiVAE for $300$ epochs, but 
use early stopping of $10$ epochs and keep the model that produces the best NDCG@10 on the validation set. We tune hyperparameters on all models except Pop, with RecBole's hyperparameter tuning module. The hyperparameter search space and optimal hyperparameters are in App.~\ref{app:extend-exp_set_up}. 
For all recommenders, when we generate the recommendation list for a user during testing, the items in the user's train or validation set are placed at the end of the user's list to avoid re-recommending them. 

\noindent\textbf{Measures}.
We evaluate models w.r.t.~a) relevance-only measures (HR, MRR, Precision (P), Recall (R), MAP, NDCG), and b) individual item fairness measures, both the original and our corrected measures. 
All measures are computed at $k=10$, unless otherwise stated. We evaluate on the full test set of items instead of a sample of them, as doing the latter is known to yield misleading results~\cite{Krichene2020OnRecommendation}. This leads to lower performance than reported when sampling the test set. 
Lastly, 
for Ent, we use the log base-$n$. 
For VoCD we choose the values of $\alpha$ and $\beta$ such that VoCD maintains comparability with the other fairness measures: all recommended items are considered similar\footnote{
All recommended items will be treated as similar items as $sim(i,i') \geq 1-2 = -1$ and $sim(i,i')\in [-1,1]$.} ($\alpha=2$)
and thus $A$ is the set of all possible pairs of different items in the top $k$, 
without any tolerance for coverage disparity ($\beta=0$). We also choose this configuration to avoid reliance on similarity scores based on item embeddings. 
For II-D and AI-D, we use $\gamma=0.8$ \cite{Wu2022JointRecommendation}.

\begin{table}[tb]
\caption{Statistics of the datasets before and after our preprocessing.}
\label{tab:dataset}
\begin{tabular}{lrrrr}
\toprule
\textbf{dataset} & \multicolumn{1}{l}{\textbf{\#users}} & \multicolumn{1}{l}{\textbf{\#items}} & \multicolumn{1}{l}{\textbf{\#interactions}} & \multicolumn{1}{l}{\textbf{sparsity (\%)}} \\ 
\midrule
\multicolumn{5}{c}{\textit{original (as provided by \cite{Zhao2021RecBole:Algorithms})}}                                                   \\ \midrule
Lastfm\footnoteref{fn:lfm} & 1,892 & 17,632 & 92,834 & 99.7217\% \\
Ml-1m \cite{Harper2015TheContext} & 6,040 & 3,706 & 1,000,209 & 95.5316\% \\
Book-x \cite{Ziegler2005ImprovingDiversification} & 105,283 & 340,556 & 1,149,780 & 99.9968\% \\
Amazon-lb  \cite{Ni2019JustifyingAspects} & 416,174 & 12,120 & 574,628 & 99.9886\% \\
Amazon-dm \cite{Ni2019JustifyingAspects} & 840,372 & 456,992 & 1,584,082 & 99.9996\% \\
Amazon-is \cite{Ni2019JustifyingAspects} & 1,246,131 & 165,764 & 1,758,333 & 99.9991\% \\
\midrule
\multicolumn{5}{c}{\textit{preprocessed (by us)}}   \\ \midrule                                                                     
Lastfm & 1,859 & 2,823 & 71,355 & 98.6403\% \\
Ml-1m & 6,038 & 3,307 & 835,789 & 95.8143\% \\
Book-x & 5,639 & 7,455 & 91,385 & 99.7826\% \\
Amazon-lb & 1,644 & 791 & 16,765 & 98.7108\% \\
Amazon-dm & 11,750 & 9,462 & 116,681 & 99.8951\% \\
Amazon-is & 6,574 & 3,569 & 45,762 & 99.8050\%\\ 
\bottomrule
\end{tabular}
\end{table}

\subsection{Analysis of relevance and fairness}
\label{ss:performance} 
We start by studying the relevance and fairness of several recommender models. We compare a) different recommender models w.r.t. relevance and fairness scores, and b) different evaluation measures, including corrected and uncorrected measures. The goal of a) is to study whether relevance and/or fairness scores vary between models, and to obtain a ranking of models that is used in subsequent analysis ($\S$\ref{ss:corr}). The goal of b) is to study measures based on diverse concepts of fairness, and highlight their differences (and similarity, if any). 

The evaluation results are shown in Tab. \ref{tab:base-main} for Lastfm and Ml-1m, and in App.~\ref{app:extend-performance} for the other datasets.

\input{tabsplit1/base-main.tex}

\subsubsection{Discussion of recommendation models in Tab.~\ref{tab:base-main}} 

BPR has the highest relevance scores, while MultiVAE generally has the highest fairness scores. Between BPR and MultiVAE, we observe that the relevance scores are higher in BPR but the fairness scores are higher for MultiVAE. 
E.g., in Lastfm, NDCG $=0.223$ for BPR and NDCG $=0.219$ for MultiVAE. Meanwhile, the higher-is-better fairness scores range between $[0.078, 0.656]$ for BPR and $[0.132, 0.763]$ for MultiVAE. 
This is not always observed for all models. E.g., for ItemKNN and SLIM, better fairness tends to be accompanied by better relevance. Furthermore, other discrepancies exist: the relevance of ItemKNN and MultiVAE is on par, but their fairness scores e.g., the scores of \up Jain of ItemKNN, are only half of those achieved by MultiVAE. 
Generally, the recommenders agree in relative ordering of scores, but some models have higher scores for .\our~than .\ori~and vice-versa, e.g., for Lastfm, \down Gini\ori$>$\down Gini\our~ for ItemKNN but \down Gini\ori$<$\down Gini\our~ for SLIM. Overall, we observe that: 

\begin{itemize}
    \item A recommender model that is the best in terms of relevance may also be relatively fair.
    \item The recommender models mostly have a similar ordering of the scores of fairness measures: if a fairness measure has a higher value than another in a recommender model X, it is the same for recommender model Y.
    \item Some models achieve relatively similar relevance scores, but with a huge disparity between their fairness scores.
\end{itemize}

\subsubsection{Discussion of fairness evaluation measures in Tab.~\ref{tab:base-main}} 
For the higher-is-better measures, Jain, QF, FSat, and Gini, the scores of the original measures and our measures are similar. Both \up Jain and \up QF should range in $[0,1]$, but \up Jain is very close to 0 i.e., ($\sim$0.1 or less), and \up QF scores are $\sim$0.7 (in Lastfm) and $\sim$0.5 (in Ml-1m). 
Similarly, the scores of \down II-D and \down AI-D are also very close to 0 while \down Gini scores are closer to 1. 
While these are due to the different underlying fairness ideas between the measures, the big differences in scores may cause confusion, e.g. that a recommendation is very unfair based on \up Jain or \down Gini, or moderately fair based on \up QF. For Lastfm and Ml-1m, we also see that the absolute scores for the same recommender, e.g., MultiVAE, follow the same order from the lowest to the highest: \up Jain, \up FSat, \up QF, and \up Ent. 
This indicates that \up Jain tends to give lower scores (more unfair) than the other measures. We observe similar trends for \down II-D and \down AI-D, which tend to give lower scores (more fair) compared to other lower-is-better measures.
The weighted \down Gini-w is also more strict than the unweighted \down Gini as \down Gini-w tends to give more unfair scores than \down Gini. We study further the strictness of these measures in $\S$\ref{ss:insert}.

We also observe that scores of \down AI-D are hardly distinguishable. They differ only in the fourth or more decimal point for both Lastfm and Ml-1m. However, differences in other measures can be seen in the first or second decimal point. The small scores of AI-D may be due to the measure quantifying the disparity between item exposure and random exposure, which is very little when we have a large number of items. This finding suggests that when computing \down AI-D, care should be taken to avoid rounding errors and failure to distinguish the scores due to the floating-point format.

For all datasets and models, the original Ent cannot be calculated because it returns NaN due to zero division errors. This happens because there are items in the dataset that are not recommended. 
Our corrected version of this measure ($\S$\ref{s:extensions}) does not suffer from this problem.

For the same dataset and $k$, regardless of the recommender, the II-D scores are always the same/constant. Due to the fixed amount of slots $km$, within a single recommendation round, the exposure values $E_{u,i} \in \{1, \gamma, \dots, \gamma^k, 0\}$ (see Eq.~\ref{eq:eui}) for all user-item pairs, and the number of user-item pair having a specific exposure value $E_{u,i}$ is always $m$. Both of these properties lead to constant II-D scores, as II-D is calculated by taking a mean squared difference between each $E_{u,i}$ value and a constant value based on random expected exposure. When considering cases with multiple rounds of recommendations, the score of II-D may not remain constant anymore, as the user-item exposure values are aggregated across recommendation rounds, resulting in the possibility of $E_{u,i}$ having linear combinations of values from the set above. We have also illustrated in $\S$\ref{sss:cause_nonreal} that II-D is not constant in multiple recommendation rounds.

Overall, we observe that: 
\begin{itemize}
    \item The different fairness measures have different ranges in these experiments, even if theoretically they have the same range. 
    \item The original Ent is always incomputable in the experiments, and our corrected Ent resolves the issue. 
    \item II-D\ori~ remains constant for the same dataset, rendering this measure notably less meaningful under this single-round experimental set-up. 
    \item Both \down II-D\ori~ and \down AI-D\ori~ have minuscule values, indicating near-perfect fairness even if this contradicts other fairness scores.   
\end{itemize}

\subsection{Correlation between measures}
\label{ss:corr}

When comparing different recommender models, sometimes the ranking of the models (e.g., from the most to least fair score) is more concerning than the absolute values of the measures that we have seen in Tab.~\ref{tab:base-main}. Motivated by this, we analyse the measures' correlation in order to study the agreement of model rankings based on different measures of relevance and fairness. We compare the following things: 
1) the agreement between measures of the same type (relevance or fairness); 
2) the agreement between measures of different types;
3) the agreement between the original measures and our corrections to the original measures; 
4) the agreement between measures across different datasets. 
By performing this analysis, we also gain insights into how measures that capture different fairness concepts (dis)agree with one another. 

We use Kendall's $\tau$ between measures to compute ranking agreement. Fig. \ref{fig:corr-lastfm}--\ref{fig:corr-ml1m} show the Kendall's $\tau$ values between relevance measures and fairness measures for Lastfm and Ml-1m (see App.~\ref{app:corr} for the other data\-sets). The computation is as follows: for each dataset, we rank the models based on the most relevant or most fair scores. 
We omit Ent\ori~as it produces NaN in Tab.~\ref{tab:base-main} and we also omit II-D as the scores for one dataset are the same across models. 
We compute the correlation significance and correct errors arising from multiple testings for a dataset, using the Benjamini-Hochberg (BH) procedure that is based on false discovery rate  \cite{Benjamini1995ControllingTesting}. Upon correction, some correlations are still significant; these are indicated by an asterisk ($^*$) in Fig.~\ref{fig:corr-lastfm}--\ref{fig:corr-ml1m}.\footnote{We also use two more conservative procedures separately to correct the errors: Bonferroni and Holm \cite{Holm1979AProcedure}. Upon correction we obtain no significant results for any tests across the six datasets.}

\begin{figure*}[tb]
\centering
    \includegraphics[width=\textwidth]{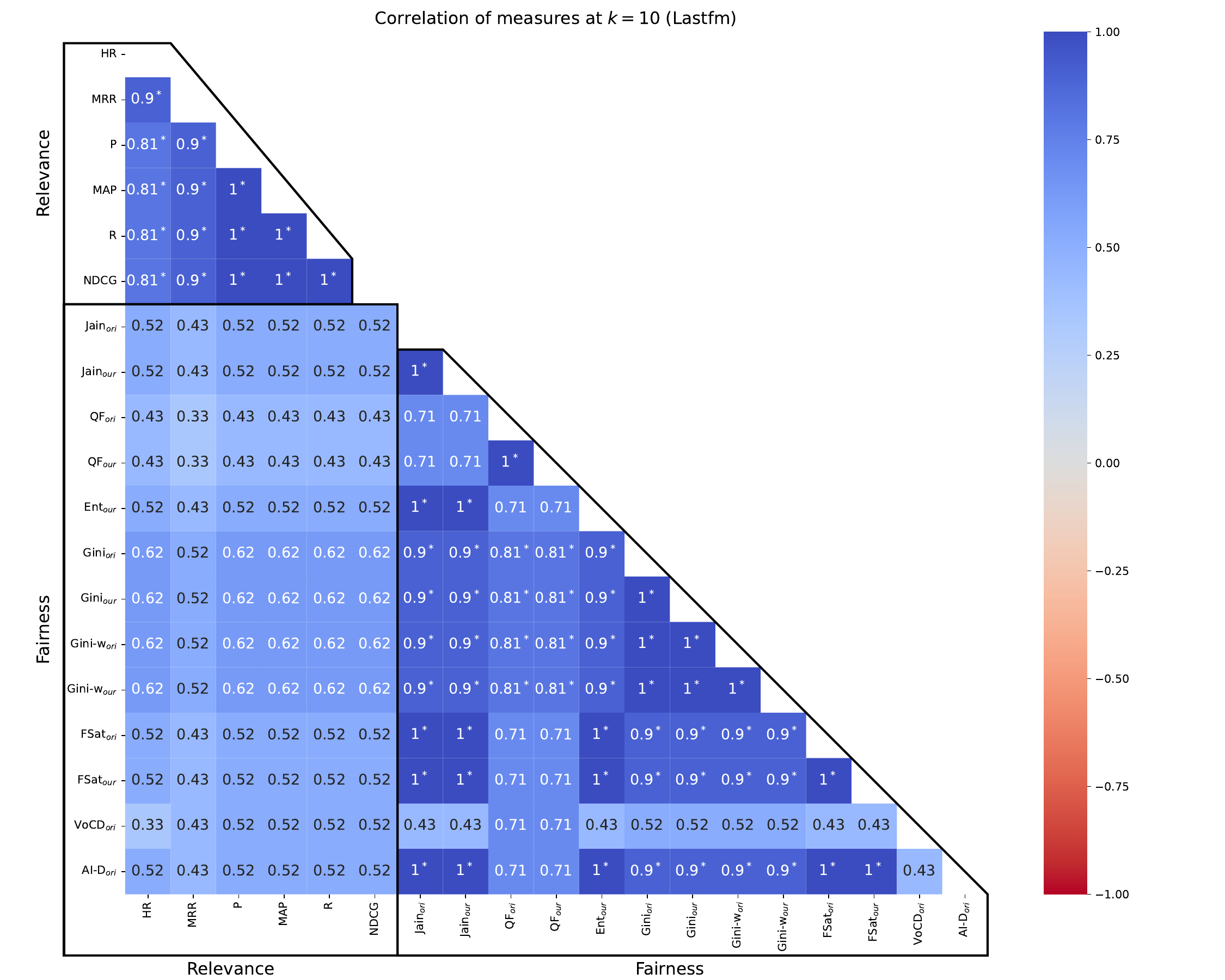}

\caption{Correlation (Kendall's $\tau$) between relevance and fairness measures for Lastfm. \explainsig}
\label{fig:corr-lastfm}
\end{figure*}

\begin{figure*}[tb]
\centering
    \includegraphics[width=\textwidth]{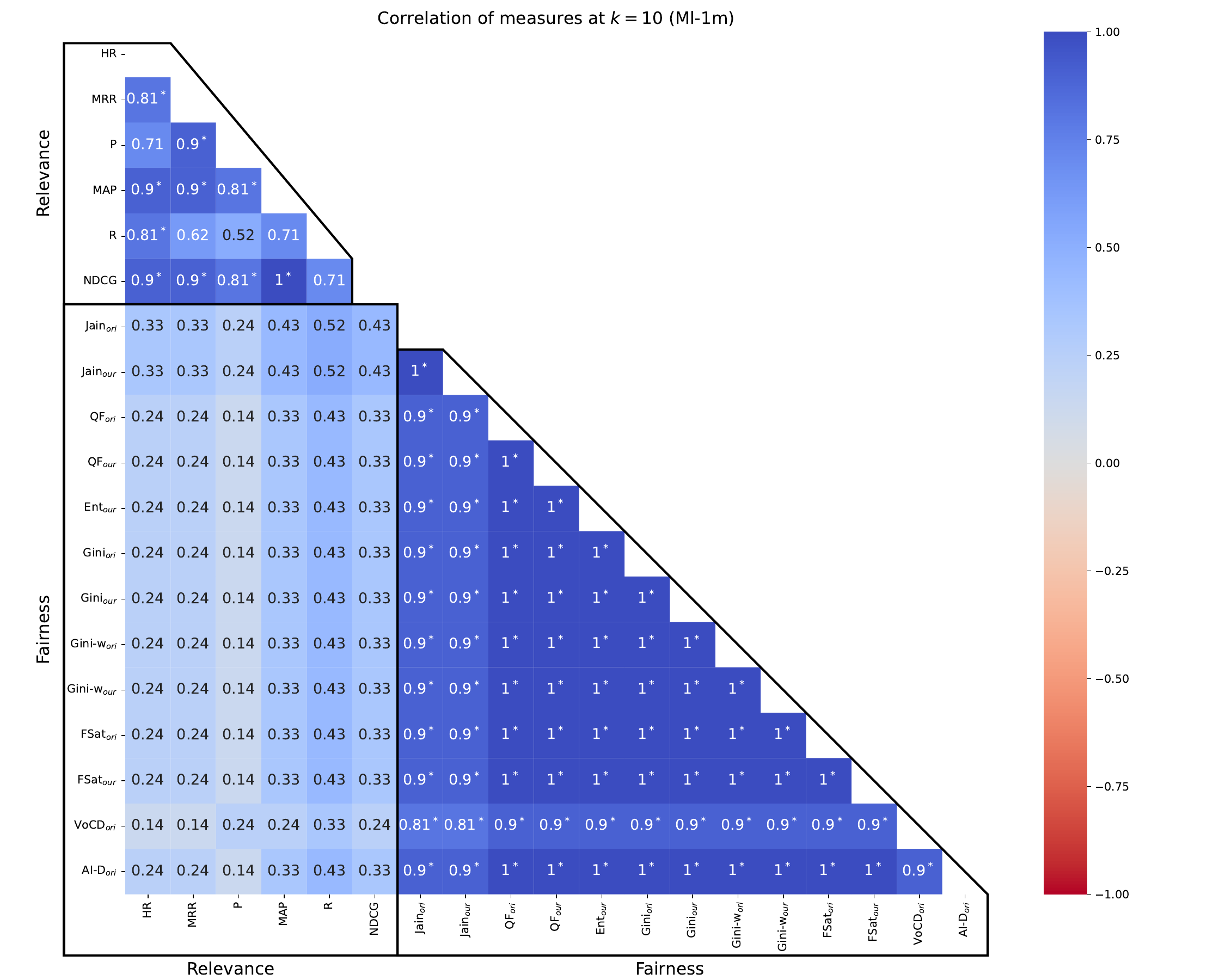}
    \caption{Correlation (Kendall's $\tau$) between relevance and fairness measures for Ml-1m. \explainsig}
\label{fig:corr-ml1m}
\end{figure*}

We first analyse the correlation among measures of the same type. The relevance measures are highly correlated with each other: $[0.81,1]$ for Lastfm and $[0.52,1]$ for Ml-1m, as expected \cite{Webber2008Precision-at-tenRedundant}. The fairness measures are also strongly correlated with each other: $[0.71,1]$ for Lastfm and $[0.81,1]$ for Ml-1m, except VoCD\ori~ for Lastfm, $[0.43,0.71]$. This is expected from how the measures treat items when computing fairness; VoCD\ori~ only considers items in the recommendation list, while the remaining fairness measures consider all items in the dataset. 
For Ml-1m, all the computed correlations between fairness measures are significant after applying the BH procedure. On the other hand, after applying the same procedure for Lastfm, neither of QF nor VoCD, has significant correlations with the rest of the fairness measures, except for QF and Gini/Gini-w. 
It is also reasonable for QF to not correlate significantly with most of the measures, as it is the only measure insensitive to the difference in the number of times an item is recommended. 

Interestingly, even though in Tab. \ref{tab:base-main} the scores of \up Jain, \up QF, \up Ent, and \up FSat occupy different parts of their range, these measures are highly correlated. The same goes for \down Gini and \down AI-D. This shows that even measures based on different concepts of fairness are still capable of producing similar rankings of models. 
Nevertheless, the absolute scores of the measures can be misinterpreted due to the measures occupying different parts of their range.

Our corrected fairness measures are always perfectly correlated with the original fairness measures (1 in both datasets). This is expected because our corrected versions are obtained by normalization, which does not change the relative order of the models.

Regarding the correlations between measures of different types, we see different trends between relevance and fairness measures for Lastfm and Ml-1m. In Lastfm, we see moderate correlations between fairness and relevance measures, $[0.33,0.62]$, but these are lower for Ml-1m $[0.14,0.52]$. These findings are expected as the fairness measures do not consider relevance. None of these correlations are significant after applying the BH procedure. Yet, some correlations between fairness and relevance measures are significant for Book-x and Amazon-is (App.~\ref{app:corr}). 

\subsection{Max/min achievable fairness}
\label{ss:maxmin}

The aim of this experiment is to quantify the extent to which the fairness measures can achieve their theoretical maximum and minimum fairness value (0 or 1) for different datasets and different $k \in \{1,2,3,5,10,15,20\}$. This relates to the \textbf{non-realisability} limitation (\textit{Causes 1--3}). We experiment solely with the fairness measures for which we have resolved this limitation, namely Jain, QF, Ent, Gini, Gini-w, and FSat. 
We primarily compare the original (uncorrected) versions of these measures against the corrected ones. 
We use two settings: repeatable recommendation, where items in the train/val split can be re-recommended to users following practical cases in industry settings; and nonrepeatable recommendation, 
which is the typical setting for evaluating recommender systems in academic work. 
For each setting, we devise two recommenders: MostFair and MostUnfair. Repeatable MostFair aims to recommend each item in the dataset the same amount of times. However, this is impossible if $n \nmid km$ and in this case some items are recommended $\floorkm$ times while others $\floorkm+1$ times. For Nonrepeatable MostFair, for each user we generate a list of \textit{recommendable items}, defined as items in $I$ that have not appeared in their corresponding train/val split. For one user at a time, we then recommend the least popular $k$ recommendable items based on the current recommendation lists of all users.
Repeatable MostUnfair recommends the same $k$ items to each user. Nonrepeatable MostUnfair does the same, but if any of those $k$ items is a non-recommendable item, the non-recommendable item is replaced by a recommendable item. The results of this experiment for Lastfm and Ml-1m are presented in Fig.~\ref{fig:mostfair_higher_better}--\ref{fig:mostunfair_lower_better} and for the remaining datasets in App.~\ref{app:maxmin}. We discuss the findings below.

\begin{figure*}[p]
    \centering
    \includegraphics[width=\textwidth]{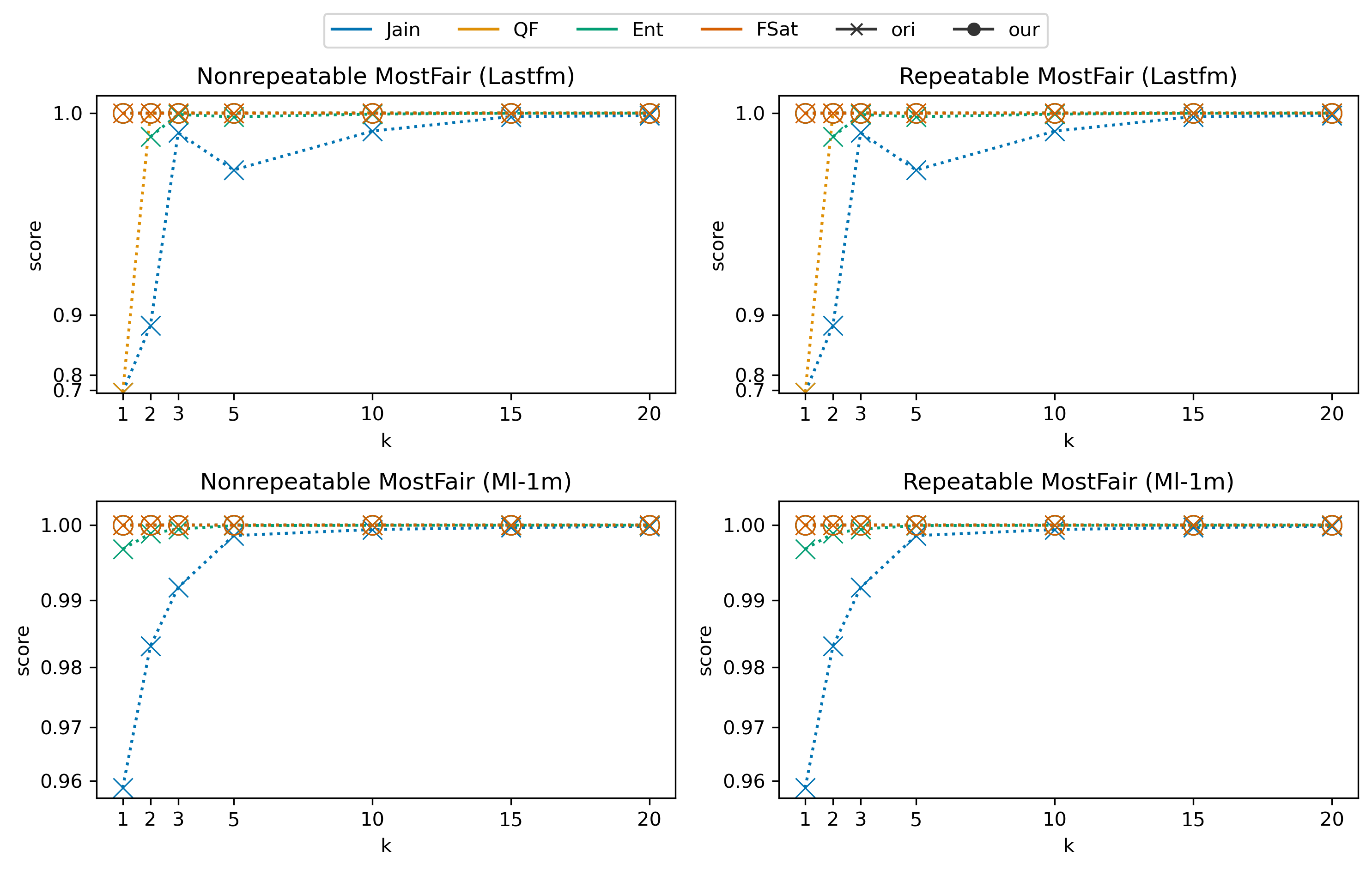}
    \caption{Most fair scores with varying $k$ for higher-is-fairer fairness measures for Lastfm and Ml-1m. All scores from the corrected measures (denoted by `our') measures overlap with each other.}
    \label{fig:mostfair_higher_better}
    \centering
    \includegraphics[width=\textwidth]{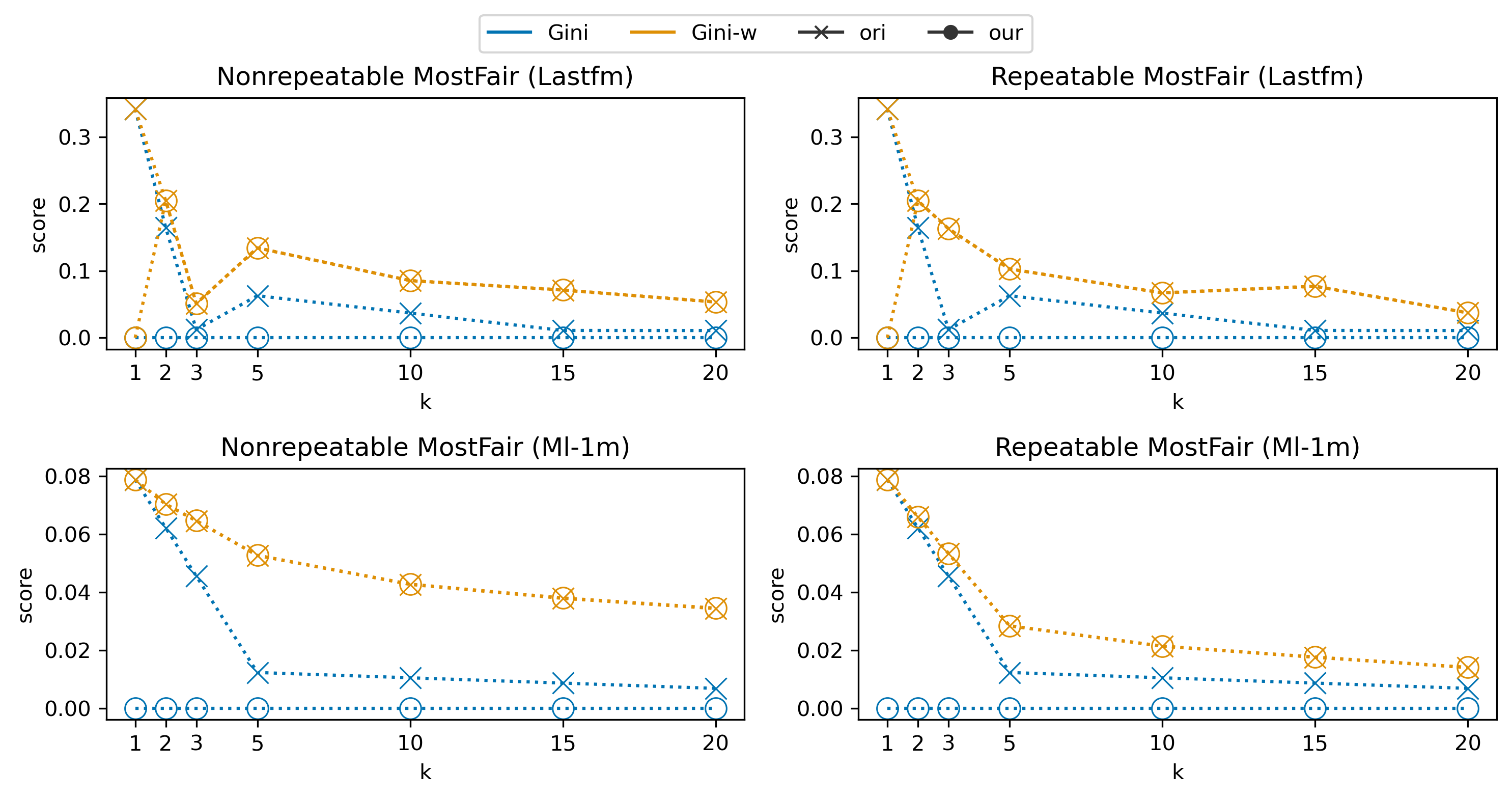}
    \caption{Most fair scores with varying $k$ for lower-is-fairer fairness measures for Lastfm and Ml-1m.}
    \label{fig:mostfair_lower_better}
\end{figure*}

\noindent\textbf{Theoretical maximum fairness}. For both nonrepeatable and repeatable settings, all original measures fail to achieve their theoretical maximum fairness values due to the \textbf{non-realisability} limitation (\textit{Causes 2--3}). The scores of the original measures get closer to the theoretical maximum fairness values as $k$ increases. 
However, these scores are still not equal to the theoretical maximum fairness value. In the original measures, having more slots due to larger $k$ does not guarantee that the scores would be higher as well. E.g., the score of \up Jain\ori~ in Fig. \ref{fig:mostfair_higher_better} is higher at $k=3$ compared to $k=5$, because of the changing values of $km \bmod n$ for different values of $k$. 
Our corrected versions always reach their theoretical maximum fairness values for both repeatability settings, except for Gini-w\our. This behaviour is due to the unresolvable non-realisability (\textit{Cause 4)} limitation for Gini-w. However, Gini-w\our~ can still reach the theoretical most fair value when $k=1$ for Lastfm (Fig.~\ref{fig:mostfair_lower_better}), while the original version fails to do so.

\begin{figure*}[p]
    \centering
    \includegraphics[width=0.97\textwidth]{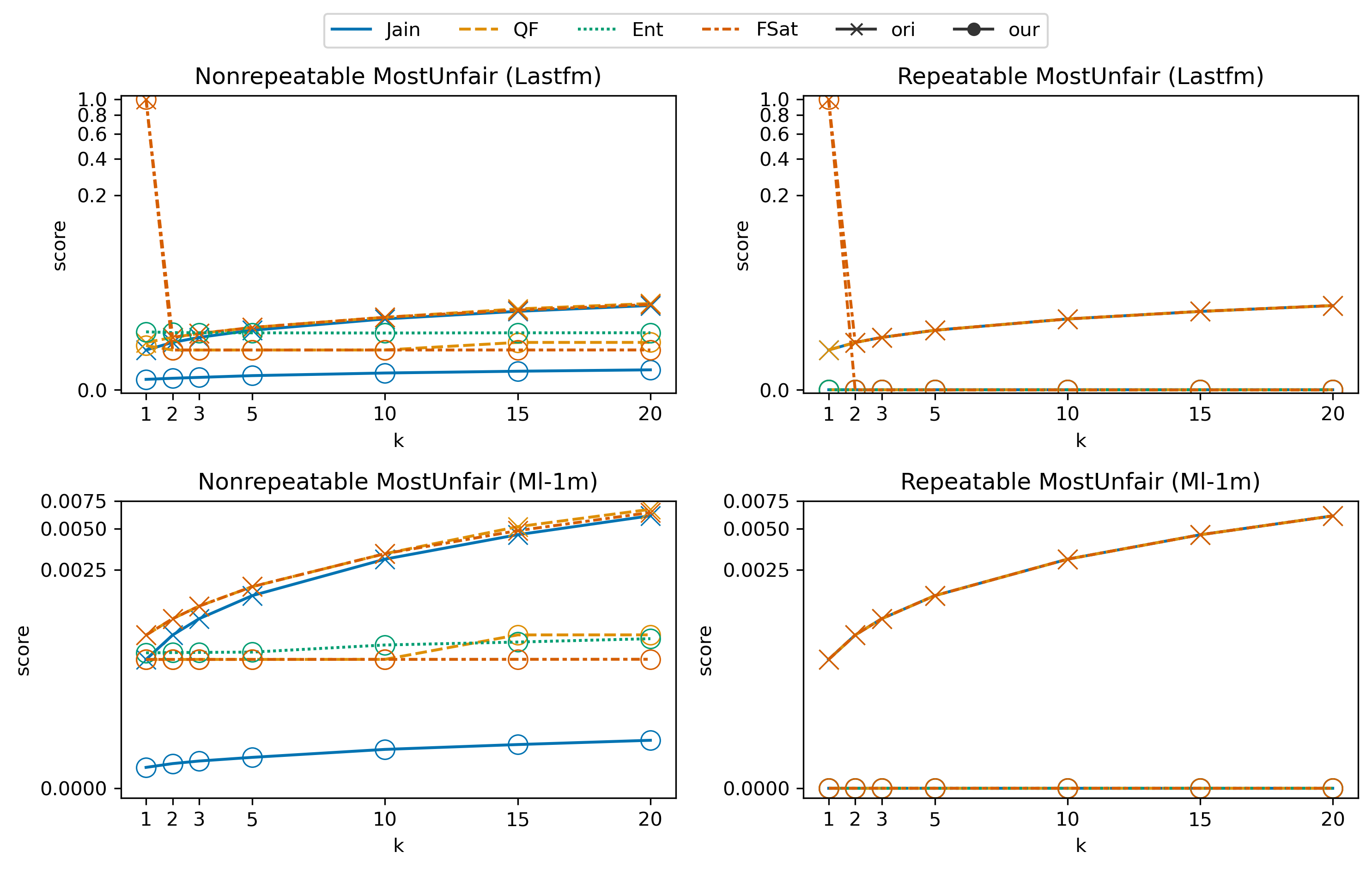}
    \caption{Most unfair scores with varying $k$ for higher-is-fairer fairness measures for Lastfm and Ml-1m. On Repeatable MostUnfair, all scores from the corrected measures (denoted by `our') overlap with each other for the shown values of $k>1$ for Lastfm and for all shown values of $k$ for Ml-1m.}
    \label{fig:mostunfair_higher_better}
    
    \centering
    \includegraphics[width=0.97\textwidth]{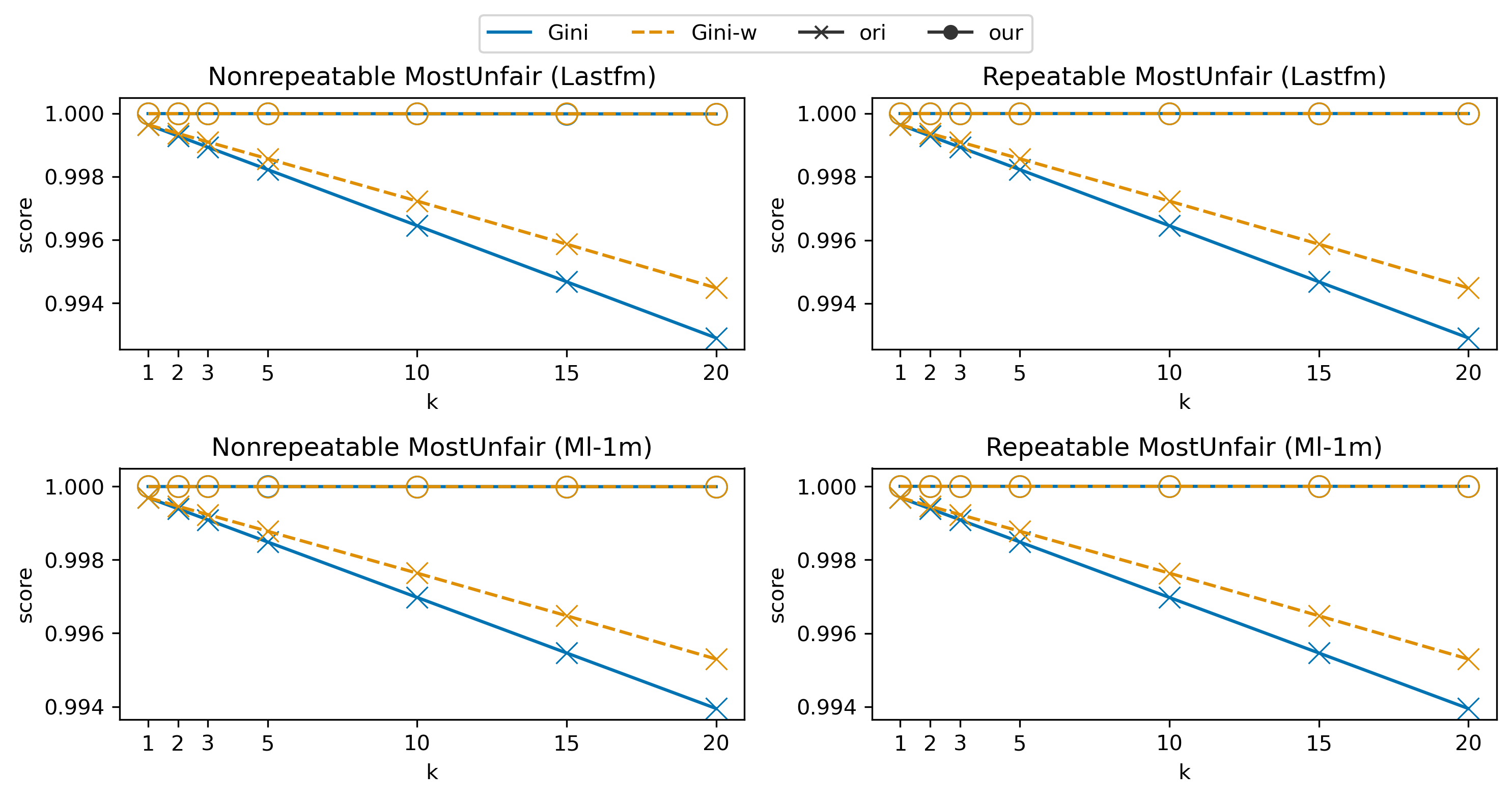}
    \caption{Most unfair scores with varying $k$ for lower-is-fairer fairness measures for Lastfm and Ml-1m. On Repeatable MostUnfair, all scores from the corrected measures (denoted by `our') overlap with each other for all shown values of $k$.}
    \label{fig:mostunfair_lower_better}
\end{figure*}

\noindent\textbf{Theoretical minimum fairness}. 
All original measures fail to reach the theoretical minimum values for all experimented values of $k$ for all settings due to \textbf{non-realisability} limitation (\textit{Cause 1}). This happens less frequently in our measures (Fig. \ref{fig:mostunfair_higher_better}). Our measures successfully achieve the theoretical minimum fair values under the repeatable settings, 
except for FSat in Lastfm (Fig. \ref{fig:mostunfair_higher_better}). This is because when $k=1$, there are not enough slots for the items and due to the \textbf{always-fair} limitation that is unresolvable ($\S$\ref{ss:nofix}), the score for \up FSat is 1, which is not the theoretical minimum fair value. Additionally, the scores of the original measures diverge from the theoretical minimum fairness value with larger $k$. This happens to our measures only in the nonrepeatable setting because the normalization is done by assuming that any item can be recommended to any users. This assumption is not true in the nonrepeatable settings because some items cannot be re-recommended to some users. 
However, differently from the original measures, the scores of our measures diverge less as $k$ increases. 

Overall, while the difference in scores between the original measures and our versions is not large, our measures quantify the actual most (un)fair situations more accurately than the original measures. The difference between the original measures and our versions in the most unfair recommendation under the repeatable setting is $\frac{k}{n}-0 = \frac{k}{n}$ for Jain, QF, Gini, and FSat; and $\log{k}$ for Ent (see Tab.~\ref{tab:bounds}). However, the difference would be greater in item-poor domains where $n$ is small and therefore possibly close to $k$, e.g. insurance \cite{BorgBruun2022LearningDomain}. 
The scores of the original measures also change with $k$. This makes their interpretation harder because the distance between the original scores and the theoretical maximum/minimum fair score also changes without an intuitive pattern for different values of $k$, as seen in the \up Jain\ori~ scores in Fig.~\ref{fig:mostfair_higher_better} which can increase or decrease as $k$ increases. Furthermore, the original measures suffer particularly for low $k$ values, which are the most important rank positions in real-life \rs. The scores of our measures rarely change with different values of $k$.

\subsection{Sliding window: relevance and fairness at different rank positions}
\label{ss:sliding}

This experiment studies how relevance and fairness scores of all measures vary at decreasing rank positions. The experiment aims to observe 1) the change in relevance scores, if any, as items should ideally be placed in the ranks according to decreasing order of true relevance; and 2) whether and how the fairness scores change across different rank positions. Due to bias in recommenders, popular items tend to be given more exposure. Thus, we expect the relevance scores to decrease and the fairness scores to become more fair at decreasing rank positions.
We study how the above changes may generally differ between relevance measures and fairness measures, as well as between different fairness measures, including the ones with different fairness notions. 

We conduct this experiment as follows. We use the runs from the BPR model, which is the best in our experiments. Given one run, we compute the measures for different sliding windows of rank positions in rankings 1--5, 2--6, and so on until 5--9. 
We reorder the recommended items such that items that were previously recommended at the top positions are now at the bottom positions when we change the window according to decreasing rank. 
The results for Lastfm and Ml-1m are presented in Fig. \ref{fig:sliding-1} and for the rest of the datasets in App. \ref{app:sliding}. 

\begin{figure*}
    \centering
    \includegraphics[width=\textwidth]{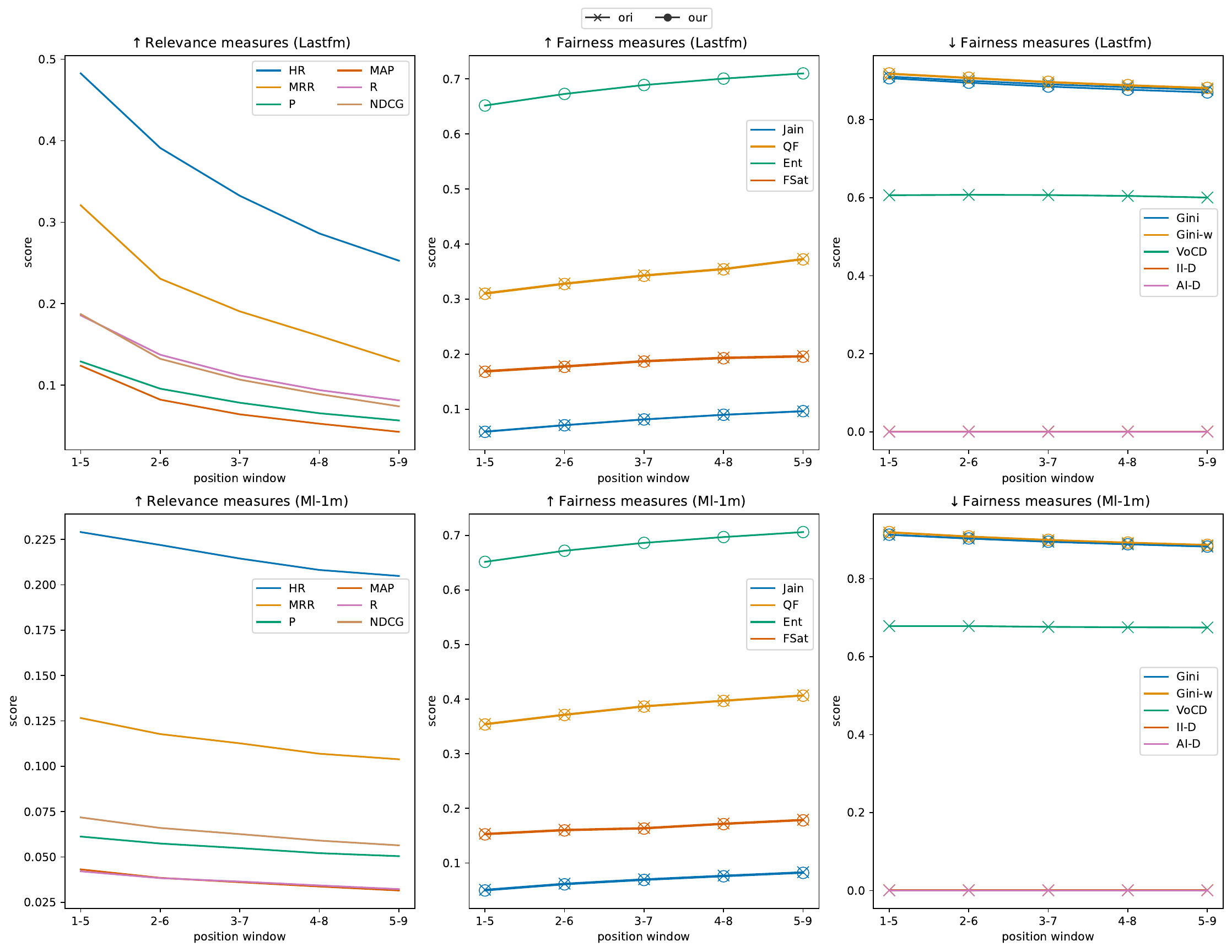}
    \caption{Sliding window evaluation for BPR model, on Lastfm and Ml-1m. Each row of figures is for one dataset, each column is for the different groups of measures (relevance, higher-is-better fairness, lower-is-better fairness measures). II-D and AI-D lines overlap.}
    \label{fig:sliding-1}
\end{figure*}

The following observations from Fig.~\ref{fig:sliding-1} apply to both the original fairness measures and our corrected versions of these measures unless otherwise stated. All relevance scores decrease as rank decreases. 
The drop of relevance scores for Ml-1m ($[0.04,0.23] \rightarrow [0.03, 0.20]$) is less extreme than in Lastfm ($[0.12,0.48] \rightarrow [0.04, 0.25]$). This is partly because the test set of Lastfm has at most five relevant items per user, while on average, Ml-1m has many more. While relevance scores decrease, 
fairness measures show that fairness slightly increases down the rank, except for \down VoCD\ori. 
The range of higher-is-better fairness measures increases from $[0.06,0.65]\rightarrow[0.10,0.71]$ for Lastfm and $[0.05,0.65]\rightarrow[0.08,0.71]$ for Ml-1m. The range of \down Gini and \down Gini-w, decreases from $[0.91,0.92]\rightarrow[0.87,0.88]$ for Lastfm and $[0.91,0.92]\rightarrow[0.88,0.89]$ for Ml-1m. 
\down VoCD\ori~ seems invariant to changes in the position window ($0.61 \rightarrow 0.60$ for Lastfm and $0.68\rightarrow0.67$ for Ml-1m). This may be because VoCD\ori~is the only measure that considers fairness exclusively for recommended items, and the recommended items differ a little in terms of the number of times they are recommended as rank decreases.
\down AI-D has even smaller changes in scores as the values are already minuscule in the first place, while \down II-D is always constantly small for a dataset. 
The small values, compared to other measures, are due to these measures quantifying fairness using different concepts from other measures, i.e. comparing exposure to random exposure (also observed and explained in $\S$\ref{ss:performance}). 
The ranges of all fairness measures are roughly the same across datasets, but the range of relevance measures varies across datasets. This also holds for the datasets in App.~\ref{app:sliding}. This may be due to the distribution of the recommended items being similar across datasets, and the distribution of the number of relevant items differing across datasets, as explained above for Lastfm and Ml-1m.

Fairness measures are also somewhat invariant to changes in relevance. This is anticipated as the equations of fairness measures are independent of relevance values.

\subsection{Measure strictness and sensitivity through artificial insertion of items}
\label{ss:insert} 
We have observed in $\S$\ref*{ss:performance} that different fairness measures vary in their strictness of quantifying fairness (e.g., some measures give scores close to the most fair values, and the opposite for others). It is however unknown how sensitive fairness measures are, given the change of the number of times an item is exposed in the recommendation list across all users. Therefore, the goal of this experiment is to study the strictness and sensitivity of the measures, and compare these aspects between measures of similar and different fairness concepts. Knowing the strictness and sensitivity of the measures matters as this affects how we interpret the scores of the measures. For example, if one uses a measure that tends to produce scores close to the most fair value, they must be aware that the score may not reflect fairness accurately.

As such, we devise an experiment to specifically study how the relevance measures, existing fairness measures and our corrected fairness measures scores change when we artificially control the fraction of jointly least exposed and relevant items in the recommendation list.  
We start with an initial recommendation list. We define a \textit{least exposed (LE) item} as an item in the dataset with the least exposure, based on the current recommendation list.\footnote{This fairness concept is closely tied to all measures in this work, except for VoCD which concerns only items in the recommendation list, as opposed to in the dataset (Tab.~\ref{tab:situations}).} An LE item in this experiment is therefore an item that has not appeared in the current recommendation list. We define a \textit{relevant item} as per the labels of relevance. 

From the initial recommendation list, we insert jointly LE and relevant items, one item at a time. 
We create a synthetic dataset with $m=1000$ users and $n=10000$ items. The number of items is exactly the number of recommendation slots $km$ for a cut-off $k=10$. 
We artificially generate a ranking of top $k$ as follows. The artificial insertion of jointly LE and relevant items begins with the recommendation of the same 10 items $i_1, i_2, \dots, i_{10}$ to all users. These items are irrelevant to each user except $u_1$, as we keep the recommendation list for $u_1$ the same throughout the experiment. This is because we want to keep the number of items exactly $km$ where theoretically each item could be recommended exactly once and if we have to completely replace all $m$ users' recommendation lists, we would need to have more than $km$ items. We expect the relevance measures to give scores close to zero on this initial recommendation list as only $u_1$ has relevant items. We expect the fairness measures to give scores that are equal to or close to the theoretical most unfair scores.\footnote{\label{fn:close_to}We say ``close to'' due to the \textbf{non-realisability} (\textit{Cause 4}) limitation in some measures.}

Let $P$ be the fraction of items in the $k$ that are artificially inserted by us. We vary from $P=0$, the original recommendation where we have not inserted any items artificially, to $P=1$ where all items in the $k$ are jointly LE and relevant items that are artificially inserted by us. We increase $P$ in steps of $\nicefrac{1}{k}$. 
From the bottom of a user's recommendation list, we replace one item at a time with a known jointly LE and relevant item, until we end with a recommendation list of different $km$ items across all users, that are all relevant only to the user to whom that item is recommended. 
At the end of the insertion process, each user is recommended exactly 10 relevant items, and those items are also fair w.r.t.~the entire recommendation list for all users, considering all items in the dataset; item fairness is not defined w.r.t.~a specific user. 
We expect the relevance measures to give scores of 1 on the final recommendation list and fairness measures to give scores that are (close to) the fairest scores.\footnoteref{fn:close_to}

\begin{figure}
    \centering
    \includegraphics[width=\textwidth]{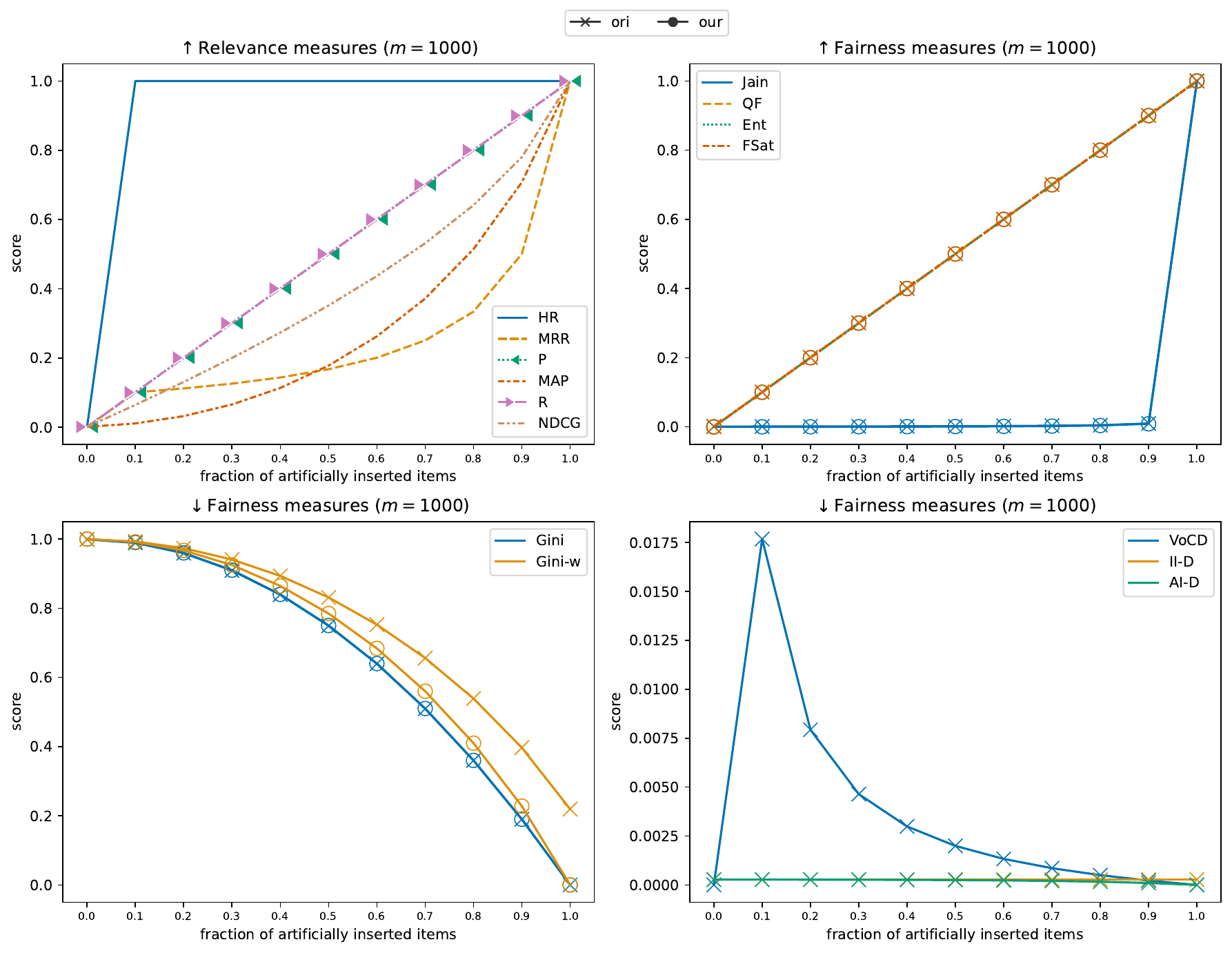}
    \caption{Results for jointly LE and relevant item insertion. All measures are at $k=10$. QF\ori~and FSat\ori~overlap. QF\our, FSat\our, and Ent\our~also overlap.}
    \label{fig:artificial-fair}
\end{figure}

The results of this experiment are presented in Fig. \ref{fig:artificial-fair}. We see that all relevance measures increase as we add more relevant items.\footnote{The relevance measures do not start from 0 when $P=0$, as there is one user with ten relevant items.} 
The following observations apply to both the original fairness measures and to our corrected versions of these measures, unless otherwise specified. All fairness measures, except VoCD\ori~ and II-D\ori, indicate more fairness as we increase $P$, but with varying sensitivity, explained next. \up Jain is one of the strictest fairness measures. Even when the proportion of LE items is 0.9 (item $i_1$ is recommended to all users, but the rest of the recommendation lists are filled with different items), the \up Jain score is still close to 0, which translates to unfair while \up QF, \up Ent, and \up FSat are 0.9, which is close to the fairest score of 1. 
The scores of QF\ori~are exactly the same as FSat\ori, because all items in the recommendation list are recommended once, which is also the maximin share (defined in $\S$\ref{ss:onlyfair}). 
\up QF\our, \up Ent\our, and \up FSat\our~also give identical scores. This is expected as the increase in the scores is constant and proportional to the fraction of artificially inserted LE items, yet this is interesting as these three measures are based on three different fairness notions (QF being insensitive to the number of times an item is recommended, and FSat being based on maximin-shared fairness). 
 
Meanwhile, the increase of fairness in \down Gini and \down Gini-w follows a non-linear trend, with \down Gini-w being stricter than \down Gini. The non-linear trend is also expected as Gini and Gini-w are based on the Lorenz curve, a graphical representation of the cumulative proportion of exposure to the cumulative proportion of items. We also see that Gini-w\our~ is able to reach the theoretical most fair when the entire recommendation list consists of artificially inserted LE items, while Gini-w\ori~ fails. 
\down VoCD\ori~ is insensitive to the insertion of items as it only considers fairness for recommended items. The number of times these items are recommended across all users does not differ much in this set-up, therefore \down VoCD\ori~ returns scores that are close to the fairest.
Most notably, \down II-D\ori~ and \down AI-D\ori~ are very close to 0 (on the scale of $10^{-3}$ or even smaller) even when the same $k$ items are recommended to all users. \down II-D \ori~ remains constant, while \down AI-D\ori~ is rather insensitive to the addition of LE and relevant items. The small scores are due to the measures quantifying fairness according to the closeness of item exposure with random exposure, while other measures have no such comparisons. Therefore, for these measures, the scales are not very meaningful, even though for AI-D, the scores still indicate improvement as we insert more LE items. 
    
We see similar trends with $m \in \{100, 500\}$, but the change of the scores is most stable with $m=1000$. As we increase the number of users (and items), the range of VoCD, II-D, and AI-D scores also becomes more compressed. In contrast, the range of the other measures remains similar. We also observe a similar but opposite trend of results when we artificially insert known irrelevant and multiple copies of items already in the recommendation list. Both of these results are in App.~\ref{app:insert}.

Overall, the artificial insertion experiment indicates that several measures respond linearly to the insertion of LE items 
i.e., \up QF, \up FSat, and \up Ent, while the rest do so non-linearly. 
This can affect the interpretation of these scores, as we observe that it is generally harder to achieve a high fairness score in some measures. In some other measures, it is also easier to improve fairness when starting from a relatively fair situation, but much harder when starting from a completely unfair situation. 

\section{Related work}\label{s:prevwork}

Prior work \cite{Do2021Two-sidedDominance, Wang2022ProvidingSystems,Zhu2020FARM:APPs,Wu2022JointRecommendation,Patro2020FairRec:Platforms,Mansoury2020FairMatch:Systems,Mansoury2021ASystems} proposes exposure-based individual item fairness measures but does not provide a comprehensive analysis of the limitations of the measures. Our work differs from this because we extensively analyse individual item fairness measures specifically for \rs, identify novel and previously-known limitations in them, and address these limitations. Meanwhile, several other work uses individual item fairness measures that also takes into account the relevance of the item to users  
\cite{Morik2020ControllingLearning-to-Rank, Borges2019EnhancingAutoencoders,Saito2022FairRanking,Wu2022JointRecommendation,Zhu2021FairnessSystems}. The investigation of these measures (of fairness and relevance) is reserved for our future work.

\citeauthor{Amigo2023ASystems}~\cite{Amigo2023ASystems} overview fairness measures in \rs~ and characterise them according to five dimensions. 
They focus on generalising the measures into  broad categories and studying the relationship between multi-stakeholder fairness, whereas we analyse each individual item fairness measure both theoretically and empirically. 
We also focus on the relationship among individual item fairness measures. 

\citeauthor{Raj2022MeasuringResults}~\cite{Raj2022MeasuringResults} analyse fairness measures for provider-side group fairness in ranked outputs. They include one measure for individual item fairness, II-D (referred to as EED in~\cite{Raj2022MeasuringResults}, which was originally proposed by \cite{Diaz2020EvaluatingExposure}), 
but only analyse it as a group fairness measure. 
They list some questions to assess the design of fairness measures, which we exploit in this work, e.g., the examination function used (Tab. \ref{tab:exp-weigh}), the ideal/fair criteria (Tab. \ref{tab:situations}), and if the measure incorporates relevance ($\S$\ref{ss:onlyfair}). 
Both our work and theirs identify the limitations in the measures e.g., edge cases where zero values cause undefinedness in the measure computation ($\S$\ref{ss:undefinedness}). 
They propose to use a small constant to avoid computing $\log{0}$ in group fairness metrics, but we do not use the same approach as using a small constant can introduce noise in the measure computation. 

\citeauthor{Majumder2021FairFairness}~\cite{Majumder2021FairFairness} examine classification measures for both individual and group fairness through empirical analysis and cluster the measures based on correlation values. They analyse the correlation among fairness measures but do not investigate the 
limitations of those measures. They find disagreements between the measures when labelling a model as fair or unfair. We did not do this mapping, as this may lead to loss of valuable information regarding the range and the actual values of the measures. 
However, we show that disagreements also exist between several individual item fairness measures in \rs.

All previous work finds that several fairness measures 
are highly correlated, and are insensitive to changes in data  
\cite{Amigo2023ASystems,Majumder2021FairFairness,Raj2022MeasuringResults}. Our analysis in $\S$\ref{s:exp} confirms these findings in a different experimental set-up and sheds light onto additional limitations which have not been reported or corrected previously. 
Next, we provide practical guidelines for choosing among individual item fairness measures.

\section{Discussion}
\label{s:discussion}
\subsection{Summary of theoretical corrections}
In this paper, we critically analyse individual item fairness measures in \rs~w.r.t.~their limitations. We point out a total of five theoretical limitations in the measures, and identify that each measure suffers from two or more limitations. Some limitations are due to intentional design choices of the measure, while the remaining limitations go against some pre-defined desirable properties of (fairness) evaluation measures. 

We posit that evaluation measures of fairness should have two important utilities: 
1) to assess systems/models in isolation (i.e., evaluating `how fair' a single system is, with one endpoint being the most unfair and the other being the most fair); 
2) to compare different \rs~and make a decision about whether and to what extent system A is more fair than system B, based on the measure.\footnote{Note that utility 1) is for assessment purposes, not a goal for model development; it may not be possible for relevant recommendations to be maximally fair at the same time.} 
A measure should ideally be usable for both use cases. At the present stage, none of the individual item fairness measures is suitable for the first use case, hence the need to modify the current measures for it.\footnote{This is unlike relevance measures, where there is still room for choice, as several measures have reachable endpoints (e.g., NDCG, RR@$k$) \cite{Moffat2013SevenMetrics}.} 

Note that having limitation(s) does not mean that the measures are completely unusable. Some measures can still be used despite having limitations, as long as one is aware of these limitations. For instance, in the case of measures that empirically cannot reach the endpoints of $[0,1]$, it is still possible to use the measure to compare fairness between two or more systems. 
Yet, evaluating fairness for a single system using such a measure is a challenge, as having a score of, for example, 0.6 does not always mean that there is still 40\% room for improvement. This point should be kept in mind, especially given the common practice of interpreting a single score in comparison to the known range of that measure. 

We also provide theoretical solutions to address the three resolvable limitations, and we argue why the remaining limitations cannot be resolved. The first set of solutions guarantees that the measures range in a bounded interval, e.g., $[0, 1]$, where both the theoretical minimum and maximum scores are achievable, with one endpoint corresponding to the most unfair recommendation list, and the other to the fairest recommendation list. This set of solutions considers the number of recommendation slots and the number of items in the dataset, and at the same time assures that the measures are well-defined for both common and edge cases.
Our second solution ensures that the affected measure is sensitive to the change of exposure received by an item, thereby fulfilling a desired property of fairness measure. 

\subsection{Summary of empirical findings} 
Extensive empirical experiments were conducted to compute relevance and fairness scores for both the original measures and for our corrected versions of these measures. 
The experiments utilised six datasets and seven recommendation models, including state-of-the-art models and well-established baseline models. 
Even though the models are all trained to optimise for relevance, we discover that the fairest model is not necessarily the worst in terms of relevance scores. This was unexpected as the fairness measures used in this work are detached from relevance. However, we also see the common observation where several models have higher relevance scores, but exhibit lower fairness. 

Our results empirically show that relevance measures and fairness measures have different ranges, which makes the interpretation of the fairness measures difficult. Further, the range of fairness measures is incomparable between different measures, and for some measures, this range is also not lower/upper-bounded empirically. Other noticeable observations include some fairness measures that tend to score much lower/higher compared to other measures, as well as uncorrected measures that are incomputable or produce constant values, given any recommendation list based on the same dataset. While the actual scores of the measures may differ, we found that most fairness measures have a high and significant agreement in ranking the recommenders from the most to least fair. The strong agreement is observed between the corrected and uncorrected measures, and even between some measures that are based on different fairness concepts.
Altogether, our results show that:
\begin{itemize}
    \item Despite limitations in quantifying the extent of fairness, the measures agreed in the ordering of models according to their individual item fairness. 
    \item Some original measures do not reflect the absolute quantity or differences in fairness.
    \item The corrected measures are required to reliably use the measures for scenarios that may contain commonly occurring and edge cases.
\end{itemize}

\subsection{Guidelines of the appropriate use of the fairness measures}
Next, we summarise guidelines on using these fairness measures based on the above theoretical and empirical findings. 

\noindent\textbf{Use original fairness measures only to evaluate relative fairness}. 
\textit{Relative fairness} refers to comparing the relative ordering of fairness scores. 
The original fairness measures suffer from theoretical limitations that limit their usage in settings where data distributions or recommendation scenarios do not fulfil the theoretical premises of the original measures. Moreover, the original measures may be more difficult to interpret as their range and scaling do not always match the intuitive expectations of being between 0 and 1. While the original measures as proposed outside recommendation can be used as their ranges are known and can be easily interpreted, we advise using them to evaluate only relative fairness. 

\noindent\textbf{Use our corrected fairness measures to evaluate absolute fairness}.
\textit{Absolute fairness} refers to measuring how close a model's recommendation is to the most (un)fair recommendation scenario. To evaluate absolute fairness, we recommend using our corrected measures for Jain, QF, Ent, Gini, and FSat. Our fairness measures are always perfectly correlated with the original measures, thus providing results that align with the original measures. Further, our measures are well-defined and have better interpretability w.r.t.~how the minimum/maximum scores correspond to the most unfair/fair recommendation scenario, as shown in $\S$\ref{ss:maxmin}. Our fairness measures are highly correlated with each other, but because they operate on different scales, one should not deduce that a model is (un)fair based on the absolute fairness measurement scores. 

Note that both FSat\ori~ and FSat\our~ should never be used when $km<n$ due to the unresolvable \textbf{always-fair} limitation, as the score will always be perfectly fair regardless of the recommendation. \down Gini-w\our~ should preferably be used when one has equal or more items than recommendation slots ($km\leq n$), as the measure works ideally in that setting: a score of 0 means the recommendation is perfectly fair, while a score of 1 means the unfairest possible recommendation. For the remaining cases, which commonly happens in many public recommendation datasets for any cut-off $k$ (Tab.~\ref{tab:dataset}), even if the most unfair recommendation entails a score of 1, the most fair recommendation is not mapped to a score of 0 in Gini-w\our. Yet, this is still better than Gini-w\ori~ which maps unrealistic scenarios to 0 and 1. For Ent, we recommend using our correction, since our correction avoids the \textbf{undefinedness} limitation, and would produce the same score as Ent\ori~ when all items are recommended.  

We discourage using the rest of the measures due to their tendency to have scores that are not representative of fairness, e.g. scale mismatch between II-D/AI-D and the rest of the measures. Additionally, II-D should not be used for single-round recommendations as the scores are always constant.

\section{Conclusions}

We have presented a novel investigation into the theoretical and empirical limitations of current evaluation measures of individual item fairness in recommender systems. 
We have further amended these measures to correct their limitations or have argued why some limitations are impossible to resolve. Extensive experiments on real-life and synthetic data reveal novel insights on how individual item fairness measures should and should not be used. 

In the present work, we solely concentrated on measures that quantify individual item fairness independently of recommendation performance. We reserve the analysis of fairness measures that are tied to relevance for our future work. Future work should investigate whether measures that aim to simultaneously quantify both recommendation performance, or relevance, and fairness suffer from similar limitations and empirical behaviours than the measures studied here. Other future work could further explore the relationship between individual item fairness and item group fairness \cite{Wu2022JointRecommendation}, or fairness between users and items \cite{Amigo2023ASystems}. 
All measures studied here also assume that exposure is the key factor for fairness, while there might be other factors to consider for fairness, e.g. speed of the recommendation or the wait time from when an item is introduced to a system until it gets recommended. 

The empirical studies could also be extended to account for the behaviour and performance of the measures using additional datasets. However, while we cannot exclude the possibility that the experiments on other settings, domains, or datasets could lead to new insights, it is unlikely that they would affect our conclusions and guidelines on the appropriate use of the studied measures. Future work could also utilize our corrected measures to optimize recommendation models for fairness to reveal whether the corrected measures could improve recommendation models as opposed to only measuring the performance of existing models. 

\section*{Code Availability}
Our source code (in Python 3.10) is publicly available on \url{https://github.com/theresiavr/individual-item-fairness-measures-recsys} and usable under the MIT License, with proper attribution to this work and possibly other related work. Other restrictions regarding the usability of the code may apply for the RecBole library \cite{Zhao2021RecBole:Algorithms}.

\begin{acks}
The work is supported by the Algorithms, Data, and Democracy project (ADD-project), funded by Villum Foundation and Velux Foundation, as well as the Academy of Finland. 
We also thank the anonymous reviewers who have provided insightful comments and suggestions to improve earlier versions of the manuscript.
\end{acks}

\printbibliography
\appendix

\section{Mathematical Workings for Bounds in Tab. \ref{tab:bounds}}
\label{app:boundsproof}

We provide the derivation of the obtained min/max achievable value of Jain, QF, Ent, Gini, FSat, and VoCD. The min/max achievable values are obtained from the most (un)fair recommendation scenarios described in $\S$\ref{ss:non_realisability}.

\subsection{Jain's Index}
The most unfair case for \up Jain produces Jain$_{\min}$ and the most fair case for \up Jain produces Jain$_{\max}$.
\begin{equation*}
\text{Jain}_{\min} = \frac{(km)^2}{n \sum\limits_{i\in I} \left[\sum\limits_{u\in U}
1_{R_{u}^{k}}(i)\right]^2} =  \frac{(km)^2}{n(km^2)} = \frac{k}{n}
\end{equation*}
\begin{align*}
\text{Jain}_{\max} 
&= \frac{(km)^2}{n \sum\limits_{i\in I} \left[\sum\limits_{u\in U}
1_{R_{u}^{k}}(i)\right]^2} 
=  \frac{(km)^2}{
n\left((n-km \bmod n)\floorkm^2 + (km \bmod n)\left(\floorkm+1\right)^2\right)} \\
&= \frac{(km)^2}{
n\left(n\floorkm^2-(km\bmod{n})\floorkm^2 + (km\bmod{n})\left(\floorkm^2+2\floorkm+1\right)\right)
} \\
&= \frac{(km)^2}{
         n\left(n\floorkm^2 
         + (km\bmod{n})\left(2 \floorkm+1\right) 
         \right)}   
\end{align*}

\subsection{Qualification Fairness}
The most unfair case for \up QF produces QF$_{\min}$ and the most fair case for \up QF produces QF$_{\max}$.
\begin{align*}
    \text{QF}_{\min} 
    = \frac{(k\cdot1)^2}{n(k \cdot 1^2)} = \frac{k}{n}
\end{align*}
When there are not enough recommendation slots for all items, $km < n$:
\begin{align*}
    \text{QF}_{\max}
    = \frac{(km\cdot 1)^2}{nkm(1^2)} = \frac{km}{n}
\end{align*}
When $km \geq n$, all items can be recommended, hence  QF$_{\max}=\frac{n}{n}=1$.

\subsection{Entropy}
The most unfair case for \up Ent produces Ent$_{\min}$ and the most fair case for \up Ent produces Ent$_{\max}$. The first term in Ent$_{\max}$ comes from $n-km\bmod{n}$ that are each recommended $\floorkm$ times and the second comes from $km\bmod{n}$ items that are each recommended $\floorkm + 1$ times.
\begin{equation*}
\text{Ent}_{\min} 
     = -k \cdot \frac{1}{k} \log{\frac{1}{k}} 
     = - \log{k^{-1}}          
     = \log{k}
\end{equation*}
\begin{equation*}
\text{Ent}_{\max} = 
-
(n - km \bmod n)\left(\frac{\floorkm}{km}\log{\frac{\floorkm}{km}}\right)
- (km \bmod n)\left(\frac{\floorkm+1}{km}\log{\frac{\floorkm+1}{km}}\right) 
\end{equation*}

\subsection{Gini Index}
The most unfair case for \down Gini produces Gini$_{\max}$ and the most fair case for \down Gini produces Gini$_{\min}$.

\begin{align*}
    \text{Gini}_{\max}  
        =& \frac{\sum\limits_{j=n-k+1}^n(2j-n-1)m}{nk(m)}
        = \frac{\sum\limits_{n-k+1}^n{2j} - n\sum\limits_{n-k+1}^n{1} - \sum\limits_{n-k+1}^n{1}}{nk} \\
        =& \frac{2\sum\limits_{n-k+1}^n{j} - n(n-(n-k+1)+1) - (n-(n-k+1)+1)}{nk} \\
        =& \frac{2\frac{(n-k+1+n)(n-(n-k+1)+1)}{2}-nk-k}{nk} \\
        =& \frac{(2n-k+1)(k)-nk-k}{nk}
        = \frac{(k)(2n-km+1-n-1)}{nk} \\
        =& \frac{n-k}{n} = 1- \frac{k}{n}
\end{align*}

To derive Gini$_{\min}$, we use the pairwise difference formula of Gini:

\begin{equation*}
    \text{Gini} = \frac{\sum\limits_{(i,i')}{
    \left|
    \sum\limits_{u\in U}{1_{R_{u}^{k}}(i)}
    -\sum\limits_{u\in U}{1_{R_{u}^{k}}(i')}
    \right|
    }}
    {2n^2{\bar{x}}}
\end{equation*}

where $\bar{x}$ is the average number of times an item is recommended, for the most fair case, calculated as follows:

\begin{equation*}
\bar{x} = \frac{(n - km \bmod n)\floorkm + (km \bmod n)(\floorkm+1)}{n}
= \frac{n\floorkm +km \bmod n }{n}
\end{equation*}

We simplify the numerator and denominator separately, for clarity in the proof. 
To simplify the numerator, in the most fair case, there are $2(n-km \bmod n)(km \bmod n)$ pairs of items with an absolute difference of $\left(\floorkm +1\right) - \floorkm=1$, which is the difference of the number of times the items are recommended. The rest of the pairs have 0 differences. Hence, the numerator is $2(n-km \bmod n)(km \bmod n)$. Putting everything together:
\begin{align*}
    \text{Gini}_{\min} 
    &=  \frac{2(n-km \bmod n)(km \bmod n)}{
    2n^2\left(\frac{n\floorkm +km \bmod n }{n}\right)
    }
    = \frac{(n-km \bmod n)(km \bmod n)}{
    n\left(n\floorkm +km \bmod n\right)
    } 
    = \frac{(n-km \bmod n)(km \bmod n)}{
    n\left(n\floorkm +km -n\floorkm\right)
    } \\
    &= \frac{(n-km \bmod n)(km \bmod n)}{
    kmn
    } 
\end{align*}

The most unfair case for \down Gini-w produces Gini-w$_{\max}$ and the most fair case for \down Gini produces Gini-w$_{\min}$. 

To obtain Gini-w$_{\max}$, we derive the numerator and denominator separately using Eq.~\ref{eq:gini-ori}. We first compute the numerator of Gini-w$_{\max}$, considering that the items with the least to the most exposure are as follows: the first $n-k$ items are with zero exposure (as they are not present in the top $k$), one item is exposed $m$ times at position $k$, one item is exposed $m$ times at position $k-1$, and so on, until one last item that is exposed $m$ times at the top of the recommendation list. The exposure received by those items respectively are $0, m\log_{k+1}{2}, m\log_{k}{2}, \dots, m\log_{2}{2}$. Therefore, the numerator of Gini-w$_{\max}$ can be written as $m \sum\limits_{\ell=1}^k{(n-2\ell+1) \log_{\ell+1}{2}}$. Meanwhile, the denominator of Gini-w$_{\max}$ is simply $n$ times the total exposure received by the items: $mn \sum\limits_{\ell=1}^k{\log_{\ell+1}{2}}$. Putting the numerator and denominator together:
\begin{equation*}
 \text{Gini-w}_{\max} = 
 \frac{
 m\sum\limits_{\ell=1}^k{(n-2\ell+1) \log_{\ell+1}{2}}}
    {mn\sum\limits_{\ell=1}^k{\log_{\ell+1}{2}}}
=  \frac{
 \sum\limits_{\ell=1}^k{(n-2\ell+1) \log_{\ell+1}{2}}}
    {n\sum\limits_{\ell=1}^k{\log_{\ell+1}{2}}}
\end{equation*}

To obtain Gini-w$_{\min}$, we also derive the numerator and denominator separately using Eq.~\ref{eq:gini-ori}. Note that for Gini-w$_{\min}$ we only consider cases where $km\leq n$ due to the unresolvable limitation of non-realisability, \textit{Cause 4} ($\S$\ref{ss:nofix}). First, we explain how to obtain the numerator. With the restriction of $km\leq n$, to make the recommendation the fairest, the $km$ items that are recommended must be unique, leaving $n-km$ items exposed. Thus, the items with the least to the most exposure are as follows: the first $n-km$ items receive zero exposure, the next $m$ items will be recommended once each at position $k$, another set of $m$ items each at position $k-1$, and so on until the last set of $m$ items that will each be recommended at the top of the recommendation list. The numerator of  Gini-w$_{\min}$ can then be calculated as follows:

\begin{align*}
    &\sum\limits_{j=(n-km)+1}^{n-km+m} (2j-n-1) \log_{k + 1}{2}
    +
    \sum\limits_{j=(n-km+m)+1}^{n-km+2m} (2j-n-1) \log_{k}{2}
    +
    \dots
    +
    \sum\limits_{j=(n-m)+1}^{n} (2j-n-1) \log_{2}{2} \\
    &=\sum\limits_{\ell=1}^k 
    \sum\limits_{j=n-\ell m+1}^{n- \ell m + m} (2j-n-1) \log_{\ell + 1}{2}
\end{align*}

As for the denominator, it is obtained the same way as in Gini-w$_{\max}$, resulting in  $mn \sum\limits_{\ell=1}^k{\log_{\ell+1}{2}}$ as the total exposure received by the items remains the same for the same cut-off $k$ and the number of user $m$. Putting the numerator and denominator together:

\begin{equation*}
 \text{Gini-w}_{\min} =
  \frac{
    \sum\limits_{\ell=1}^k 
    \sum\limits_{j=n-\ell m+1}^{n- \ell m + m} (2j-n-1) \log_{\ell + 1}{2}}
    {mn\sum\limits_{\ell=1}^k{\log_{\ell+1}{2}}}
\end{equation*}

\subsection{Fraction of Satisfied Items}
The most unfair case for \up FSat produces FSat$_{\min}$ and the most fair case for \up FSat produces FSat$_{\max}$. Note that $k \leq n$ thus $1 \geq \frac{k}{n} \Leftrightarrow m \geq \frac{km}{n}\geq\floorkm$.
\begin{align*}
    \text{FSat}_{\min} 
    =  \frac{1}{n}
    \left(k \cdot \delta\left(
    m \geq \floorkm
    \right)\right) 
    = \frac{k}{n}
\end{align*}
\begin{align*}
    \text{FSat}_{\max} 
    &=  \frac{1}{n}\left[ 
    (n-km \bmod n)\cdot \delta\left(\floorkm\geq\floorkm\right)
    +
    (km \bmod n)\cdot \delta\left(\floorkm+1\geq\floorkm\right)
    \right] \\
    &=\frac{1}{n}[ 
    (n-km \bmod n) + km \bmod n]
    =\frac{n}{n}
    =1
\end{align*}

\subsection{Proofs for the maximum value of VoCD}
\label{app:vocd}
First, we prove in Theorem~\ref{th:max-vocd} that  when there is only one pair of similar items, the maximum VoCD value can be obtained when the two items are recommended $1$ time and $m$ times each. We then use Theorem~\ref{th:max-vocd} to show that when there are more than one pair of similar items, the maximum VoCD value does not increase (Theorem~\ref{th:vocd-not-increase}).

\begin{theorem}
\label{th:max-vocd}
If there is only one pair of $(i,i')\in A$, $\text{VoCD}_{\max}$ is obtained when  $\sum\limits_{u\in U}1_{R_{u}^{k}}(i)=1$ and $\sum\limits_{u\in U}1_{R_{u}^{k}}(i')=m$
\end{theorem}
\begin{proof}
We prove by contradiction: if an item $i$ is recommended, $1 \leq \sum\limits_{u\in U}1_{R_{u}^{k}}(i) \leq m$. 
Suppose $\exists x_1, x_2$ where $1<x_1<x_2<m$ such that 
$\frac{m-1}{m} < \frac{x_2-x_1}{x_2} \Leftrightarrow
\frac{-1}{m} < \frac{-x_1}{x_2} \Leftrightarrow
\frac{1}{m} >\frac{x_1}{x_2} \Leftrightarrow
m  < \frac{x_2}{x_1}$. 
However, $x_2 < m \Leftrightarrow \frac{x_2}{x_1} < \frac{m}{x_1} < m$, which contradicts the previous inequality, so it must be $x_1 = 1$ and $x_2 = m$.
\end{proof}
\begin{theorem}
\label{th:vocd-not-increase}
The max VoCD score does not increase with $|A|$.
\end{theorem}
\begin{proof}
Case 1: for each item pair in $A$ that is disjoint from the other item pairs, e.g., $(i_1,i_2)$ and $(i_3,i_4)$, by Theorem~\ref{th:max-vocd}, the maximum score for each pair and hence the average score of those pairs, is still $\frac{m-1}{m} - \beta$. 
Case 2: suppose the pairs are not disjoint, e.g., $(i_1, i_3), (i_2, i_3)$, and $i_1, i_2, i_3$ are recommended $x_1, x_2, x_3$ times respectively, 
where $1\leq x_1\leq x_2 \leq x_3 \leq m$. We show that for Case 2, the maximum value does not increase.

Note that $x_3 \leq m \Leftrightarrow 
-\frac{1}{x_3} \leq -\frac{1}{m} \Leftrightarrow
-\frac{x_1}{2x_3} \leq -\frac{x_1}{2m} \leq -\frac{1}{2m}$ and $-\frac{x_2}{2x_3} \leq -\frac{1}{2m}$. Thus,
$\frac{1}{2}
    \left(
        \frac{x_3-x_1}{x_3}+\frac{x_3-x_2}{x_3}
    \right)
    = 1 - \frac{x_1}{2x_3}-\frac{x_2}{2x_3} \leq 1-\frac{1}{m}$.
\end{proof}

\section{Extended results of experiments}
\label{app:extend-result}

\subsection{Experimental set-up}
\label{app:extend-exp_set_up}

We first report the hyperparameter search space for each recommender (Tab.~\ref{tab:searchhyper}) and the optimal hyperparameters for each model and each dataset (Tab.~\ref{tab:besthyper}).

\begin{table}[tb]
    \centering
    \caption{Hyperparameter search space for each recommender}
    \label{tab:searchhyper}
    \scalebox{1}{\input{tab/searchhyper.tex}}
\end{table}

\begin{table}[tb]
    \centering
    \caption{Optimal hyperparameters for each recommender and each dataset}
    \label{tab:besthyper}
    \scalebox{0.75}{\input{tab/besthyper.tex}}
\end{table}

\subsection{Analysis of relevance and fairness}
\label{app:extend-performance}

We present the performance scores of the recommender systems on the Amazon-* and Book-x datasets in Tab.~\ref{tab:base-extra1} \& \ref{tab:base-extra2}. The scores of the original version of Ent cannot be calculated due to zero divisions error for the same reasons explained in $\S$\ref{ss:performance}. The constant scores of II-D have also been explained in the same section. The best relevance and fairness scores are bolded.

BPR generally performs the best in relevance, with the exception of Amazon-lb where NeuMF is best, while ItemKNN gives the best fairness scores. Other trends observed on Amazon-* and Book-x are similar to that on Lastfm and Ml-1m.

\input{tabsplit1/base-extra1.tex}

\input{tabsplit1/base-extra2.tex}

\subsection{Correlation between measures}
\label{app:corr}

We show the Kendall's Tau values between relevance measures and fairness measures for the Amazon-* and Book-x datasets in Fig.~\ref{fig:corr-lb}--\ref{fig:corr-dm}.

\begin{figure*}[tb]
\centering
    \includegraphics[width=\textwidth]{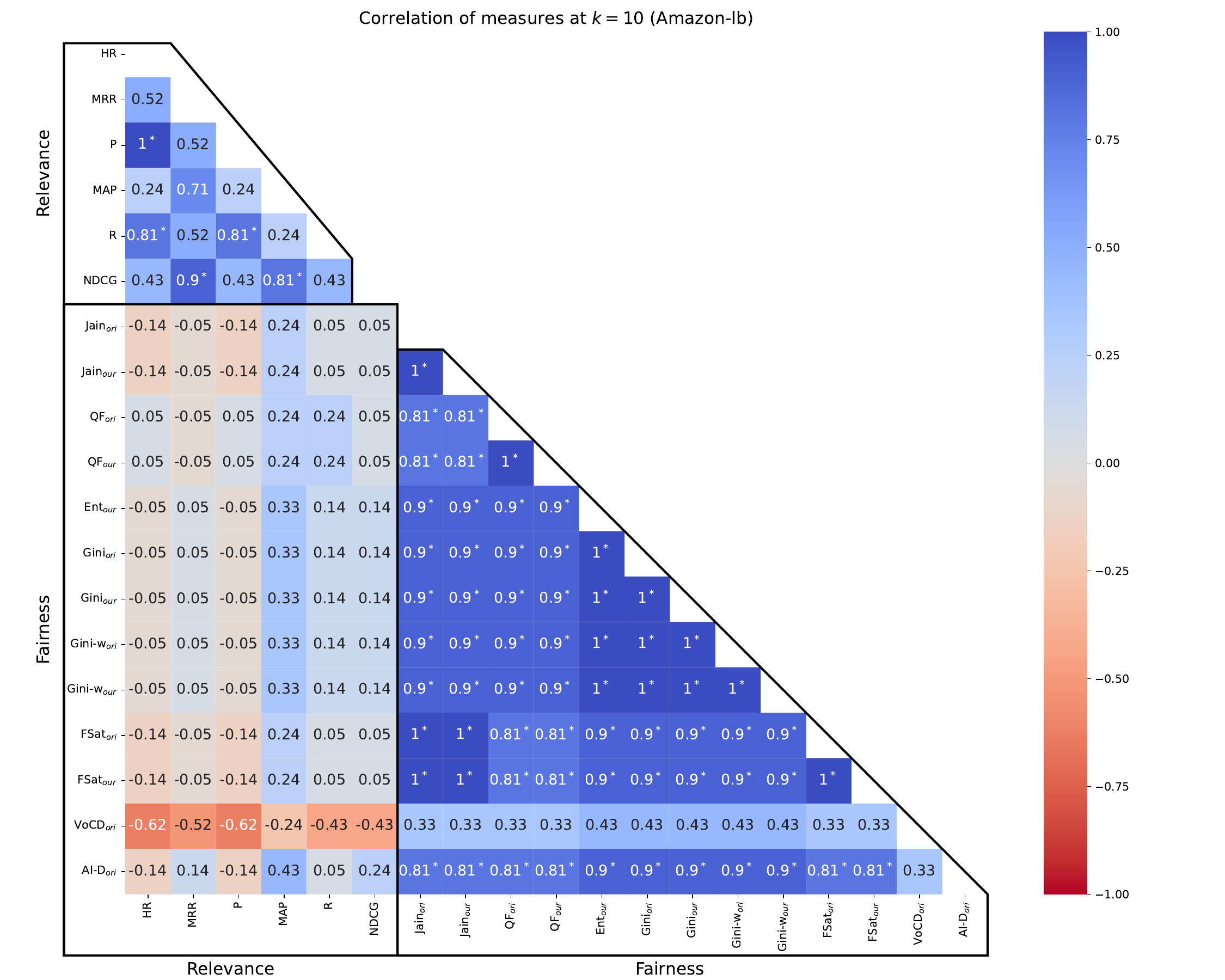}
\caption{Correlation (Kendall's $\tau$) between relevance and fairness measures for Amazon-lb. \explainsig}
\label{fig:corr-lb}
\end{figure*}
\begin{figure*}[tb]
\centering
    \includegraphics[width=\textwidth]{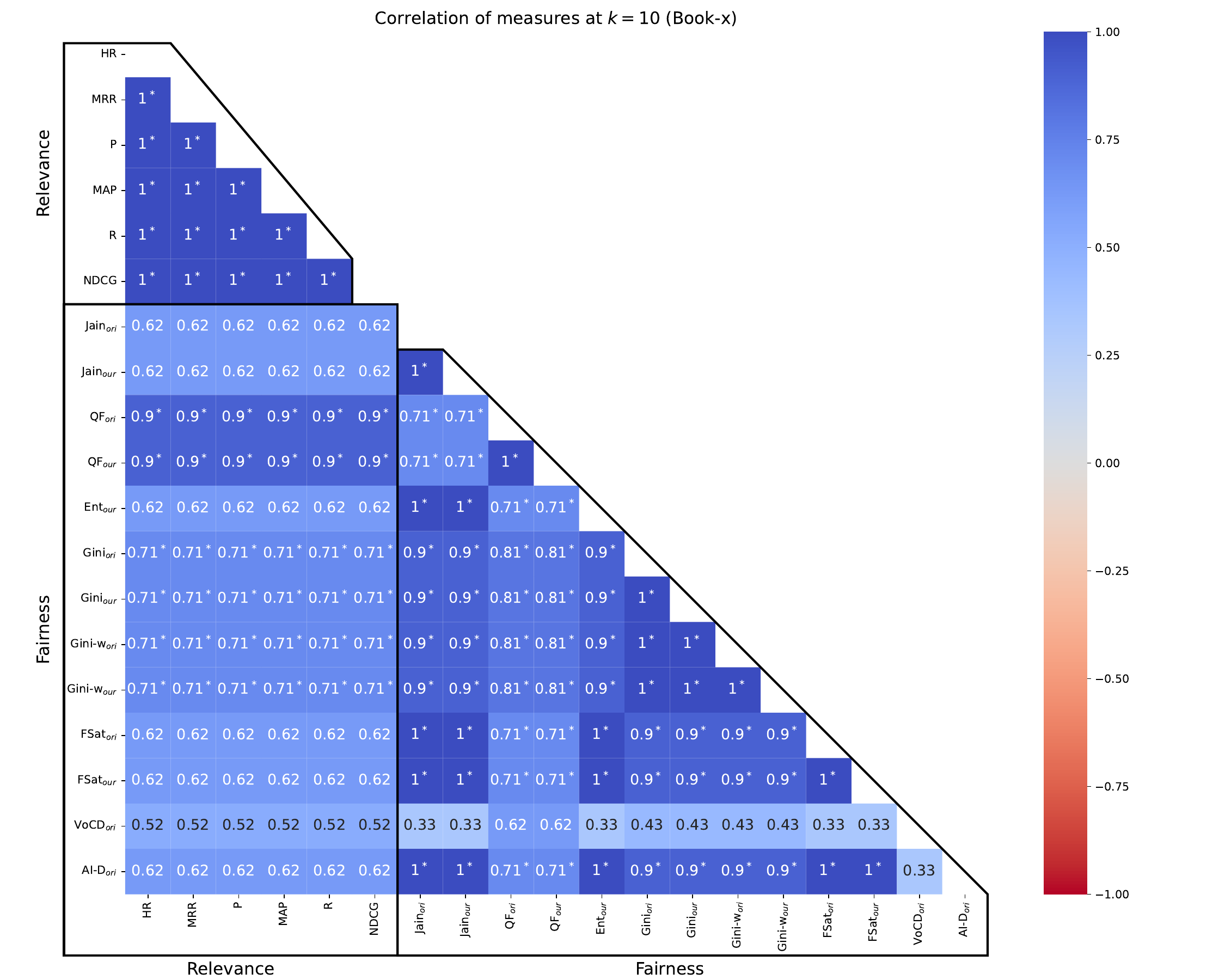}
\caption{Correlation (Kendall's $\tau$) between relevance and fairness measures for Book-x. \explainsig}
\label{fig:corr-bookx}
\end{figure*}

\begin{figure*}[tb]
\centering
    \includegraphics[width=\textwidth]{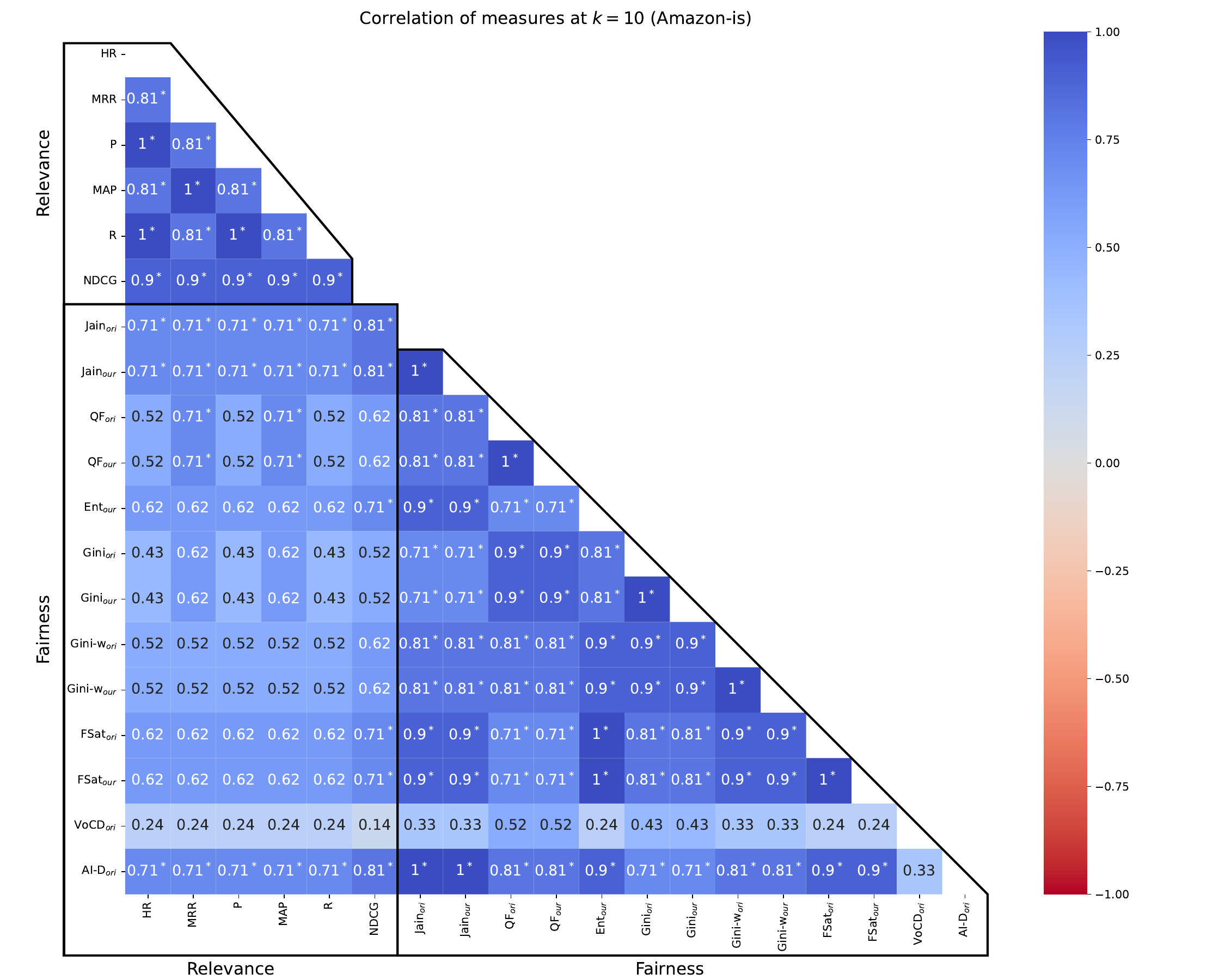}
\caption{Correlation (Kendall's $\tau$) between relevance and fairness measures for Amazon-is. \explainsig}
\label{fig:corr-is}
\end{figure*}

\begin{figure*}[tb]
\centering
    \includegraphics[width=\textwidth]{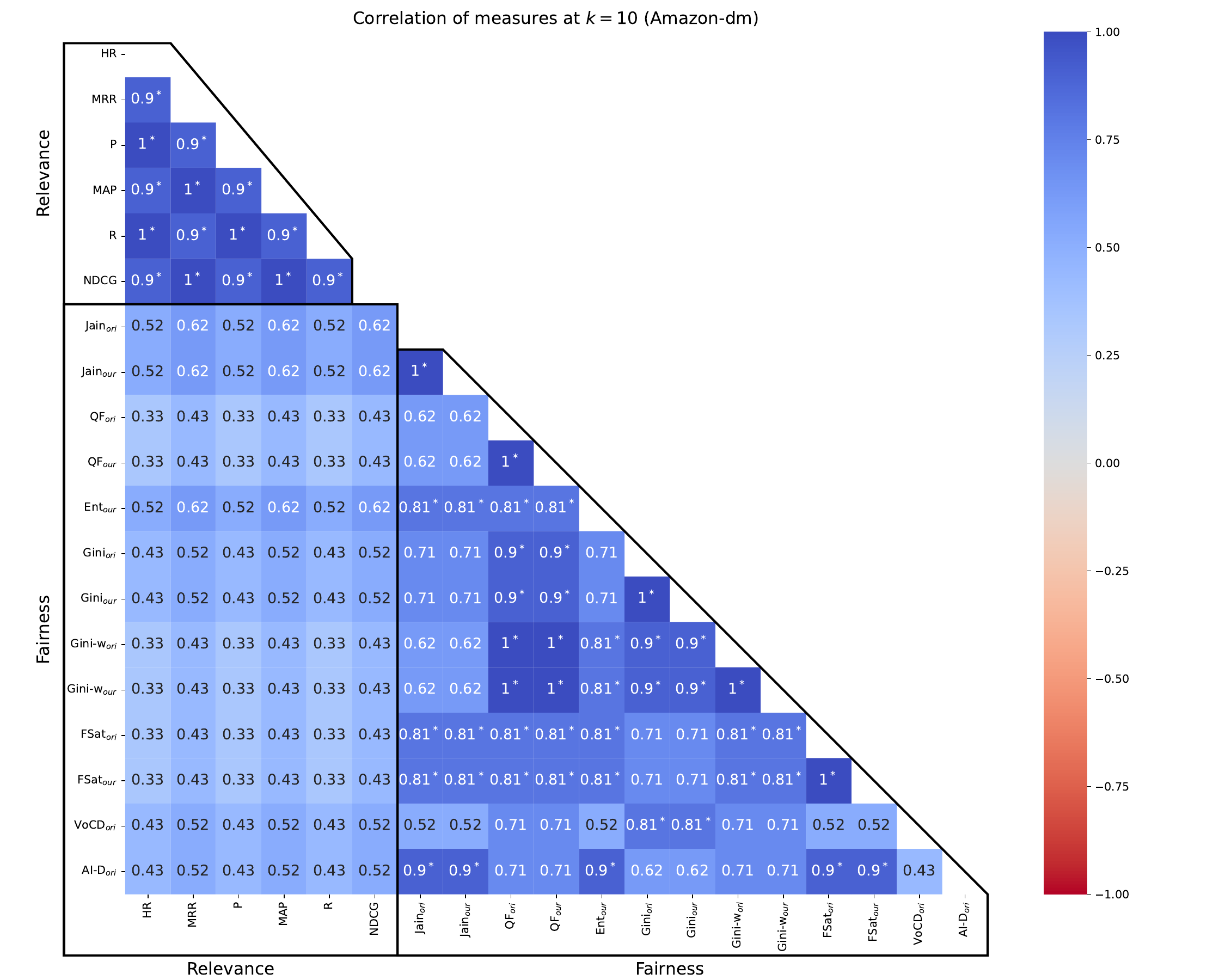}
\caption{Correlation (Kendall's $\tau$) between relevance and fairness measures for Amazon-dm. \explainsig}
\label{fig:corr-dm}
\end{figure*}

\subsection{Max/min achievable fairness}
\label{app:maxmin}
The results of the max/min achievable fairness experiment for Amazon-* and Book-x are in Fig.~\ref{fig:app_mostfair_higher_better}--\ref{fig:app_mostunfair_lower_better}.

\begin{figure}[tb]
    \centering
    \includegraphics[width=\textwidth]{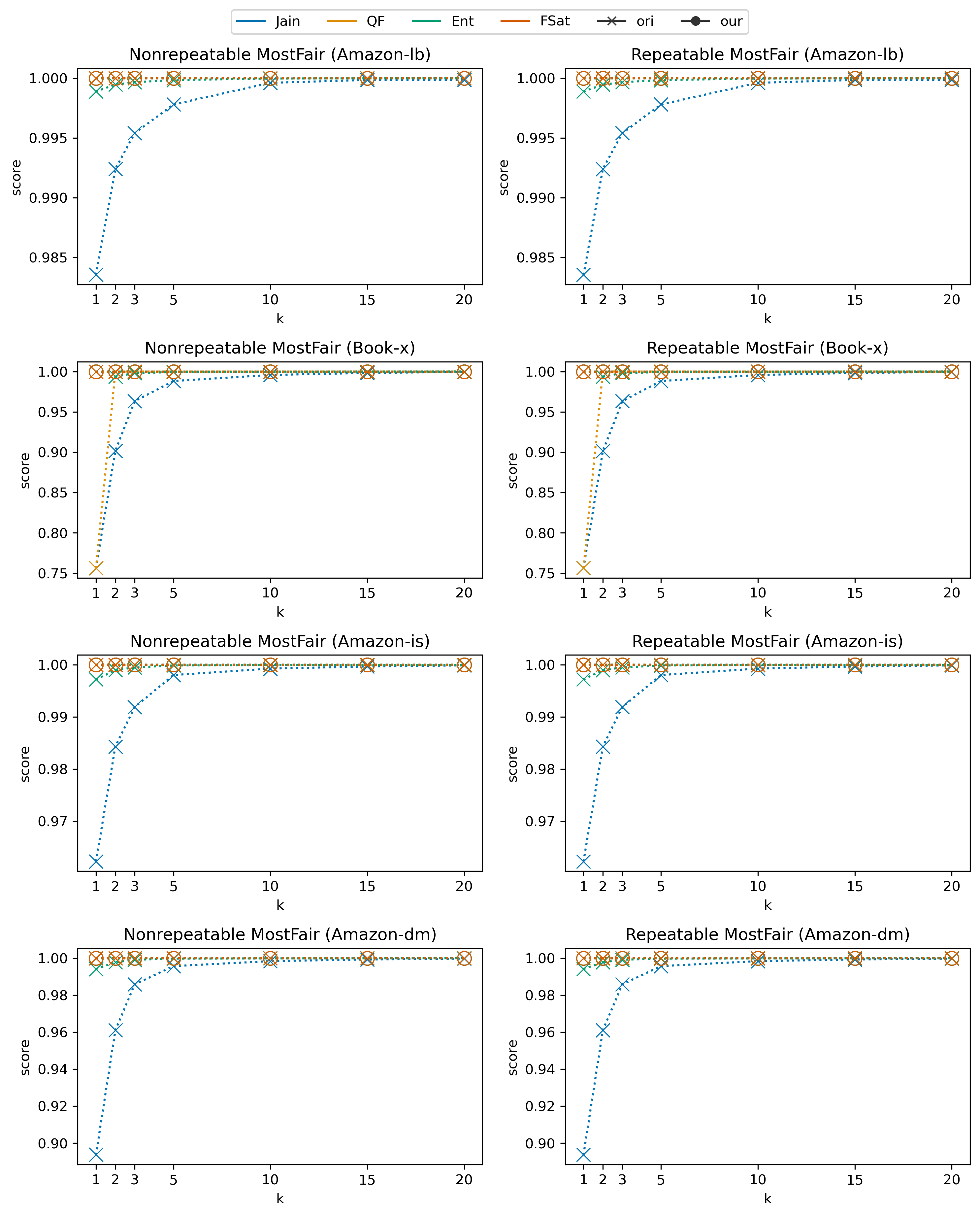}
    \caption{Most fair scores with varying $k$ for higher-is-fairer fairness measures on Amazon-* and Book-x.}
    \label{fig:app_mostfair_higher_better}
\end{figure}
\begin{figure}[tb]
    \centering
    \includegraphics[width=\textwidth]{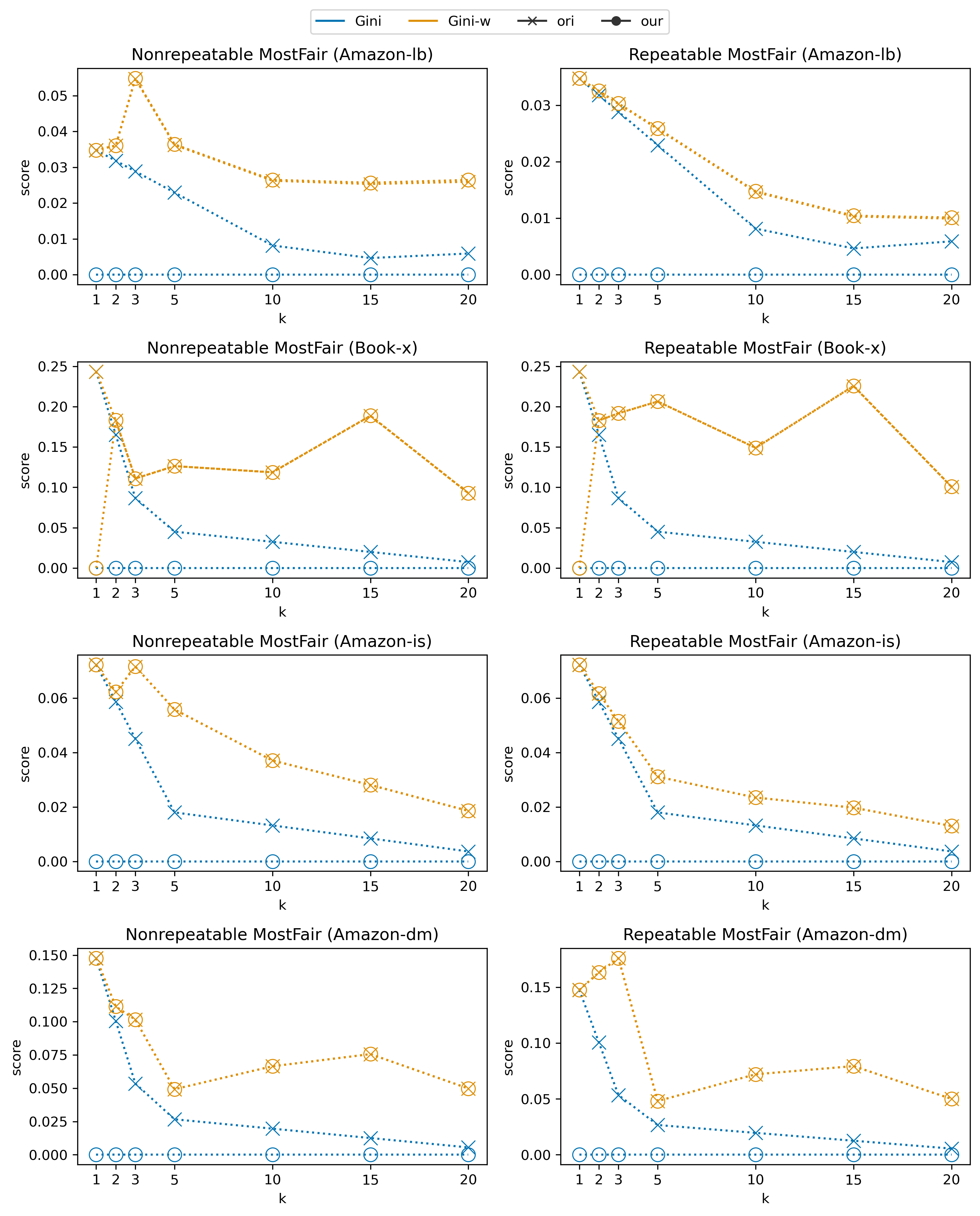}
    \caption{Most fair scores with varying $k$ for lower-is-fairer fairness measures on Amazon-* and Book-x.}
    \label{fig:app_mostfair_lower_better}
\end{figure}
\begin{figure}[tb]
    \centering
    \includegraphics[width=\textwidth]{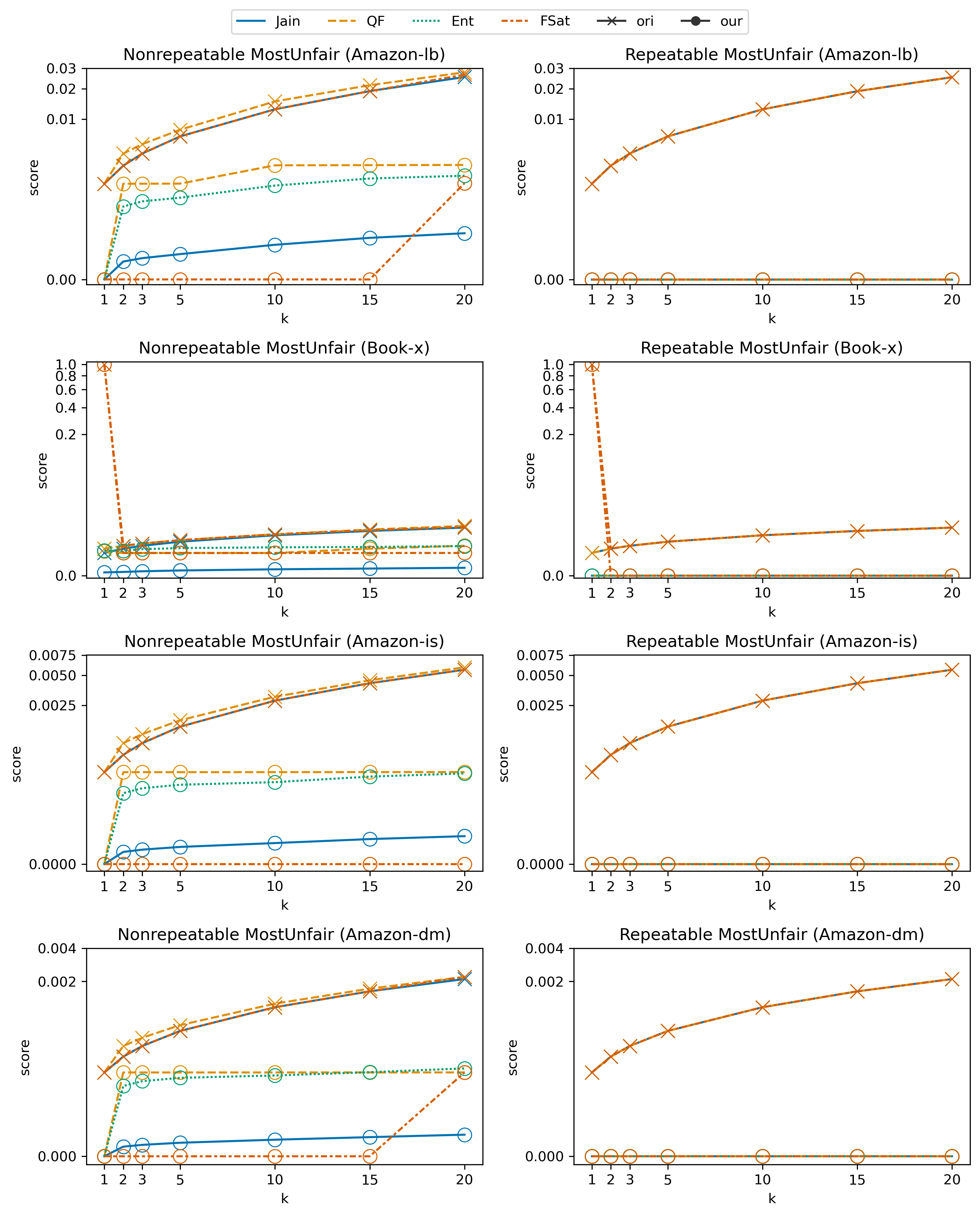}
    \caption{Most unfair scores with varying $k$ for higher-is-fairer fairness measures on Amazon-* and Book-x.}
    \label{fig:app_mostunfair_higher_better}
\end{figure}
\begin{figure}[tb]
    \centering
    \includegraphics[width=\textwidth]{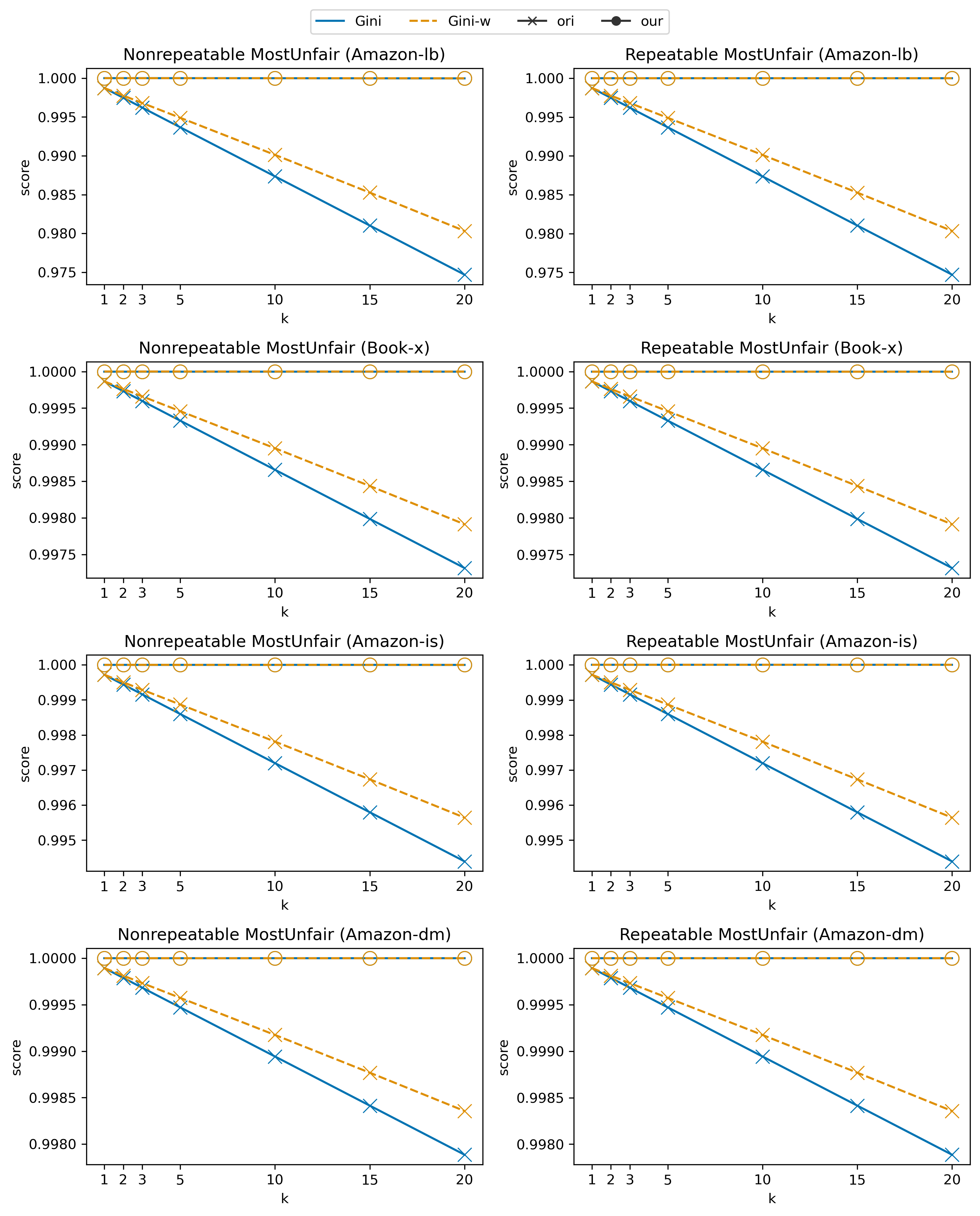}
    \caption{Most unfair scores with varying $k$ for lower-is-fairer fairness measures on Amazon-* and Book-x.}
    \label{fig:app_mostunfair_lower_better}
\end{figure}

\subsection{Sliding window: relevance and fairness at different rank positions}
\label{app:sliding}
The results of the sliding window experiment for Amazon-* and Book-x are in Fig.~\ref{fig:app_sliding_nonfairrel_amazonlb-book-x}--\ref{fig:app_sliding_nonfairrel_amazonisdm}.

\begin{figure}[tb]
    \centering
    \includegraphics[width=\textwidth]{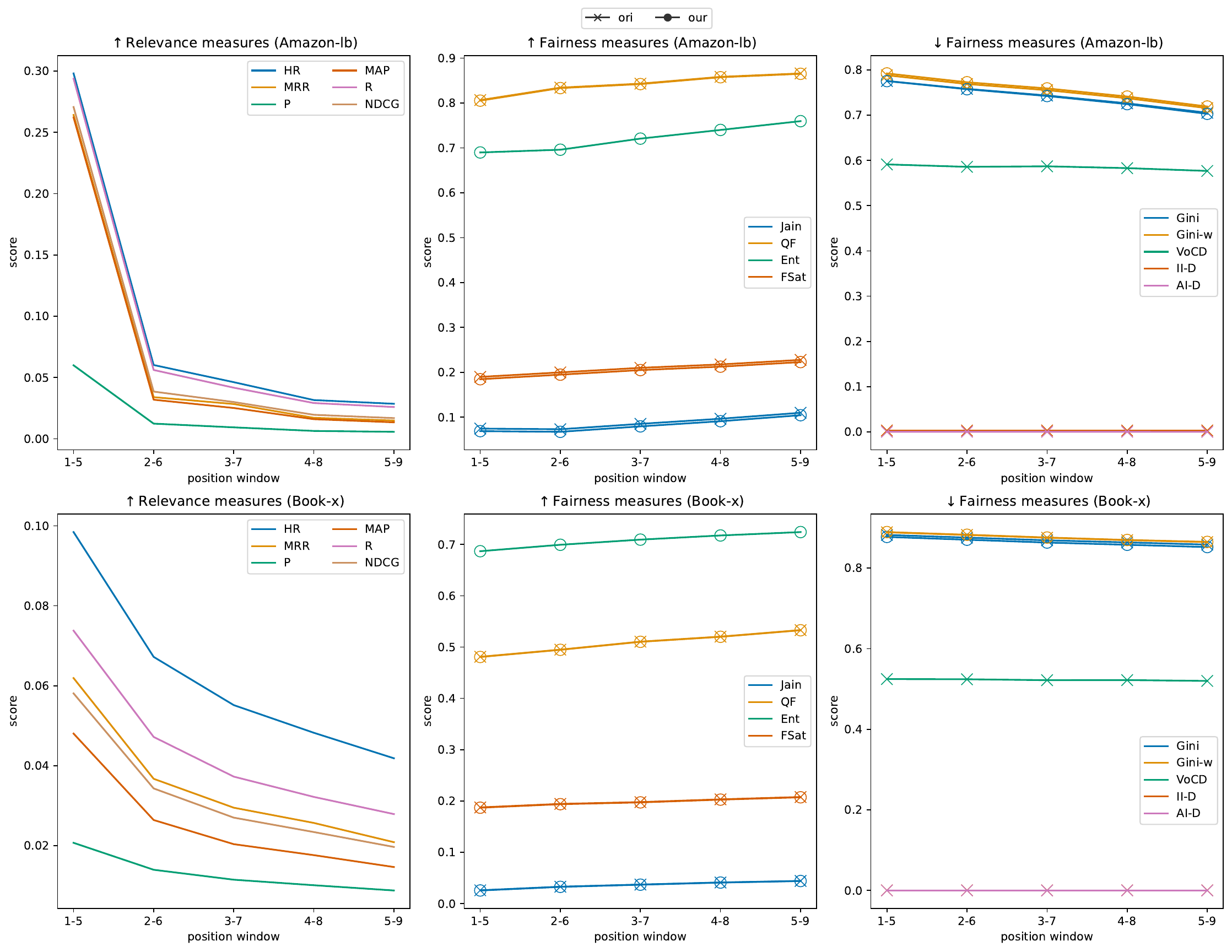}
    \caption{Sliding window evaluation for BPR model, on Amazon-lb and Book-x. Each row is for one dataset, each column is for the different groups of measures (relevance, higher-is-better fairness, lower-is-better fairness measures). II-D and AI-D lines overlap.}
    \label{fig:app_sliding_nonfairrel_amazonlb-book-x}
\end{figure}

\begin{figure}[tb]
    \centering
    \includegraphics[width=\textwidth]{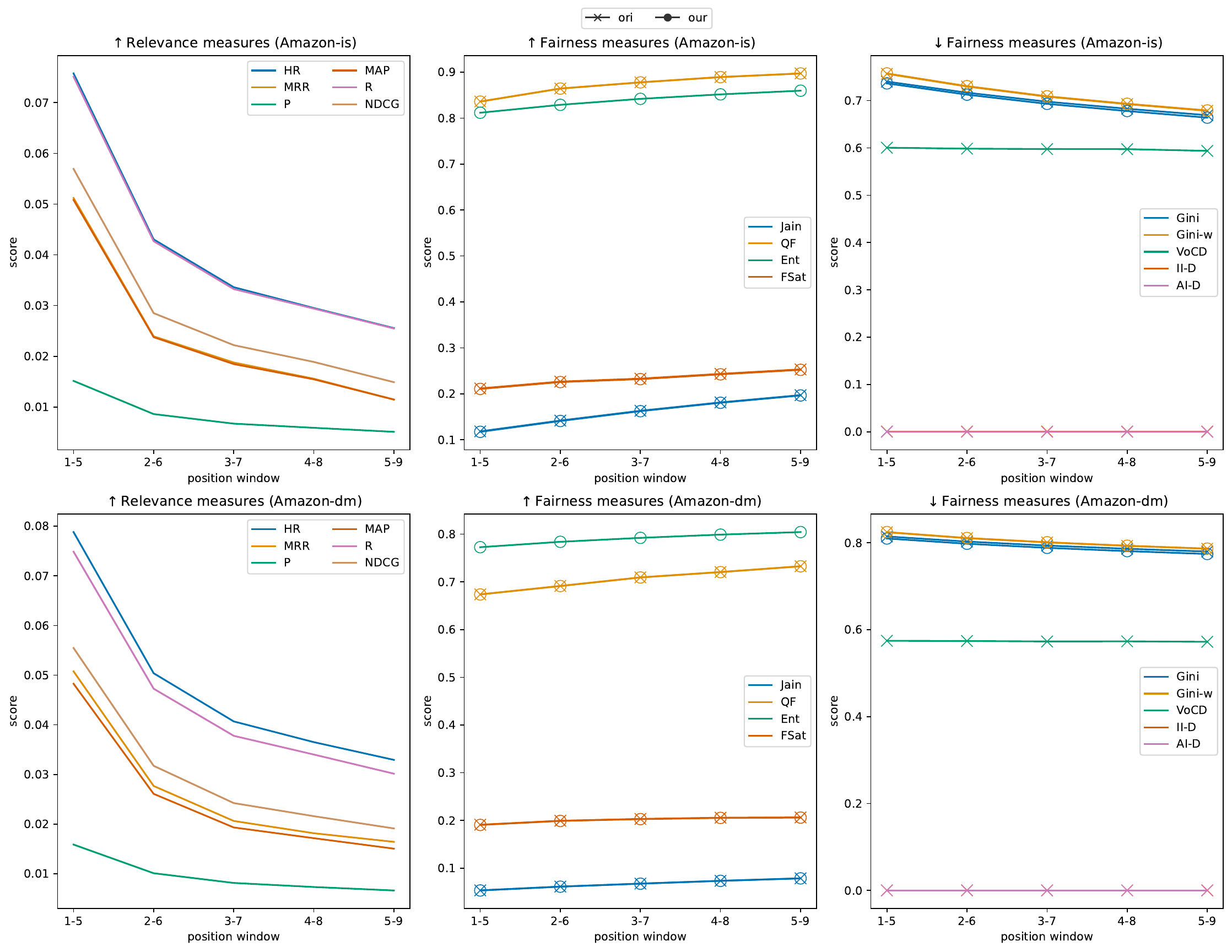}
    \caption{Sliding window evaluation for BPR model, on Amazon-is and Amazon-dm. Each row is for one dataset, each column is for the different groups of measures (relevance, higher-is-better fairness, lower-is-better fairness measures). II-D and AI-D lines overlap.}
    \label{fig:app_sliding_nonfairrel_amazonisdm}
\end{figure}

\subsection{Measure strictness and sensitivity through artificial insertion of items}
\label{app:insert}
We present in Fig. \ref{fig:artificial-fair-extend} the extended results of artificially inserting least exposed (LE) and relevant items for $m=\{100,500\}$ for relevance measures and fairness measures. The changes in scores are less stable compared to $m=1000$, but the general trends are the same.

\begin{figure}[tb]
    \centering
 \includegraphics[width=\textwidth]{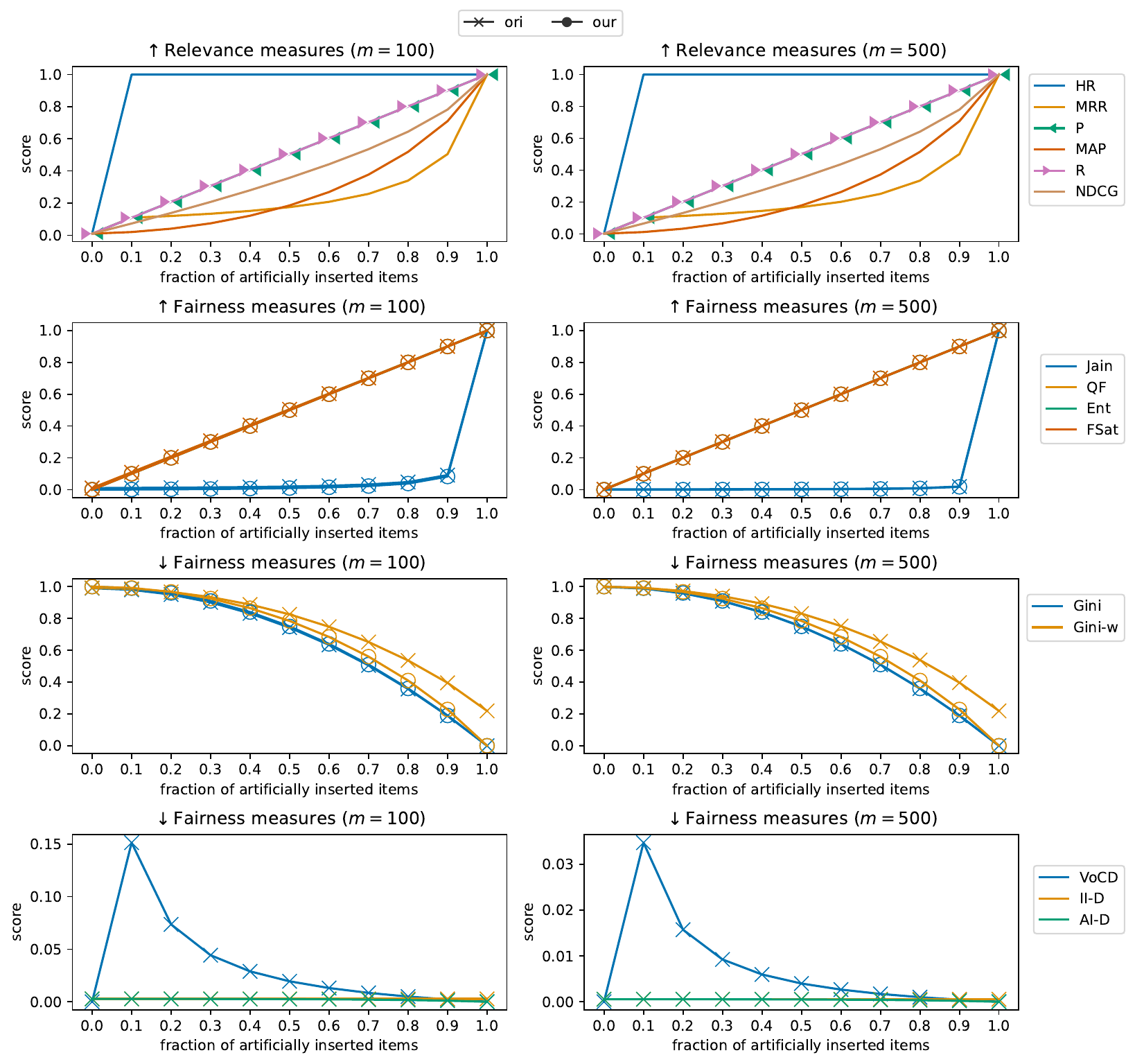}
    \caption{Results for jointly LE and relevant item insertion for $m\in\{100,500\}$. All measures are calculated at $k=10$. QF\ori~and FSat\ori~overlap. QF\our, FSat\our, and Ent\our~also overlap.}
    \label{fig:artificial-fair-extend}
\end{figure}

We also experiment with the artificial insertion of multiple copies of items that are already in the recommendation list and the insertion of irrelevant items, using a similar methodology. We refer to the insertion of multiple item copies as inserting the \textit{most exposed (ME) items}, as in this experiment we aim to maximise exposure of as few items as possible. This is done by iteratively inserting a copy of several items that currently have the most exposure, one copy at a time. 
We swap the starting and ending recommendation list of the artificial insertion of LE and relevant items such that at the end of the experiments, only $k$ unique items are in the recommendation list. These $k$ items will get the most exposure, while the rest of the items in the dataset get zero exposure. The item replacement is still done from the bottom of the recommendation list. 
In Fig.~\ref{fig:artificial-unfair}, we see that the trends of the measures are similar, but the opposite to that of the artificial insertion of LE and relevant items.

\begin{figure*}
    \centering
    \includegraphics[width=\textwidth]{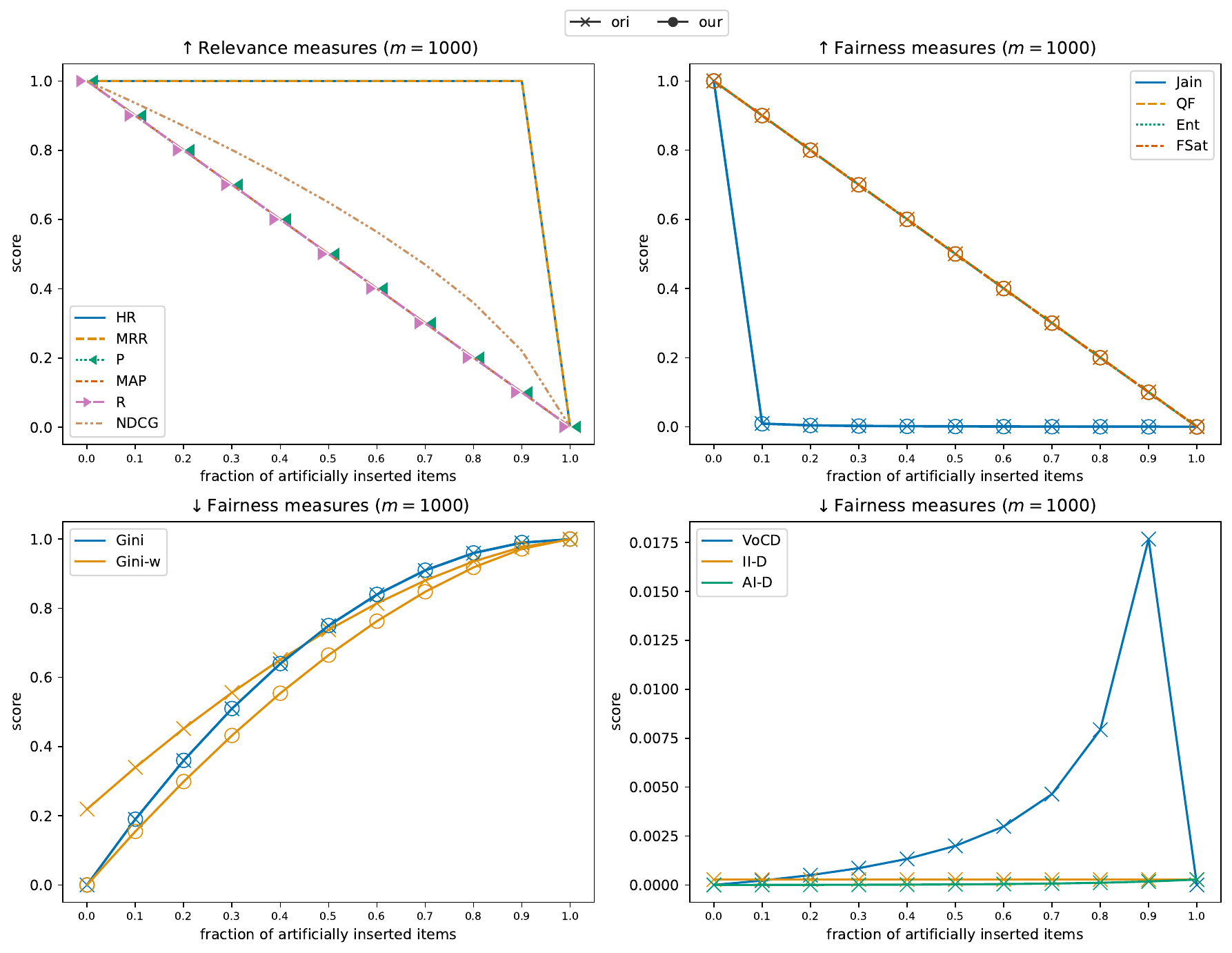}
    \caption{Results for jointly most exposed (ME) and irrelevant item insertion. All measures are calculated at $k=10$. QF\ori~and FSat\ori~overlap. QF\our, FSat\our, and Ent\our~also overlap.}
    \label{fig:artificial-unfair}
\end{figure*}

\end{document}

%% file: tabsplit1/base-main.tex
\begin{table*} 
\caption{Relevance \textsc{(rel)} and fairness \textsc{(fair)} scores of the recommender models on Lastfm and Ml-1m. The most relevant and most fair score per measure is in bold. $\uparrow$ means the higher the better, $\downarrow$ the lower the better. `nan' stands for `not a number'.}
\label{tab:base-main}
\scalebox{1}{
\begin{tabular}{lllrrrrrrr}
\toprule
 &  && Pop$^{*}$ & ItemKNN & SLIM & BPR & NGCF & NeuMF & MultiVAE \\
\midrule
\multirow[c]{21}{*}{Lastfm} 
& \multirow[c]{6}{*}{\textsc{rel}} 
& $\uparrow$ $\text{HR}$ & 0.236686 & 0.563206 & 0.520710 & \bfseries 0.603012 & 0.598171 & 0.571813 & 0.597633 \\
& & $\uparrow$ $\text{MRR}$ & 0.102296 & 0.320710 & 0.290459 & \bfseries 0.336577 & 0.327790 & 0.301449 & 0.326032 \\
& & $\uparrow$ $\text{P}$ & 0.029855 & 0.082948 & 0.075901 & \bfseries 0.091286 & 0.090479 & 0.082679 & 0.091070 \\
& & $\uparrow$ $\text{MAP}$ & 0.033572 & 0.129943 & 0.112303 & \bfseries 0.140626 & 0.136694 & 0.120650 & 0.136794 \\
& & $\uparrow$ $\text{R}$ & 0.078205 & 0.240326 & 0.210221 & \bfseries 0.262193 & 0.260140 & 0.238121 & 0.261565 \\
& & $\uparrow$ $\text{NDCG}$ & 0.063259 & 0.207209 & 0.183180 & \bfseries 0.223416 & 0.219215 & 0.198248 & 0.219408 \\
\cline{2-10}
& \multirow[c]{15}{*}{\textsc{fair}} 
& $\uparrow$ $\text{Jain}_{\text{ori}}$ & 0.005350 & 0.050544 & 0.029023 & 0.080549 & 0.086601 & 0.096789 & \bfseries 0.134035 \\
& & $\uparrow$ $\text{Jain}_{\text{our}}$ & 0.001824 & 0.047434 & 0.025715 & 0.077714 & 0.083822 & 0.094103 & \bfseries 0.131692 \\
& & $\uparrow$ $\text{QF}_{\text{ori}}$ & 0.009564 & 0.432519 & 0.145590 & 0.428268 & 0.399221 & 0.462983 & \bfseries 0.683316 \\
& & $\uparrow$ $\text{QF}_{\text{our}}$ & 0.006043 & 0.430501 & 0.142552 & 0.426235 & 0.397085 & 0.461074 & \bfseries 0.682190 \\
& & $\uparrow$ $\text{Ent}_{\text{ori}}$ & nan & nan & nan & nan & nan & nan & nan \\
& & $\uparrow$ $\text{Ent}_{\text{our}}$ & 0.096360 & 0.595274 & 0.439907 & 0.656154 & 0.660951 & 0.686602 & \bfseries 0.763463 \\
& & $\uparrow$ $\text{FSat}_{\text{ori}}$ & 0.009564 & 0.136734 & 0.071909 & 0.171803 & 0.178888 & 0.195537 & \bfseries 0.230960 \\
& & $\uparrow$ $\text{FSat}_{\text{our}}$ & 0.006043 & 0.133665 & 0.068610 & 0.168859 & 0.175969 & 0.192677 & \bfseries 0.228226 \\
& & $\downarrow$ $\text{Gini}_{\text{ori}}$ & 0.995080 & 0.908416 & 0.968311 & 0.884752 & 0.885149 & 0.864875 & \bfseries 0.780612 \\
& & $\downarrow$ $\text{Gini}_{\text{our}}$ & 0.998564 & 0.908251 & 0.970668 & 0.883591 & 0.884004 & 0.862877 & \bfseries 0.775067 \\
& & $\downarrow$ $\text{Gini-w}_{\text{ori}}$ & 0.995984 & 0.915773 & 0.972331 & 0.895154 & 0.896870 & 0.879730 & \bfseries 0.793962 \\
& & $\downarrow$ $\text{Gini-w}_{\text{our}}$ & 0.998747 & 0.918313 & 0.975028 & 0.897637 & 0.899358 & 0.882170 & \bfseries 0.796164 \\
& & $\downarrow$ $\text{VoCD}_{\text{ori}}$ & 0.669135 & 0.609061 & 0.704415 & 0.640846 & 0.656609 & 0.641465 & \bfseries 0.598510 \\
& & $\downarrow$ $\text{II-D}_{\text{ori}}$ & \bfseries 0.000970 & \bfseries 0.000970 & \bfseries 0.000970 & \bfseries 0.000970 & \bfseries 0.000970 & \bfseries 0.000970 & \bfseries 0.000970 \\
& & $\downarrow$ $\text{AI-D}_{\text{ori}}$ & 0.000668 & 0.000061 & 0.000114 & 0.000037 & 0.000034 & 0.000032 & \bfseries 0.000019 \\

\cline{1-10}
\multirow[c]{21}{*}{Ml-1m} 
& \multirow[c]{6}{*}{\textsc{rel}} 
& $\uparrow$ $\text{HR}$ & 0.273435 & 0.336204 & 0.343326 & \bfseries 0.348791 & 0.337198 & 0.324114 & 0.331070 \\
& & $\uparrow$ $\text{MRR}$ & 0.112995 & 0.136303 & 0.139593 & \bfseries 0.142428 & 0.140765 & 0.133319 & 0.130503 \\
& & $\uparrow$ $\text{P}$ & 0.045661 & 0.054190 & 0.053445 & \bfseries 0.055167 & 0.055018 & 0.052302 & 0.051308 \\
& & $\uparrow$ $\text{MAP}$ & 0.025940 & 0.036132 & 0.036830 & \bfseries 0.038389 & 0.037460 & 0.033568 & 0.035284 \\
& & $\uparrow$ $\text{R}$ & 0.040195 & 0.064935 & 0.071923 & \bfseries 0.073647 & 0.067244 & 0.060861 & 0.068755 \\
& & $\uparrow$ $\text{NDCG}$ & 0.056219 & 0.073897 & 0.075873 & \bfseries 0.078110 & 0.076032 & 0.070162 & 0.072263 \\
\cline{2-10}
& \multirow[c]{15}{*}{\textsc{fair}} 
& $\uparrow$ $\text{Jain}_{\text{ori}}$ & 0.006867 & 0.023143 & 0.045463 & \bfseries 0.068296 & 0.057974 & 0.059396 & 0.065298 \\
& & $\uparrow$ $\text{Jain}_{\text{our}}$ & 0.003857 & 0.020192 & 0.042593 & \bfseries 0.065508 & 0.055148 & 0.056576 & 0.062499 \\
& & $\uparrow$ $\text{QF}_{\text{ori}}$ & 0.040822 & 0.188388 & 0.236166 & 0.444209 & 0.306018 & 0.384336 & \bfseries 0.485939 \\
& & $\uparrow$ $\text{QF}_{\text{our}}$ & 0.037913 & 0.185927 & 0.233849 & 0.442524 & 0.303913 & 0.382469 & \bfseries 0.484380 \\
& & $\uparrow$ $\text{Ent}_{\text{ori}}$ & nan & nan & nan & nan & nan & nan & nan \\
& & $\uparrow$ $\text{Ent}_{\text{our}}$ & 0.187196 & 0.435739 & 0.552900 & 0.650556 & 0.606893 & 0.623249 & \bfseries 0.652672 \\
& & $\uparrow$ $\text{FSat}_{\text{ori}}$ & 0.020865 & 0.069549 & 0.114605 & 0.164500 & 0.147263 & 0.153311 & \bfseries 0.167523 \\
& & $\uparrow$ $\text{FSat}_{\text{our}}$ & 0.017895 & 0.066727 & 0.111920 & 0.161965 & 0.144677 & 0.150743 & \bfseries 0.164998 \\
& & $\downarrow$ $\text{Gini}_{\text{ori}}$ & 0.993229 & 0.970583 & 0.942655 & 0.893080 & 0.919816 & 0.908995 & \bfseries 0.888977 \\
& & $\downarrow$ $\text{Gini}_{\text{our}}$ & 0.996201 & 0.973246 & 0.944935 & 0.894681 & 0.921783 & 0.910813 & \bfseries 0.890521 \\
& & $\downarrow$ $\text{Gini-w}_{\text{ori}}$ & 0.994377 & 0.973088 & 0.948836 & 0.902403 & 0.927651 & 0.918188 & \bfseries 0.894798 \\
& & $\downarrow$ $\text{Gini-w}_{\text{our}}$ & 0.996731 & 0.975391 & 0.951082 & 0.904539 & 0.929846 & 0.920361 & \bfseries 0.896916 \\
& & $\downarrow$ $\text{VoCD}_{\text{ori}}$ & 0.789888 & 0.733042 & 0.737530 & 0.705762 & 0.724919 & 0.712774 & \bfseries 0.699980 \\
& & $\downarrow$ $\text{II-D}_{\text{ori}}$ & \bfseries 0.000828 & \bfseries 0.000828 & \bfseries 0.000828 & \bfseries 0.000828 & \bfseries 0.000828 & \bfseries 0.000828 & \bfseries 0.000828 \\
& & $\downarrow$ $\text{AI-D}_{\text{ori}}$ & 0.000381 & 0.000096 & 0.000049 & 0.000032 & 0.000038 & 0.000038 & \bfseries 0.000030 \\
\bottomrule
\multicolumn{10}{l}{\multirow{2}{*}{
\parbox[t]{0.95\textwidth}{*The scores of our \textsc{fair} measures for Pop are not 0 or 1, because in our experiment set-up, items from users' train or validation splits are excluded from the top $k$ recommendation list.}}}
\end{tabular}%
}
\end{table*}

%% file: tab/searchhyper.tex
\begin{tabular}{ll}
\toprule
 & Hyperparameter search space \\
 \midrule
ItemKNN & k: {[}10, 20, 30, 40, 50, 60, 70, 80, 90, 100, 150, 200, 250, 300, 400, 500, 1000{]} \\
 & shrink: {[}0.0, 0.5, 1.0{]} \\
 \midrule
\multirow{2}{*}{SLIM} & alpha: {[}0.2, 0.5, 0.8, 1.0{]} \\
 & l1 ratio: {[}0.01, 0.02, 0.05, 0.1, 0.5{]} \\
 \midrule
\multirow{2}{*}{BPR} & embedding size: {[}16, 32, 64, 128, 256, 512, 1024, 2048, 4096{]} \\
 & lr: {[}5e-5, 1e-4, 5e-4, 1e-3, 5e-3, 1e-2, 5e-2{]} \\
 \midrule
\multirow{4}{*}{NGCF} & dropout prob: {[}0.1, 0.2{]} \\
 & embedding size: {[}64, 128, 256, 512{]} \\
 & hidden size: {[}{[}64, 64, 64{]}, {[}128, 128, 128{]}, {[}256, 256, 256{]}{]} \\
 & lr: {[}5e-4, 1e-3, 5e-3{]} \\
 \midrule
\multirow{3}{*}{NeuMF} & dropout prob: {[}0.1, 0.2{]} \\
 & hidden size: {[}{[}128, 64{]}, {[}128, 64, 32{]}, {[}64, 32, 16{]}, {[}32, 16, 8{]}{]} \\
 & lr: {[}5e-4, 1e-3, 5e-3, 1e-2{]} \\
 \midrule
\multirow{4}{*}{MultiVAE} & dropout prob: {[}0.1, 0.2, 0.5{]} \\
 & hidden size: {[}{[}100{]}, {[}300{]}, {[}600{]}{]} \\
 & latent dimension: {[}64, 128, 256{]} \\
 & lr: {[}5e-4, 1e-3, 5e-3, 1e-2{]} \\
 \bottomrule
\end{tabular}

%% file: tab/besthyper.tex
\begin{tabular}{lllllll}
\toprule
{} &                                    ItemKNN &                                             SLIM &                                                  BPR &                                                                                                     NGCF &                                                                           NeuMF &                                                                                     MultiVAE \\
\midrule
Amazon-lb &    \parbox[t]{1.5cm}{k: 20, \\shrink: 1.0} &  \parbox[t]{1.8cm}{alpha: 0.2, \\l1 ratio: 0.01} &  \parbox[t]{3cm}{embedding size: 4096, \\lr: 0.0001} &   \parbox[t]{3.5cm}{dropout prob: 0.2, \\embedding size: 256, \\hidden size: [128,128,128], \\lr: 0.005} &   \parbox[t]{3.2cm}{dropout prob: 0.2, \\hidden size: [128,64,32], \\lr: 0.005} &   \parbox[t]{2.5cm}{dropout prob: 0.5, \\hidden size: [600], \\latent dim: 128, \\lr: 0.005} \\
\midrule
Lastfm    &   \parbox[t]{1.5cm}{k: 300, \\shrink: 1.0} &  \parbox[t]{1.8cm}{alpha: 0.2, \\l1 ratio: 0.01} &  \parbox[t]{3cm}{embedding size: 2048, \\lr: 0.0005} &   \parbox[t]{3.5cm}{dropout prob: 0.2, \\embedding size: 512, \\hidden size: [256,256,256], \\lr: 0.001} &     \parbox[t]{3.2cm}{dropout prob: 0.2, \\hidden size: [32,16,8], \\lr: 0.001} &   \parbox[t]{2.5cm}{dropout prob: 0.5, \\hidden size: [600], \\latent dim: 64, \\lr: 0.0005} \\
\midrule
Ml-1m     &  \parbox[t]{1.5cm}{k: 1000, \\shrink: 1.0} &  \parbox[t]{1.8cm}{alpha: 0.2, \\l1 ratio: 0.01} &  \parbox[t]{3cm}{embedding size: 2048, \\lr: 0.0001} &  \parbox[t]{3.5cm}{dropout prob: 0.1, \\embedding size: 256, \\hidden size: [256,256,256], \\lr: 0.0005} &   \parbox[t]{3.2cm}{dropout prob: 0.2, \\hidden size: [64,32,16], \\lr: 0.0005} &     \parbox[t]{2.5cm}{dropout prob: 0.1, \\hidden size: [600], \\latent dim: 64, \\lr: 0.01} \\
\midrule
Book-x    &    \parbox[t]{1.5cm}{k: 20, \\shrink: 1.0} &  \parbox[t]{1.8cm}{alpha: 0.2, \\l1 ratio: 0.01} &  \parbox[t]{3cm}{embedding size: 4096, \\lr: 0.0001} &  \parbox[t]{3.5cm}{dropout prob: 0.1, \\embedding size: 512, \\hidden size: [256,256,256], \\lr: 0.0005} &     \parbox[t]{3.2cm}{dropout prob: 0.1, \\hidden size: [32,16,8], \\lr: 0.001} &  \parbox[t]{2.5cm}{dropout prob: 0.5, \\hidden size: [600], \\latent dim: 128, \\lr: 0.0005} \\
\midrule
Amazon-is &   \parbox[t]{1.5cm}{k: 250, \\shrink: 1.0} &  \parbox[t]{1.8cm}{alpha: 0.2, \\l1 ratio: 0.01} &  \parbox[t]{3cm}{embedding size: 4096, \\lr: 0.0001} &   \parbox[t]{3.5cm}{dropout prob: 0.1, \\embedding size: 512, \\hidden size: [256,256,256], \\lr: 0.001} &   \parbox[t]{3.2cm}{dropout prob: 0.2, \\hidden size: [128,64,32], \\lr: 0.001} &   \parbox[t]{2.5cm}{dropout prob: 0.5, \\hidden size: [600], \\latent dim: 128, \\lr: 0.005} \\
\midrule
Amazon-dm &    \parbox[t]{1.5cm}{k: 10, \\shrink: 1.0} &  \parbox[t]{1.8cm}{alpha: 0.2, \\l1 ratio: 0.01} &  \parbox[t]{3cm}{embedding size: 4096, \\lr: 0.0001} &   \parbox[t]{3.5cm}{dropout prob: 0.2, \\embedding size: 512, \\hidden size: [256,256,256], \\lr: 0.001} &  \parbox[t]{3.2cm}{dropout prob: 0.2, \\hidden size: [128,64,32], \\lr: 0.0005} &    \parbox[t]{2.5cm}{dropout prob: 0.2, \\hidden size: [600], \\latent dim: 256, \\lr: 0.01} \\
\bottomrule
\end{tabular}

%% file: tabsplit1/base-extra1.tex
\begin{table*}
\caption{Relevance \textsc{(rel)} and fairness \textsc{(fair)} scores of the recommender models for Amazon-lb and Book-x. The most relevant and most fair score per measure is in bold. $\uparrow$ means the higher the better, $\downarrow$ the lower the better. `nan' stands for `not a number'.}
\label{tab:base-extra1}
\scalebox{0.95}{
\begin{tabular}{lllrrrrrrr}
\toprule
 &  && Pop$^{*}$  & ItemKNN & SLIM & BPR & NGCF & NeuMF & MultiVAE \\
\midrule
\multirow[c]{21}{*}{Amazon-lb} 
& \multirow[c]{6}{*}{\textsc{rel}} 
& $\uparrow$ $\text{HR}$ & 0.257908 & 0.309611 & 0.313869 & 0.320560 & 0.305961 & \bfseries 0.325426 & 0.321776 \\
& & $\uparrow$ $\text{MRR}$ & 0.222141 & 0.265047 & \bfseries 0.270739 & 0.267188 & 0.263003 & 0.267511 & 0.264885 \\
& & $\uparrow$ $\text{P}$ & 0.025912 & 0.031022 & 0.031752 & 0.032360 & 0.030718 & \bfseries 0.032968 & 0.032664 \\
& & $\uparrow$ $\text{MAP}$ & 0.221568 & 0.264694 & \bfseries 0.268219 & 0.265084 & 0.261340 & 0.263527 & 0.261339 \\
& & $\uparrow$ $\text{R}$ & 0.252251 & 0.308698 & 0.306767 & 0.315237 & 0.302616 & \bfseries 0.318375 & 0.314852 \\
& & $\uparrow$ $\text{NDCG}$ & 0.228580 & 0.275196 & \bfseries 0.278262 & 0.277727 & 0.271526 & 0.277666 & 0.275024 \\
\cline{2-10}
& \multirow[c]{15}{*}{\textsc{fair}} 
 & $\uparrow$ $\text{Jain}_{\text{ori}}$ & 0.017661 & 0.188400 & 0.050735 & 0.134711 & \bfseries 0.193167 & 0.060450 & 0.028521 \\
 && $\uparrow$ $\text{Jain}_{\text{our}}$ & 0.005085 & 0.178079 & 0.038596 & 0.123681 & \bfseries 0.182909 & 0.048439 & 0.016089 \\
 && $\uparrow$ $\text{QF}_{\text{ori}}$ & 0.024020 & \bfseries 0.978508 & 0.379267 & 0.921618 & 0.967130 & 0.790139 & 0.386852 \\
 && $\uparrow$ $\text{QF}_{\text{our}}$ & 0.011524 & \bfseries 0.978233 & 0.371319 & 0.920615 & 0.966709 & 0.787452 & 0.379001 \\
 && $\uparrow$ $\text{Ent}_{\text{ori}}$ & nan & nan & nan & nan & nan & nan & nan \\
 && $\uparrow$ $\text{Ent}_{\text{our}}$ & 0.094992 & \bfseries 0.822309 & 0.460701 & 0.731135 & 0.812621 & 0.564546 & 0.310875 \\
 && $\uparrow$ $\text{FSat}_{\text{ori}}$ & 0.022756 & 0.285714 & 0.134008 & 0.212389 & \bfseries 0.294564 & 0.158028 & 0.074589 \\
 && $\uparrow$ $\text{FSat}_{\text{our}}$ & 0.010243 & 0.276569 & 0.122919 & 0.202305 & \bfseries 0.285531 & 0.147247 & 0.062740 \\
& & $\downarrow$ $\text{Gini}_{\text{ori}}$ & 0.983735 & \bfseries 0.580725 & 0.915792 & 0.710528 & 0.606176 & 0.841397 & 0.955542 \\
& & $\downarrow$ $\text{Gini}_{\text{our}}$ & 0.996300 & \bfseries 0.584731 & 0.926915 & 0.717292 & 0.610723 & 0.850940 & 0.967509 \\
& & $\downarrow$ $\text{Gini-w}_{\text{ori}}$ & 0.986481 & \bfseries 0.605205 & 0.914661 & 0.731609 & 0.628605 & 0.861775 & 0.959218 \\
& & $\downarrow$ $\text{Gini-w}_{\text{our}}$ & 0.996317 & \bfseries 0.611240 & 0.923782 & 0.738904 & 0.634873 & 0.870368 & 0.968782 \\
& & $\downarrow$ $\text{VoCD}_{\text{ori}}$ & \bfseries 0.509398 & 0.560071 & 0.654539 & 0.617237 & 0.590049 & 0.645770 & 0.661114 \\
& & $\downarrow$ $\text{II-D}_{\text{ori}}$ & \bfseries 0.003439 & \bfseries 0.003439 & \bfseries 0.003439 & \bfseries 0.003439 & \bfseries 0.003439 & \bfseries 0.003439 & \bfseries 0.003439 \\
& & $\downarrow$ $\text{AI-D}_{\text{ori}}$ & 0.002443 & \bfseries 0.000158 & 0.000610 & 0.000280 & 0.000165 & 0.000728 & 0.001385 \\

\cline{1-10}
\multirow[c]{21}{*}{Book-x} 
& \multirow[c]{6}{*}{\textsc{rel}} 
& $\uparrow$ $\text{HR}$ & 0.034581 & 0.115269 & 0.059762 & \bfseries 0.130342 & 0.103564 & 0.085299 & 0.089910 \\
& & $\uparrow$ $\text{MRR}$ & 0.012664 & 0.063709 & 0.034535 & \bfseries 0.066125 & 0.041422 & 0.037474 & 0.039687 \\
& & $\uparrow$ $\text{P}$ & 0.003476 & 0.012999 & 0.006349 & \bfseries 0.014276 & 0.010977 & 0.008973 & 0.009576 \\
& & $\uparrow$ $\text{MAP}$ & 0.008713 & 0.047170 & 0.023412 & \bfseries 0.050778 & 0.032601 & 0.029294 & 0.031962 \\
& & $\uparrow$ $\text{R}$ & 0.022909 & 0.082607 & 0.037784 & \bfseries 0.097505 & 0.078578 & 0.064687 & 0.069501 \\
& & $\uparrow$ $\text{NDCG}$ & 0.013074 & 0.059732 & 0.029559 & \bfseries 0.065820 & 0.045847 & 0.039839 & 0.042925 \\

\cline{2-10}
& \multirow[c]{15}{*}{\textsc{fair}} 
& $\uparrow$ $\text{Jain}_{\text{ori}}$ & 0.001420 & \bfseries 0.375894 & 0.002297 & 0.036191 & 0.020391 & 0.030978 & 0.050532 \\
& & $\uparrow$ $\text{Jain}_{\text{our}}$ & 0.000080 & \bfseries 0.376670 & 0.000961 & 0.035046 & 0.019157 & 0.029804 & 0.049469 \\
& & $\uparrow$ $\text{QF}_{\text{ori}}$ & 0.002414 & \bfseries 0.899799 & 0.019584 & 0.671630 & 0.662777 & 0.519920 & 0.587928 \\
& & $\uparrow$ $\text{QF}_{\text{our}}$ & 0.001075 & \bfseries 0.899664 & 0.018267 & 0.671189 & 0.662324 & 0.519275 & 0.587374 \\
& & $\uparrow$ $\text{Ent}_{\text{ori}}$ & nan & nan & nan & nan & nan & nan & nan \\
& & $\uparrow$ $\text{Ent}_{\text{our}}$ & 0.014182 & \bfseries 0.916559 & 0.156894 & 0.688733 & 0.623813 & 0.646908 & 0.716425 \\
& & $\uparrow$ $\text{FSat}_{\text{ori}}$ & 0.002012 & \bfseries 0.392622 & 0.017170 & 0.158820 & 0.120724 & 0.145942 & 0.188062 \\
& & $\uparrow$ $\text{FSat}_{\text{our}}$ & 0.000672 & \bfseries 0.391807 & 0.015850 & 0.157690 & 0.119543 & 0.144795 & 0.186971 \\

& & $\downarrow$ $\text{Gini}_{\text{ori}}$ & 0.998615 & \bfseries 0.548876 & 0.996714 & 0.853353 & 0.877575 & 0.894874 & 0.848862 \\
& & $\downarrow$ $\text{Gini}_{\text{our}}$ & 0.999955 & \bfseries 0.534458 & 0.997988 & 0.849604 & 0.874675 & 0.892579 & 0.844955 \\
& & $\downarrow$ $\text{Gini-w}_{\text{ori}}$ & 0.998906 & \bfseries 0.582320 & 0.996386 & 0.863144 & 0.889082 & 0.904719 & 0.860833 \\
& & $\downarrow$ $\text{Gini-w}_{\text{our}}$ & 0.999954 & \bfseries 0.582930 & 0.997431 & 0.864049 & 0.890015 & 0.905668 & 0.861736 \\
& & $\downarrow$ $\text{VoCD}_{\text{ori}}$ & 0.663365 & 0.564473 & 0.720722 & 0.570963 & \bfseries 0.537809 & 0.600639 & 0.611070 \\
& & $\downarrow$ $\text{II-D}_{\text{ori}}$ & \bfseries 0.000368 & \bfseries 0.000368 & \bfseries 0.000368 & \bfseries 0.000368 & \bfseries 0.000368 & \bfseries 0.000368 & \bfseries 0.000368 \\
& & $\downarrow$ $\text{AI-D}_{\text{ori}}$ & 0.000350 & \bfseries 0.000001 & 0.000176 & 0.000013 & 0.000027 & 0.000015 & 0.000009 \\

\bottomrule
\multicolumn{10}{l}{\multirow{2}{*}{
\parbox[t]{0.95\textwidth}{*The scores of our \textsc{fair} measures for Pop are not 0 or 1, because in our experiment set-up, items from users' train or validation splits are excluded from the top $k$ recommendation list.}}}
\end{tabular}}
\end{table*}

%% file: tabsplit1/base-extra2.tex
\begin{table*}
\caption{Relevance \textsc{(rel)} and fairness \textsc{(fair)} scores of the recommender models for Amazon-is and Amazon-dm. The most relevant and most fair score per measure is in bold. $\uparrow$ means the higher the better, $\downarrow$ the lower the better. `nan' stands for `not a number'.}
\label{tab:base-extra2}
\scalebox{0.95}{
\begin{tabular}{lllrrrrrrr}
\toprule
 &  && Pop$^{*}$  & ItemKNN & SLIM & BPR & NGCF & NeuMF & MultiVAE \\
\midrule
\multirow[c]{21}{*}{Amazon-is} 
& \multirow[c]{6}{*}{\textsc{rel}} 
 & $\uparrow$ $\text{HR}$ & 0.031792 & 0.084271 & 0.028902 & \bfseries 0.100548 & 0.095224 & 0.077122 & 0.081838 \\
& & $\uparrow$ $\text{MRR}$ & 0.011671 & 0.047550 & 0.023075 & \bfseries 0.054403 & 0.048366 & 0.042206 & 0.041294 \\
& & $\uparrow$ $\text{P}$ & 0.003179 & 0.008427 & 0.002890 & \bfseries 0.010055 & 0.009522 & 0.007712 & 0.008184 \\
& & $\uparrow$ $\text{MAP}$ & 0.011593 & 0.047313 & 0.022880 & \bfseries 0.054002 & 0.048093 & 0.041905 & 0.041036 \\
& & $\uparrow$ $\text{R}$ & 0.031462 & 0.083916 & 0.028572 & \bfseries 0.099825 & 0.094717 & 0.076387 & 0.080993 \\
& & $\uparrow$ $\text{NDCG}$ & 0.016224 & 0.055902 & 0.024254 & \bfseries 0.064820 & 0.059064 & 0.050052 & 0.050436 \\

\cline{2-10}
& \multirow[c]{15}{*}{\textsc{fair}} 
 & $\uparrow$ $\text{Jain}_{\text{ori}}$ & 0.002925 & \bfseries 0.493178 & 0.002958 & 0.167723 & 0.101140 & 0.061059 & 0.093401 \\
& & $\uparrow$ $\text{Jain}_{\text{our}}$ & 0.000124 & \bfseries 0.492108 & 0.000156 & 0.165503 & 0.098685 & 0.058463 & 0.090919 \\
& & $\uparrow$ $\text{QF}_{\text{ori}}$ & 0.004483 & \bfseries 0.989633 & 0.014570 & 0.956851 & 0.869151 & 0.919866 & 0.856823 \\
& & $\uparrow$ $\text{QF}_{\text{our}}$ & 0.001686 & \bfseries 0.989604 & 0.011801 & 0.956729 & 0.868783 & 0.919640 & 0.856420 \\
& & $\uparrow$ $\text{Ent}_{\text{ori}}$ & nan & nan & nan & nan & nan & nan & nan \\
& & $\uparrow$ $\text{Ent}_{\text{our}}$ & 0.012261 & \bfseries 0.939484 & 0.024132 & 0.832303 & 0.746126 & 0.741374 & 0.766085 \\
& & $\uparrow$ $\text{FSat}_{\text{ori}}$ & 0.003642 & \bfseries 0.373494 & 0.008686 & 0.223032 & 0.171477 & 0.170636 & 0.204819 \\
& & $\uparrow$ $\text{FSat}_{\text{our}}$ & 0.000843 & \bfseries 0.371734 & 0.005901 & 0.220849 & 0.169149 & 0.168306 & 0.202585 \\
& & $\downarrow$ $\text{Gini}_{\text{ori}}$ & 0.997136 & \bfseries 0.445874 & 0.997066 & 0.679338 & 0.798275 & 0.763940 & 0.769841 \\
& & $\downarrow$ $\text{Gini}_{\text{our}}$ & 0.999937 & \bfseries 0.439697 & 0.999866 & 0.676963 & 0.797837 & 0.762944 & 0.768941 \\
& & $\downarrow$ $\text{Gini-w}_{\text{ori}}$ & 0.997738 & \bfseries 0.476640 & 0.997518 & 0.703146 & 0.818197 & 0.784878 & 0.784863 \\
& & $\downarrow$ $\text{Gini-w}_{\text{our}}$ & 0.999926 & \bfseries 0.477685 & 0.999705 & 0.704688 & 0.819991 & 0.786599 & 0.786584 \\
& & $\downarrow$ $\text{VoCD}_{\text{ori}}$ & 0.592760 & \bfseries 0.523208 & 0.787225 & 0.615085 & 0.640122 & 0.616732 & 0.648353 \\
& & $\downarrow$ $\text{II-D}_{\text{ori}}$ & \bfseries 0.000768 & \bfseries 0.000768 & \bfseries 0.000768 & \bfseries 0.000768 & \bfseries 0.000768 & \bfseries 0.000768 & \bfseries 0.000768 \\
& & $\downarrow$ $\text{AI-D}_{\text{ori}}$ & 0.000738 & \bfseries 0.000002 & 0.000705 & 0.000011 & 0.000019 & 0.000036 & 0.000020 \\
\cline{1-10}
\multirow[c]{21}{*}{Amazon-dm} 
& \multirow[c]{6}{*}{\textsc{rel}} 
& $\uparrow$ $\text{HR}$ & 0.022809 & 0.087660 & 0.005702 & \bfseries 0.108596 & 0.093787 & 0.073872 & 0.079660 \\
& & $\uparrow$ $\text{MRR}$ & 0.009252 & 0.048030 & 0.004607 & \bfseries 0.054654 & 0.043576 & 0.033947 & 0.036754 \\
& & $\uparrow$ $\text{P}$ & 0.002289 & 0.008928 & 0.000570 & \bfseries 0.011004 & 0.009489 & 0.007489 & 0.008051 \\
& & $\uparrow$ $\text{MAP}$ & 0.008463 & 0.044812 & 0.004106 & \bfseries 0.051930 & 0.041533 & 0.032263 & 0.034956 \\
& & $\uparrow$ $\text{R}$ & 0.020528 & 0.081430 & 0.005123 & \bfseries 0.102600 & 0.089059 & 0.069847 & 0.075436 \\
& & $\uparrow$ $\text{NDCG}$ & 0.011454 & 0.054259 & 0.004447 & \bfseries 0.064525 & 0.053219 & 0.041538 & 0.044960 \\
\cline{2-10}
& \multirow[c]{15}{*}{\textsc{fair}} 
 & $\uparrow$ $\text{Jain}_{\text{ori}}$ & 0.001092 & \bfseries 0.358718 & 0.001080 & 0.068624 & 0.078606 & 0.038114 & 0.121317 \\
& & $\uparrow$ $\text{Jain}_{\text{our}}$ & 0.000036 & \bfseries 0.358605 & 0.000023 & 0.067745 & 0.077754 & 0.037155 & 0.120578 \\
& & $\uparrow$ $\text{QF}_{\text{ori}}$ & 0.001902 & \bfseries 0.957620 & 0.003593 & 0.843585 & 0.720672 & 0.843691 & 0.879624 \\
& & $\uparrow$ $\text{QF}_{\text{our}}$ & 0.000846 & \bfseries 0.957575 & 0.002539 & 0.843419 & 0.720377 & 0.843525 & 0.879496 \\
& & $\uparrow$ $\text{Ent}_{\text{ori}}$ & nan & nan & nan & nan & nan & nan & nan \\
& & $\uparrow$ $\text{Ent}_{\text{our}}$ & 0.008686 & \bfseries 0.940649 & 0.008735 & 0.779117 & 0.768888 & 0.764866 & 0.840448 \\
& & $\uparrow$ $\text{FSat}_{\text{ori}}$ & 0.001480 & \bfseries 0.404143 & 0.002219 & 0.187804 & 0.199535 & 0.197738 & 0.246354 \\
& & $\uparrow$ $\text{FSat}_{\text{our}}$ & 0.000423 & \bfseries 0.403512 & 0.001164 & 0.186945 & 0.198688 & 0.196890 & 0.245556 \\

& & $\downarrow$ $\text{Gini}_{\text{ori}}$ & 0.998925 & \bfseries 0.462579 & 0.998932 & 0.778288 & 0.814128 & 0.771164 & 0.702994 \\
& & $\downarrow$ $\text{Gini}_{\text{our}}$ & 0.999982 & \bfseries 0.452327 & 0.999989 & 0.774693 & 0.811289 & 0.767418 & 0.697811 \\
& & $\downarrow$ $\text{Gini-w}_{\text{ori}}$ & 0.999154 & \bfseries 0.493364 & 0.999132 & 0.791188 & 0.828961 & 0.788203 & 0.718667 \\
& & $\downarrow$ $\text{Gini-w}_{\text{our}}$ & 0.999979 & \bfseries 0.493772 & 0.999957 & 0.791841 & 0.829646 & 0.788854 & 0.719261 \\
& & $\downarrow$ $\text{VoCD}_{\text{ori}}$ & 0.668615 & \bfseries 0.532741 & 0.796493 & 0.615043 & 0.643785 & 0.611526 & 0.616569 \\
& & $\downarrow$ $\text{II-D}_{\text{ori}}$ & \bfseries 0.000290 & \bfseries 0.000290 & \bfseries 0.000290 & \bfseries 0.000290 & \bfseries 0.000290 & \bfseries 0.000290 & \bfseries 0.000290 \\
& & $\downarrow$ $\text{AI-D}_{\text{ori}}$ & 0.000282 & \bfseries 0.000000 & 0.000278 & 0.000004 & 0.000004 & 0.000008 & 0.000002 \\

\bottomrule
\multicolumn{10}{l}{\multirow{2}{*}{
\parbox[t]{0.95\textwidth}{*The scores of our \textsc{fair} measures for Pop are not 0 or 1, because in our experiment set-up, items from users' train or validation splits are excluded from the top $k$ recommendation list.}}}
\end{tabular}%
}
\end{table*}

%% file: references.bib
@misc{jain1984quantitative,
    title = {{A Quantitative Measure Of Fairness And Discrimination For Resource Allocation In Shared Computer Systems}},
    year = {1998},
    booktitle = {Eastern Research Laboratory, Digital Equipment Corporation, Hudson, MA},
    author = {Jain, Rajendra K and Chiu, Dah-Ming W and Hawe, William R and {others}},
    volume = {21},
    url = {http://arxiv.org/abs/cs/9809099},
    arxivId = {cs/9809099}
}

@article{Wang2022,
    title = {{A Survey on the Fairness of Recommender Systems}},
    year = {2023},
    journal = {ACM Trans. Inf. Syst.},
    author = {Wang, Yifan and Ma, Weizhi and Zhang, Min and Liu, Yiqun and Ma, Shaoping},
    number = {3},
    month = {2},
    pages = {1--43},
    volume = {41},
    publisher = {Association for Computing Machinery},
    url = {https://doi.org/10.1145/3547333},
    address = {New York, NY, USA},
    doi = {10.1145/3547333},
    issn = {1046-8188}
}

@inproceedings{RendleBPR:Feedback,
    title = {{BPR: Bayesian Personalized Ranking from Implicit Feedback}},
    year = {2009},
    booktitle = {Proceedings of the Twenty-Fifth Conference on Uncertainty in Artificial Intelligence},
    author = {Rendle, Steffen and Freudenthaler, Christoph and Gantner, Zeno and Schmidt-Thieme, Lars},
    pages = {452--461},
    series = {UAI '09},
    publisher = {AUAI Press},
    address = {Arlington, Virginia, USA},
    isbn = {9780974903958}
}

@article{HerlockerEvaluatingSystems,
    title = {{Evaluating Collaborative Filtering Recommender Systems}},
    year = {2004},
    journal = {ACM Trans. Inf. Syst.},
    author = {Herlocker, Jonathan L and Konstan, Joseph A and Terveen, Loren G and Riedl, John T},
    number = {1},
    month = {1},
    pages = {5--53},
    volume = {22},
    publisher = {Association for Computing Machinery},
    url = {https://doi.org/10.1145/963770.963772},
    address = {New York, NY, USA},
    doi = {10.1145/963770.963772},
    issn = {1046-8188}
}

@article{Majumder2021FairFairness,
    title = {{Fair Enough: Searching for Sufficient Measures of Fairness}},
    year = {2023},
    journal = {ACM Trans. Softw. Eng. Methodol.},
    author = {Majumder, Suvodeep and Chakraborty, Joymallya and Bai, Gina R and Stolee, Kathryn T and Menzies, Tim},
    number = {6},
    month = {9},
    volume = {32},
    publisher = {Association for Computing Machinery},
    url = {https://doi.org/10.1145/3585006},
    address = {New York, NY, USA},
    doi = {10.1145/3585006},
    issn = {1049-331X}
}

@article{LiYunqi2023FairnessApplications,
    title = {{Fairness in Recommendation: Foundations, Methods and Applications}},
    year = {2023},
    journal = {ACM Transactions on Intelligent Systems and Technology},
    author = {Li, Yunqi and Chen, Hanxiong and Xu, Shuyuan and Ge, Yingqiang and Tan, Juntao and Liu, Shuchang and Zhang, Yongfeng},
    month = {1},
    publisher = {ACM},
    url = {https://dl.acm.org/doi/10.1145/3610302},
    doi = {10.1145/3610302},
    issn = {2157-6904}
}

@article{Mansoury2021ASystems,
    title = {{A Graph-Based Approach for Mitigating Multi-Sided Exposure Bias in Recommender Systems}},
    year = {2021},
    journal = {ACM Transactions on Information Systems (TOIS)},
    author = {Mansoury, Masoud and Abdollahpouri, Himan and Pechenizkiy, Mykola and Mobasher, Bamshad and Burke, Robin},
    number = {2},
    month = {11},
    pages = {32},
    volume = {40},
    publisher = {ACM PUB27 New York, NY},
    url = {https://dl.acm.org/doi/10.1145/3470948},
    doi = {10.1145/3470948},
    issn = {15582868},
    arxivId = {2107.03415}
}

@article{Shannon1948ACommunication,
    title = {{A Mathematical Theory of Communication}},
    year = {1948},
    journal = {The Bell System Technical Journal},
    author = {Shannon, C E},
    pages = {623--656},
    volume = {27}
}

@inproceedings{Yang2013ASystem,
    title = {{A sentiment-enhanced personalized location recommendation system}},
    year = {2013},
    booktitle = {HT 2013 - Proceedings of the 24th ACM Conference on Hypertext and Social Media},
    author = {Yang, Dingqi and Zhang, Daqing and Yu, Zhiyong and Wang, Zhu},
    pages = {119--128},
    url = {https://dl.acm.org/doi/10.1145/2481492.2481505},
    isbn = {9781450319676},
    doi = {10.1145/2481492.2481505}
}

@article{Holm1979AProcedure,
    title = {{A Simple Sequentially Rejective Multiple Test Procedure}},
    year = {1979},
    journal = {Scandinavian Journal of Statistics},
    author = {Holm, S.},
    doi = {10.2307/4615733}
}

@article{Amigo2023ASystems,
    title = {{A unifying and general account of fairness measurement in recommender systems}},
    year = {2023},
    journal = {Information Processing {\&} Management},
    author = {Amig{\'{o}}, Enrique and Deldjoo, Yashar and Mizzaro, Stefano and Bellog{\'{i}}n, Alejandro},
    number = {1},
    month = {1},
    pages = {103115},
    volume = {60},
    publisher = {Pergamon},
    url = {https://linkinghub.elsevier.com/retrieve/pii/S0306457322002163},
    doi = {10.1016/J.IPM.2022.103115},
    issn = {0306-4573}
}

@inproceedings{Kingma2014Adam:Optimization,
    title = {{Adam: A Method for Stochastic Optimization}},
    year = {2014},
    booktitle = {3rd International Conference on Learning Representations, ICLR 2015 - Conference Track Proceedings},
    author = {Kingma, Diederik P. and Ba, Jimmy Lei},
    month = {12},
    publisher = {International Conference on Learning Representations, ICLR},
    url = {https://arxiv.org/abs/1412.6980v9},
    doi = {10.48550/arxiv.1412.6980},
    arxivId = {1412.6980}
}

@inproceedings{Morik2020ControllingLearning-to-Rank,
    title = {{Controlling Fairness and Bias in Dynamic Learning-to-Rank}},
    year = {2020},
    booktitle = {SIGIR 2020 - Proceedings of the 43rd International ACM SIGIR Conference on Research and Development in Information Retrieval},
    author = {Morik, Marco and Singh, Ashudeep and Hong, Jessica and Joachims, Thorsten},
    month = {7},
    pages = {429--438},
    publisher = {Association for Computing Machinery, Inc},
    isbn = {9781450380164},
    doi = {10.1145/3397271.3401100},
    arxivId = {2005.14713}
}

@article{Benjamini1995ControllingTesting,
    title = {{Controlling the False Discovery Rate: A Practical and Powerful Approach to Multiple Testing}},
    year = {1995},
    journal = {Journal of the Royal Statistical Society: Series B (Methodological)},
    author = {Benjamini, Yoav and Hochberg, Yosef},
    number = {1},
    month = {1},
    pages = {289--300},
    volume = {57},
    publisher = {John Wiley {\&} Sons, Ltd},
    url = {https://onlinelibrary.wiley.com/doi/full/10.1111/j.2517-6161.1995.tb02031.x https://onlinelibrary.wiley.com/doi/abs/10.1111/j.2517-6161.1995.tb02031.x https://rss.onlinelibrary.wiley.com/doi/10.1111/j.2517-6161.1995.tb02031.x},
    doi = {10.1111/J.2517-6161.1995.TB02031.X},
    issn = {2517-6161}
}

@article{Jarvelin2002CumulatedTechniques,
    title = {{Cumulated gain-based evaluation of IR techniques}},
    year = {2002},
    journal = {ACM Transactions on Information Systems},
    author = {J{\"{a}}rvelin, Kalervo and Kek{\"{a}}l{\"{a}}inen, Jaana},
    number = {4},
    month = {10},
    pages = {422--446},
    volume = {20},
    url = {https://dl.acm.org/doi/10.1145/582415.582418},
    doi = {10.1145/582415.582418},
    issn = {10468188}
}

@inproceedings{Borges2019EnhancingAutoencoders,
    title = {{Enhancing Long Term Fairness in Recommendations with Variational Autoencoders}},
    year = {2019},
    booktitle = {Proceedings of the 11th International Conference on Management of Digital EcoSystems},
    author = {Borges, Rodrigo and Stefanidis, Kostas},
    publisher = {ACM},
    url = {https://doi.org/10.1145/3297662.3365798},
    address = {New York, NY, USA},
    isbn = {9781450362382},
    doi = {10.1145/3297662}
}

@inproceedings{Biega2018EquityRankings,
    title = {{Equity of attention: Amortizing individual fairness in rankings}},
    year = {2018},
    booktitle = {41st International ACM SIGIR Conference on Research and Development in Information Retrieval, SIGIR 2018},
    author = {Biega, Asia J. and Gummadi, Krishna P. and Weikum, Gerhard},
    month = {6},
    pages = {405--414},
    volume = {18},
    publisher = {Association for Computing Machinery, Inc},
    url = {https://doi.org/10.1145/3209978.3210063},
    isbn = {9781450356572},
    doi = {10.1145/3209978.3210063},
    arxivId = {1805.01788}
}

@inproceedings{Diaz2020EvaluatingExposure,
    title = {{Evaluating Stochastic Rankings with Expected Exposure}},
    year = {2020},
    booktitle = {Proceedings of the 29th ACM International Conference on Information {\&} Knowledge Management},
    author = {Diaz, Fernando and Mitra, Bhaskar and Ekstrand, Michael D and Biega, Asia J and Carterette, Ben},
    publisher = {ACM},
    url = {https://doi.org/10.1145/3340531.3411962},
    address = {New York, NY, USA},
    isbn = {9781450368599},
    doi = {10.1145/3340531}
}

@inproceedings{Saito2022FairRanking,
    title = {{Fair Ranking as Fair Division:  Impact-Based Individual Fairness in Ranking}},
    year = {2022},
    booktitle = {Proceedings of the 28th ACM SIGKDD Conference on Knowledge Discovery and Data Mining (KDD '22), August 14-18, 2022, Washington, DC, USA},
    author = {Saito, Yuta and Joachims, Thorsten},
    pages = {1514--1524},
    volume = {1},
    publisher = {ACM},
    isbn = {9781450393850},
    doi = {10.1145/3534678.3539353}
}

@inproceedings{Mansoury2020FairMatch:Systems,
    title = {{FairMatch: A Graph-based Approach for Improving Aggregate Diversity in Recommender Systems}},
    year = {2020},
    booktitle = {UMAP 2020 - Proceedings of the 28th ACM Conference on User Modeling, Adaptation and Personalization},
    author = {Mansoury, Masoud and Abdollahpouri, Himan and Pechenizkiy, Mykola and Mobasher, Bamshad and Burke, Robin},
    month = {7},
    pages = {154--162},
    volume = {20},
    publisher = {ACM},
    url = {https://doi.org/10.1145/3340631.3394860},
    isbn = {9781450368612},
    doi = {10.1145/3340631.3394860},
    arxivId = {2005.01148}
}

@inproceedings{Zhu2021FairnessSystems,
    title = {{Fairness among New Items in Cold Start Recommender Systems}},
    year = {2021},
    booktitle = {SIGIR 2021 - Proceedings of the 44th International ACM SIGIR Conference on Research and Development in Information Retrieval},
    author = {Zhu, Ziwei and Kim, Jingu and Nguyen, Trung and Fenton, Aish and Caverlee, James},
    month = {7},
    pages = {767--776},
    publisher = {Association for Computing Machinery, Inc},
    url = {https://doi.org/10.1145/3404835.3462948},
    isbn = {9781450380379},
    doi = {10.1145/3404835.3462948}
}

@article{Zehlike2022FairnessSystems,
    title = {{Fairness in Ranking, Part II: Learning-to-Rank and Recommender Systems}},
    year = {2022},
    journal = {ACM Computing Surveys},
    author = {Zehlike, Meike and Yang, Ke and Stoyanovich, Julia},
    number = {6},
    month = {12},
    volume = {55},
    publisher = {ACM PUB27 New York, NY},
    url = {https://dl.acm.org/doi/10.1145/3533380},
    doi = {10.1145/3533380},
    issn = {15577341}
}

@inproceedings{Dwork2012FairnessAwareness,
    title = {{Fairness through awareness}},
    year = {2012},
    booktitle = {ITCS 2012 - Innovations in Theoretical Computer Science Conference},
    author = {Dwork, Cynthia and Hardt, Moritz and Pitassi, Toniann and Reingold, Omer and Zemel, Richard},
    pages = {214--226},
    url = {https://dl.acm.org/doi/10.1145/2090236.2090255},
    isbn = {9781450311151},
    doi = {10.1145/2090236.2090255},
    arxivId = {1104.3913}
}

@inproceedings{Patro2020FairRec:Platforms,
    title = {{FairRec: Two-Sided Fairness for Personalized Recommendations in Two-Sided Platforms}},
    year = {2020},
    booktitle = {The Web Conference 2020 - Proceedings of the World Wide Web Conference, WWW 2020},
    author = {Patro, Gourab K. and Biswas, Arpita and Ganguly, Niloy and Gummadi, Krishna P. and Chakraborty, Abhijnan},
    month = {4},
    pages = {1194--1204},
    publisher = {Association for Computing Machinery, Inc},
    isbn = {9781450370233},
    doi = {10.1145/3366423.3380196},
    arxivId = {2002.10764}
}

@article{Zhu2020FARM:APPs,
    title = {{FARM: A Fairness-Aware Recommendation Method for High Visibility and Low Visibility Mobile APPs}},
    year = {2020},
    journal = {IEEE Access},
    author = {Zhu, Qiliang and Sun, Qibo and Li, Zengxiang and Wang, Shangguang},
    pages = {122747--122756},
    volume = {8},
    publisher = {Institute of Electrical and Electronics Engineers Inc.},
    doi = {10.1109/ACCESS.2020.3007617},
    issn = {21693536}
}

@inproceedings{Rashid2002GettingYou,
    title = {{Getting to know you}},
    year = {2002},
    booktitle = {Proceedings of the 7th international conference on Intelligent user interfaces  - IUI '02},
    author = {Rashid, Al Mamunur and Albert, Istvan and Cosley, Dan and Lam, Shyong K. and McNee, Sean M. and Konstan, Joseph A. and Riedl, John},
    pages = {127},
    publisher = {ACM Press},
    url = {https://dl.acm.org/doi/10.1145/502716.502737},
    address = {New York, New York, USA},
    isbn = {1581133820},
    doi = {10.1145/502716.502737}
}

@inproceedings{Ziegler2005ImprovingDiversification,
    title = {{Improving Recommendation Lists through Topic Diversification}},
    year = {2005},
    booktitle = {Proceedings of the 14th International Conference on World Wide Web},
    author = {Ziegler, Cai-Nicolas and McNee, Sean M and Konstan, Joseph A and Lausen, Georg},
    pages = {22--32},
    series = {WWW '05},
    publisher = {Association for Computing Machinery},
    url = {https://doi.org/10.1145/1060745.1060754},
    address = {New York, NY, USA},
    isbn = {1595930469},
    doi = {10.1145/1060745.1060754}
}

@article{Deshpande2004Item-basedAlgorithms,
    title = {{Item-based top-N recommendation algorithms}},
    year = {2004},
    journal = {ACM Transactions on Information Systems},
    author = {Deshpande, Mukund and Karypis, George},
    number = {1},
    month = {1},
    pages = {143--177},
    volume = {22},
    publisher = {ACM},
    url = {https://dl.acm.org/doi/10.1145/963770.963776},
    doi = {10.1145/963770.963776},
    issn = {10468188}
}

@inproceedings{Wu2022JointRecommendation,
    title = {{Joint Multisided Exposure Fairness for Recommendation}},
    year = {2022},
    booktitle = {SIGIR 2022 - Proceedings of the 45th International ACM SIGIR Conference on Research and Development in Information Retrieval},
    author = {Wu, Haolun and Mitra, Bhaskar and Ma, Chen and Diaz, Fernando and Liu, Xue},
    month = {7},
    pages = {703--714},
    publisher = {Association for Computing Machinery, Inc},
    url = {https://doi.org/10.1145/3477495.3532007},
    isbn = {9781450387323},
    doi = {10.1145/3477495.3532007},
    arxivId = {2205.00048}
}

@inproceedings{Ni2019JustifyingAspects,
    title = {{Justifying Recommendations using Distantly-Labeled Reviews and Fine-Grained Aspects}},
    year = {2019},
    booktitle = {EMNLP-IJCNLP 2019 - 2019 Conference on Empirical Methods in Natural Language Processing and 9th International Joint Conference on Natural Language Processing, Proceedings of the Conference},
    author = {Ni, Jianmo and Li, Jiacheng and McAuley, Julian},
    pages = {188--197},
    publisher = {Association for Computational Linguistics},
    url = {https://aclanthology.org/D19-1018},
    isbn = {9781950737901},
    doi = {10.18653/V1/D19-1018}
}

@inproceedings{BorgBruun2022LearningDomain,
    title = {{Learning Recommendations from User Actions in the Item-poor Insurance Domain}},
    year = {2022},
    booktitle = {RecSys 2022 - Proceedings of the 16th ACM Conference on Recommender Systems},
    author = {Borg Bruun, Simone and Maistro, Maria and Lioma, Christina},
    month = {9},
    pages = {113--123},
    publisher = {Association for Computing Machinery, Inc},
    url = {https://dl.acm.org/doi/10.1145/3523227.3546775},
    isbn = {9781450392785},
    doi = {10.1145/3523227.3546775}
}

@article{Allison1978MeasuresInequality,
    title = {{Measures of Inequality}},
    year = {1978},
    journal = {American Sociological Review},
    author = {Allison, Paul D},
    number = {6},
    month = {12},
    pages = {865--880},
    volume = {43},
    publisher = {American Sociological Review}
}

@inproceedings{Raj2022MeasuringResults,
    title = {{Measuring Fairness in Ranked Results}},
    year = {2022},
    booktitle = {Proceedings of the 45th International ACM SIGIR Conference on Research and Development in Information Retrieval},
    author = {Raj, Amifa and Ekstrand, Michael D.},
    month = {7},
    pages = {726--736},
    publisher = {ACM},
    url = {https://dl.acm.org/doi/10.1145/3477495.3532018},
    address = {New York, NY, USA},
    isbn = {9781450387323},
    doi = {10.1145/3477495.3532018}
}

@inproceedings{He2017NeuralFiltering,
    title = {{Neural collaborative filtering}},
    year = {2017},
    booktitle = {26th International World Wide Web Conference, WWW 2017},
    author = {He, Xiangnan and Liao, Lizi and Zhang, Hanwang and Nie, Liqiang and Hu, Xia and Chua, Tat Seng},
    pages = {173--182},
    publisher = {International World Wide Web Conferences Steering Committee},
    url = {http://dx.doi.org/10.1145/3038912.3052569},
    isbn = {9781450349130},
    doi = {10.1145/3038912.3052569},
    arxivId = {1708.05031}
}

@inproceedings{Wang2019NeuralFiltering,
    title = {{Neural graph collaborative filtering}},
    year = {2019},
    booktitle = {SIGIR 2019 - Proceedings of the 42nd International ACM SIGIR Conference on Research and Development in Information Retrieval},
    author = {Wang, Xiang and He, Xiangnan and Wang, Meng and Feng, Fuli and Chua, Tat Seng},
    month = {7},
    pages = {165--174},
    publisher = {Association for Computing Machinery, Inc},
    url = {https://dl.acm.org/doi/10.1145/3331184.3331267},
    isbn = {9781450361729},
    doi = {10.1145/3331184.3331267},
    arxivId = {1905.08108}
}

@inproceedings{Krichene2020OnRecommendation,
    title = {{On Sampled Metrics for Item Recommendation}},
    year = {2020},
    booktitle = {Proceedings of the ACM SIGKDD International Conference on Knowledge Discovery and Data Mining},
    author = {Krichene, Walid and Rendle, Steffen},
    month = {8},
    pages = {1748--1757},
    publisher = {Association for Computing Machinery},
    url = {https://doi.org/10.1145/3394486.3403226},
    isbn = {9781450379984},
    doi = {10.1145/3394486.3403226}
}

@inproceedings{Do2022OptimizingRankings,
    title = {{Optimizing Generalized Gini Indices for Fairness in Rankings}},
    year = {2022},
    booktitle = {Proceedings of the 45th International ACM SIGIR Conference on Research and Development in Information Retrieval},
    author = {Do, Virginie and Usunier, Nicolas},
    pages = {737--747},
    volume = {1},
    publisher = {ACM},
    url = {https://doi.org/10.1145/3477495.3532035},
    address = {New York, NY, USA},
    isbn = {9781450387323},
    doi = {10.1145/3477495}
}

@inproceedings{Zhu2021Popularity-OpportunityFiltering,
    title = {{Popularity-Opportunity Bias in Collaborative Filtering}},
    year = {2021},
    booktitle = {WSDM 2021 - Proceedings of the 14th ACM International Conference on Web Search and Data Mining},
    author = {Zhu, Ziwei and He, Yun and Zhao, Xing and Zhang, Yin and Wang, Jianling and Caverlee, James},
    month = {8},
    pages = {85--93},
    publisher = {Association for Computing Machinery, Inc},
    url = {https://doi.org/10.1145/3437963.3441820},
    isbn = {9781450382977},
    doi = {10.1145/3437963.3441820}
}

@inproceedings{Webber2008Precision-at-tenRedundant,
    title = {{Precision-at-ten considered redundant}},
    year = {2008},
    booktitle = {ACM SIGIR 2008 - 31st Annual International ACM SIGIR Conference on Research and Development in Information Retrieval, Proceedings},
    author = {Webber, William and Moffat, Alistair and Zobel, Justin and Sakai, Tetsuya},
    pages = {695--696},
    url = {https://dl.acm.org/doi/10.1145/1390334.1390456},
    address = {Singapore},
    isbn = {9781605581644},
    doi = {10.1145/1390334.1390456}
}

@inproceedings{Wang2022ProvidingSystems,
    title = {{Providing Item-side Individual Fairness for Deep Recommender Systems}},
    year = {2022},
    booktitle = {ACM International Conference Proceeding Series},
    author = {Wang, Xiuling and Wang, Wendy Hui},
    month = {6},
    pages = {117--127},
    volume = {22},
    publisher = {Association for Computing Machinery},
    url = {https://doi.org/10.1145/3531146.3533079},
    isbn = {9781450393522},
    doi = {10.1145/3531146.3533079}
}

@article{Moffat2008Rank-biasedEffectiveness,
    title = {{Rank-biased precision for measurement of retrieval effectiveness}},
    year = {2008},
    journal = {ACM Transactions on Information Systems},
    author = {Moffat, Alistair and Zobel, Justin},
    number = {1},
    month = {12},
    volume = {27},
    publisher = {ACM},
    url = {https://dl.acm.org/doi/10.1145/1416950.1416952},
    doi = {10.1145/1416950.1416952},
    issn = {10468188}
}

@inproceedings{Zhao2021RecBole:Algorithms,
    title = {{RecBole: Towards a Unified, Comprehensive and Efficient Framework for Recommendation Algorithms}},
    year = {2021},
    booktitle = {International Conference on Information and Knowledge Management, Proceedings},
    author = {Zhao, Wayne Xin and Mu, Shanlei and Hou, Yupeng and Lin, Zihan and Chen, Yushuo and Pan, Xingyu and Li, Kaiyuan and Lu, Yujie and Wang, Hui and Tian, Changxin and Min, Yingqian and Feng, Zhichao and Fan, Xinyan and Chen, Xu and Wang, Pengfei and Ji, Wendi and Li, Yaliang and Wang, Xiaoling and Wen, Ji Rong},
    pages = {4653--4664},
    publisher = {ACM},
    url = {https://doi.org/10.1145/3459637.3482016},
    address = {New York, NY, USA},
    isbn = {9781450384469},
    doi = {10.1145/3459637.3482016},
    arxivId = {2011.01731}
}

@inproceedings{Moffat2013SevenMetrics,
    title = {{Seven Numeric Properties of Effectiveness Metrics}},
    year = {2013},
    booktitle = {Information Retrieval Technology},
    author = {Moffat, Alistair},
    editor = {Banchs, Rafael E and Silvestri, Fabrizio and Liu, Tie-Yan and Zhang, Min and Gao, Sheng and Lang, Jun},
    pages = {1--12},
    publisher = {Springer Berlin Heidelberg},
    address = {Berlin, Heidelberg},
    isbn = {978-3-642-45068-6}
}

@inproceedings{Ning2011SLIM:Systems,
    title = {{SLIM: Sparse LInear Methods for top-N recommender systems}},
    year = {2011},
    booktitle = {Proceedings - IEEE International Conference on Data Mining, ICDM},
    author = {Ning, Xia and Karypis, George},
    pages = {497--506},
    url = {https://dl.acm.org/doi/10.1109/ICDM.2011.134},
    isbn = {9780769544083},
    doi = {10.1109/ICDM.2011.134},
    issn = {15504786}
}

@inproceedings{Schedl2016TheRecommendation,
    title = {{The LFM-1b dataset for music retrieval and recommendation}},
    year = {2016},
    booktitle = {ICMR 2016 - Proceedings of the 2016 ACM International Conference on Multimedia Retrieval},
    author = {Schedl, Markus},
    month = {6},
    pages = {103--110},
    publisher = {Association for Computing Machinery, Inc},
    url = {https://dl.acm.org/doi/10.1145/2911996.2912004},
    isbn = {9781450343596},
    doi = {10.1145/2911996.2912004}
}

@article{Harper2015TheContext,
    title = {{The MovieLens datasets: History and context}},
    year = {2015},
    journal = {ACM Trans. Interact. Intell. Syst. 5, 4, Article},
    author = {Harper, F Maxwell and Konstan, Joseph A},
    volume = {19},
    url = {http://dx.doi.org/10.1145/2827872},
    doi = {10.1145/2827872}
}

@article{Ceriani2012TheGini,
    title = {{The origins of the Gini index: extracts from Variabilit{\`{a}} e Mutabilit{\`{a}} (1912) by Corrado Gini}},
    year = {2012},
    journal = {J Econ Inequal},
    author = {Ceriani, Lidia and Verme, Paolo},
    pages = {421--443},
    volume = {10},
    url = {http://www.umass.edu/wsp/statistics/tales/gini.html},
    doi = {10.1007/s10888-011-9188-x}
}

@inproceedings{Do2021Two-sidedDominance,
    title = {{Two-sided fairness in rankings via Lorenz dominance}},
    year = {2021},
    booktitle = {Advances in Neural Information Processing Systems},
    author = {Do, Virginie and Corbett-Davies, Sam and Atif, Jamal and Usunier, Nicolas},
    pages = {8596--8608},
    volume = {34}
}

@article{Gini1912VariabilitaMutabilita,
    title = {{Variabilit{\`{a}} e mutabilit{\`{a}}}},
    year = {1912},
    journal = {Contributo allo Studio delle Distribuzioni e delle Relazioni Statistiche},
    author = {Gini, C.}
}

@inproceedings{Liang2018VariationalFiltering,
    title = {{Variational autoencoders for collaborative filtering}},
    year = {2018},
    booktitle = {The Web Conference 2018 - Proceedings of the World Wide Web Conference, WWW 2018},
    author = {Liang, Dawen and Krishnan, Rahul G. and Hoffman, Matthew D. and Jebara, Tony},
    month = {4},
    pages = {689--698},
    volume = {10},
    publisher = {Association for Computing Machinery, Inc},
    url = {https://doi.org/10.1145/3178876.3186150},
    isbn = {9781450356398},
    doi = {10.1145/3178876.3186150},
    arxivId = {1802.05814}
}
